\RequirePackage{fix-cm}
\documentclass[twocolumn]{svjour3}          
\smartqed  
\usepackage{graphicx}
\usepackage{upgreek}
\usepackage{thmtools,thm-restate}
\usepackage{tikz}
\usepackage{pgfplots}
\usepackage{pgfkeys}
\usepackage{enumitem}   
\usepackage[misc]{ifsym}
\usepackage{footmisc}
\usepackage{subcaption}
\usepackage{amsfonts}
\usepackage[ruled,lined,linesnumbered]{algorithm2e} 
\usepackage{algorithmic}
\usepackage{listings}
\usepackage{etoolbox}

\newcommand{\mec}{\mathcal{M}}
\newcommand{\mecld}{M_{LD}}
\newcommand{\mecsld}{M_{SLD}}
\newcommand{\mecsldbar}{\overline{M}_{SLD}}
\newcommand{\mecexp}{M_{EM}}
\newcommand{\mecpf}{M_{PF}}
\newcommand{\domain}{\mathcal{D}^n}
\newcommand{\error}{\mathcal{E}}
\newcommand{\range}{\mathcal{R}}
\newcommand{\real}{\mathbb{R}}
\newcommand{\expected}{\mathbb{E}}
\newcommand{\deltabar}{\bar{\delta}}
\newcommand{\deltau}{\delta^u}

\newcommand{\AAA}{\mathcal{A}}
\newcommand{\TT}{\mathcal{T}}
\newcommand{\ttt}{\tau}
\newcommand{\tj}{\ttt_j^A}
\newcommand{\tjc}{\ttt_{j,c}^A}
\newcommand{\tjx}[1]{\ttt_j^{A,#1}}
\newcommand{\tjcx}[1]{\ttt_{j,c}^{A,#1}}

\newcommand{\umed}{u_{med}}

\makeatletter
\newcommand{\vast}{\bBigg@{3}}
\newcommand{\vvast}{\bBigg@{4}}
\newcommand{\Vast}{\bBigg@{5}}

\usepackage[hidelinks]{hyperref}

\usepackage{mathtools}
\DeclarePairedDelimiter{\ceil}{\lceil}{\rceil}

\allowdisplaybreaks

\makeatletter
\tikzset{
    database/.style={
        path picture={
            \draw (0, 1.5*\database@segmentheight) circle [x radius=\database@radius,y radius=\database@aspectratio*\database@radius];
            \draw (-\database@radius, 0.5*\database@segmentheight) arc [start angle=180,end angle=360,x radius=\database@radius, y radius=\database@aspectratio*\database@radius];
            \draw (-\database@radius,-0.5*\database@segmentheight) arc [start angle=180,end angle=360,x radius=\database@radius, y radius=\database@aspectratio*\database@radius];
            \draw (-\database@radius,1.5*\database@segmentheight) -- ++(0,-3*\database@segmentheight) arc [start angle=180,end angle=360,x radius=\database@radius, y radius=\database@aspectratio*\database@radius] -- ++(0,3*\database@segmentheight);
        },
        minimum width=2*\database@radius + \pgflinewidth,
        minimum height=3*\database@segmentheight + 2*\database@aspectratio*\database@radius + \pgflinewidth,
    },
    database segment height/.store in=\database@segmentheight,
    database radius/.store in=\database@radius,
    database aspect ratio/.store in=\database@aspectratio,
    database segment height=0.1cm,
    database radius=0.25cm,
    database aspect ratio=0.35,
}

\makeatother

\newenvironment{customlegend}[1][]{
\begingroup
\csname pgfplots@init@cleared@structures\endcsname
\pgfplotsset{#1}
}{%
\csname pgfplots@createlegend\endcsname
\endgroup
}%
\def\addlegendimage{\csname pgfplots@addlegendimage\endcsname}

\tikzstyle{vertex}=[circle,minimum size=15pt,inner sep=0pt, draw=black]
\tikzstyle{edge} = [draw,-]
\tikzstyle{dashed edge} = [draw,-,densely dashed]

\providetoggle{vldb}
\settoggle{vldb}{false}

\journalname{VLDB Journal}

\begin{document}

\iftoggle{vldb}{
  {\bfseries\LARGE Cover letter \par}	
\setcounter{page}{0}


We would like to thank the Editors-in-Chief of PVLDB Volume 14 for the invitation to submit an extended version of our PVLDB paper ``Local Dampening: Differential Privacy for Non-Numeric Queries via Local Sensitivity'' to the special issue of the VLDB Journal for the best papers from PVLDB volume 14. Please find our paper attached to this letter.

We added significant new technical content/contributions to the PVLDB paper, extending it by $30\%$ as described below.

\begin{itemize}
    \item We coin the term \textit{sensitivity functions} as a generic notion to denote the local sensitivity notions presented in the paper and, also, functions that compute the upper bound on the local sensitivity. Additionally, we create a classification for the sensitivity functions on admissibility, boundedness, monotonicity and stability to support the privacy guarantee analysis and the theoretical accuracy analysis (Section \ref{sec:sensitivity_functions}).
    \item We provide a new theoretical and empirical accuracy analysis where we provide tools to compare two instances of the local dampening mechanism. We show that the exponential mechanism is an instance of the local dampening mechanism and we prove that, for stable sensitivity functions, the exponential mechanism is the worst instance of the local dampening mechanism in terms of accuracy. Moreover, we discuss that the local dampening mechanism can still be accurate for non-stable sensitivity functions (Section \ref{section:accuracy_analysis}).
    \item We tackle a new application, the percentile selection problem, where the task is to release the label of the $p$-th percentile element. Empirical results show that the local dampening mechanism can improve up to $73\%$ on global sensitivity based approaches (Section \ref{section:percentile_selection}).
    \item We add a new related work, the permute-and-flip mechanism, that is also a non-numeric mechanism like the exponential mechanism. We carry out a new experimental comparison of the permute-and-flip mechanism to the local dampening mechanism and the exponential mechanism on the percentile selection problem and on the node influential analysis problem. In the percentile selection problem, the permute-and-flip mechanism had about the same accuracy as the exponential mechanism, the local dampening mechanism could improve by up to $73\%$ on the permute-and-flip. In the node influential node analysis, the permute-and-flip could reduce the use of privacy budget by one order of magnitude compared to the exponential. Nevertheless, the local dampening reduce the use of privacy budget by 2 to 3 orders of magnitude. (Sections \ref{section:percentile_selection} and \ref{section:influential_node_analysis}). 
    \item We rewrote several parts of the paper for better understating, redesigned the paper structure and enhanced the examples.
\end{itemize}

We would be glad to consider the comments/suggestions of the reviewers. 

Please let us know if you need any additional information from us.

Best regards,

}

\sloppy

\title{Local Dampening: Differential Privacy for Non-numeric Queries via Local Sensitivity
}


\author{Victor A. E. Farias  \and Felipe T. Brito \and Cheryl Flynn \and Javam C. Machado \and Subhabrata Majumdar \and Divesh Srivastava.
}




\institute{Victor A. E. Farias  \at
        Universidade Federal do Ceará, Fortaleza, Brazil  \\
              \email{victor.farias@lsbd.ufc.br}           
           \and
           Felipe T. Brito \at
           Universidade Federal do Ceará, Fortaleza, Brazil \\ 
           \email{felipe.timbo@lsbd.ufc.br}         
           \and
           Cheryl Flynn Brooks \at
           AT\&T Chief Data Office, Bedminster, NJ, USA \\ 
           \email{cflynn@research.att.com}
           \and
           Javam C. Machado \at
           Universidade Federal do Ceará, Fortaleza, Brazil \\ 
           \email{javam.machado@lsbd.ufc.br}         
          \and
          Subhabrata Majumdar \at
           Splunk, Seattle, WA, USA \\ 
           \email{majumdar@splunk.com}         
          \and
          Divesh Srivastava \at
           AT\&T Chief Data Office, Bedminster, NJ, USA \\ 
           \email{divesh@research.att.com}
}

\date{Received: date / Accepted: date}

\maketitle

\makeatletter



\begin{abstract}

  \textit{Differential privacy} is the state-of-the-art formal definition for data release under strong privacy guarantees. A variety of mechanisms have been proposed in the literature for releasing the output of numeric queries (e.g., the Laplace mechanism and smooth sensitivity mechanism). Those mechanisms guarantee differential privacy by adding noise to the true query's output. The amount of noise added is calibrated by the notions of global sensitivity and local sensitivity of the query that measure the impact of the addition or removal of an individual on the query's output. Mechanisms that use local sensitivity add less noise and, consequently, have a more accurate answer. However, although there has been some work on generic mechanisms for releasing the output of non-numeric queries using global sensitivity (e.g., the Exponential mechanism), the literature lacks generic mechanisms for releasing the output of non-numeric queries using local sensitivity to reduce the noise in the query's output. In this work, we remedy this shortcoming and present the \textit{local dampening mechanism}. We adapt the notion of local sensitivity for the non-numeric setting and leverage it to design a generic non-numeric mechanism. We provide theoretical comparisons to the exponential mechanism and show under which conditions the local dampening mechanism is more accurate than the exponential mechanism. We illustrate the effectiveness of the local dampening mechanism by applying it to three diverse problems: (i) percentile selection problem. We report the $p$-th element in the database; (ii) Influential node analysis. Given an influence metric, we release the top-k most influential nodes while preserving the privacy of the relationship between nodes in the network; (iii) Decision tree induction. We provide a private adaptation to the ID3 algorithm to build decision trees from a given tabular dataset. Experimental evaluation shows that we can reduce the error for percentile selection application up to $73\%$, reduce the use of privacy budget by $2$ to $4$ orders of magnitude for influential node analysis application and increase accuracy up to $12\%$ for decision tree induction when compared to global sensitivity based approaches.
\keywords{Differential Privacy \and Decision Tree Induction \and Graph Analysis \and Local Sensitivity \and Percentile Selection}
\end{abstract}

\section{Introduction}
\label{sec:introduction}
\thispagestyle{empty}

Recent regulations on data privacy, such as \textit{General Data Protection Regulation (GPDR)} \cite{EUdataregulations2018} and \textit{Lei Geral de Proteção de Dados Pessoais (LGPD)} \cite{lgpd}, pose strict privacy requirements when gathering, storing and sharing data. Specifically, they require that an individual's information be rendered anonymous so that the individual is no longer identifiable from the published information. 

\textit{Differential privacy} \cite{dwork2011differential,dwork2006calibrating} is the state-of-the-art formal definition for data release under strong privacy guarantees. It imposes near-indistinguishability on the released information whether an individual belongs to a sensitive database or not. The key intuition is that the output distribution of a differentially private query should not change significantly based on the presence or absence of an individual. It provides statistical guarantees against the inference of private information through the use of auxiliary information. 

Algorithms can achieve differential privacy by employing output perturbation, which releases the true output of a given non-private query $f$ with noise added. The magnitude of the noise should be large enough to cover the identity of the individuals in the input database $x$. 

For a numeric query (i.e., query with numeric output) $f$, the \textit{Laplace mechanism} \cite{dwork2006calibrating} is a well-known output perturbing private method. It adds numeric noise to the output of $f$ and calibrates the noise based only on $f$ and not on $x$. The noise magnitude is proportional to the concept of \textit{global sensitivity}, which measures the worst-case impact on $f$'s output of the addition or removal of an individual over the set of possible input databases. This may result in an unreasonably high amount of noise when $x$ is far from the database with the worst-case impact, which is the case for many realistic databases. To remedy this, Nissim et al. \cite{nissim2007smooth} proposed to add instance-based noise calibrated as a function of $x$. They introduced the notion of \textit{local sensitivity}, which quantifies the impact of addition or removal of an individual for the database instance $x$, resulting in a lower bound to the global sensitivity. Many works use this notion to shrink the amount of noise added to release more useful statistical information \cite{blocki2013differentially,karwa2011private,kasiviswanathan2013analyzing,lu2014exponential,zhang2015private}.

For the class of non-numeric queries $f$, i.e. $f$ has a non-numeric range $\mathcal{R}$, the \textit{exponential mechanism} \cite{mcsherry2007mechanism} ensures differential privacy by sampling elements from $\mathcal{R}$ using the exponential distribution. This requires a utility function $u(x,r)$ that takes as input a database $x$ and an element $r \in \mathcal{R}$ and outputs a numeric score that measures the utility of $r$. The larger $u(x,r)$, the higher the probability of the exponential mechanism outputting $r$. The exponential mechanism uses a similar notion of global sensitivity to that found in \cite{dwork2006calibrating} where it measures the worst-case impact on the utility  $u(x,r)$ for all elements $r \in \mathcal{R}$ by adding or removing an individual from all databases. However, to the best of our knowledge, the literature lacks generic mechanisms that apply local sensitivity to the non-numeric setting. Example \ref{example:ebc_global} introduces the running example of this paper and discusses the computation of global sensitivity. 

\begin{example}
    Consider an application where, given a graph $G$, the query should report the node with the largest \textit{Egocentric Betweenness Centrality} (EBC) \cite{freeman1978centrality,marsden2002egocentric,everett2005ego}. The EBC metric measures the degree to which nodes stand between each other, defined as 
    $$
        EBC(c) = \sum_{u,v \in N_c | u \neq v}  \frac{g_{uv}(c)}{g_{uv}},
    $$    
    where $g_{uv}$ is the number of geodesic paths connecting $u \neq c$ and $v \neq c$ on the subgraph composed by $c$ and its neighbors $N_c$ and $g_{uv}(c)$ is the number of those paths that include $c$.    
    
    For instance, let $G$ be the graph illustrated in Figure \ref{fig:grapha}. Node $a$ has EBC score equal to $7.5$ since there are ${{6}\choose{2}}=15$ pairs of neighbors of the form $\{v_i,v_j\}$, for $0 \leq i < j \leq 5$, that each contributes with $0.5$ as they have two geodesic paths of length $2$ from $v_i$ to $v_j$, where only one contains $a$. Pairs of the form $\{ b,v_i \}$, for $0 \leq i \leq 5$ do not contribute to the score of $a$ since there is only one geodesic path (length $1$) from $b$ to $v_i$ and it does not contain $a$. Verify that for the graph illustrated in \ref{fig:graphb}, node $a$ has EBC score equal to $6.5$.

\tikzstyle{vertex}=[circle,minimum size=15pt,inner sep=0pt, draw=black]
\tikzstyle{edge} = [draw,-]
\tikzstyle{dashed edge} = [draw,-,densely dashed]

\begin{figure}[ht]
    
    \begin{subfigure}[b]{0.23\textwidth}
        \centering
        \begin{tikzpicture}[]
            \foreach \pos/\name in {{(1.11,0)/a}, {(2.22,0)/b}, {(3.33,-1.5)/v_5}, {(2.67,-1.5)/v_4}, {(2,-1.5)/v_3}, {(1.34,-1.5)/v_2}, {(0.67,-1.5)/v_1}, {(0,-1.5)/v_0}}
            \node[vertex] (\name) at \pos {$\name$};
            
            \foreach \source/ \dest in {a/v_0, a/v_1, a/v_2, a/v_3, a/v_4, a/v_5, b/v_0, b/v_1, b/v_2, b/v_3, b/v_4, b/v_5}
            \path[edge] (\source) -- (\dest);    
            
            \path[dashed edge] (a) -- (b);
            
        \end{tikzpicture}
        \caption{EBC: Worst Case}
        \label{fig:grapha}
    \end{subfigure}
    \begin{subfigure}[b]{0.23\textwidth}
        \centering
        \begin{tikzpicture}[]
            \foreach \pos/\name in {{(1.11,0)/a}, {(2.22,0)/b}, {(3.33,-1.5)/v_5}, {(2.67,-1.5)/v_4}, {(2,-1.5)/v_3}, {(1.34,-1.5)/v_2}, {(0.67,-1.5)/v_1}, {(0,-1.5)/v_0}}
            \node[vertex] (\name) at \pos {$\name$};
            
            \foreach \source/ \dest in {a/b, a/v_1, a/v_2, a/v_3, b/v_2, b/v_3, b/v_4, b/v_5, v_0/v_1, v_4/v_5}
            \path[edge] (\source) -- (\dest);    
            
            \path[dashed edge] (a) -- (v_0);
        \end{tikzpicture}
        
        \caption{EBC: Usual case}
        \label{fig:graphb}
    \end{subfigure}
    
    \caption{Sensitivity of EBC}
    \label{example:ebc_global}
\end{figure}
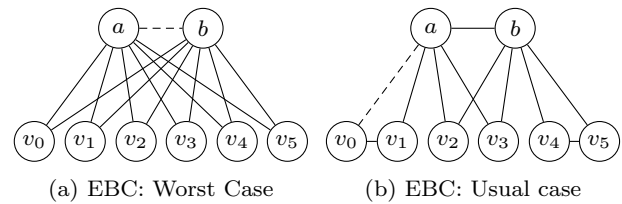

We use edge differential privacy where the sensitivity is measured by adding or removing one edge from the input graph. The global sensitivity for EBC is obtained from a graph of the form displayed in Figure \ref{fig:grapha}. Removing the edge $(a,b)$ is the worst case. The new EBC score of $a$ is $15$, $1$ point for each one of the $15$ pairs $\{ v_i, v_j \}$, $0 \leq i \leq 5$, as now there is only one geodesic path from $v_i$ to $v_j$ which includes $a$ (path $<v_i, a ,v_j>$). The paths of the form $<v_i, b, v_j>$ are not counted since $b$ no longer belongs to $N_a$.

Note this gadget that sets the global sensitivity is formed by two nodes with high degree that share all neighbors and those neighbors do not have an edge to each other. Thus, this gadget is very unlikely to be found in real graphs. Figure \ref{fig:graphb} illustrates a more typical example. In this instance, the worst local measurement of the sensitivity is given by the removal of the edge $(a,v_0)$ that shrinks the EBC score of $a$ by only $3$ ($1$ for each pair $\{v_0, b\}$, $\{v_0, v_2\}$ and $\{v_0, v_3\}$ since $v_0$ is no longer a neighbor of $a$).

\label{example:example1}
\end{example}

Additionally, the global sensitivity of $EBC$ grows quadratically with respect to the maximum degree as we discuss in Section \ref{section:influential_node_analysis}. Hence a non-numeric mechanism applying local sensitivity could add less noise to the output compared to a global sensitivity based approach.

In this paper, we propose the \textit{local dampening mechanism}, which adapts the notion of local sensitivity to the non-numeric setting and uses it to dampen the utility function $u$ in order to increase the signal-to-noise ratio. Local dampening also employs the exponential distribution as the exponential mechanism \cite{mcsherry2007mechanism}. Applications in which local sensitivity is significantly smaller than global sensitivity can benefit from our approach. For the scenario where local sensitivity is near the global sensitivity, the local dampening mechanism reverts to the exponential mechanism, so that the maximum addition of noise across all values of $r$ is bounded by the exponential mechanism.

To this end, we present a new version of the local sensitivity, called {\it element local sensitivity}. Traditional local sensitivity measures the largest impact of the addition or deletion of an individual to the input database over all outputs $r \in \mathcal{R}$. Element local sensitivity computes this impact, but only for some given element $r \in \mathcal{R}$. This allows us to explore local measurements of the sensitivity of $f$ even if traditional local sensitivity is near the global sensitivity, but, for most elements in $\mathcal{R}$, the element local sensitivity is low. 

\begin{example}
Consider the graph in Figure \ref{fig:grapha}. The removal of edge $(a,b)$ sets the traditional local sensitivity to $7.5$ which is also the case for global sensitivity. But measurements of sensitivity per node (element) are much smaller. For instance, the sensitivity for a node $v_i$ ($0 \leq i \leq 5$) is $1$ which is set by the removal of edge $(a,b)$ where $EBC(v_i)$ increases from $0$ to $1$ (path $<a,v_i,b>$). We explore this to improve the accuracy further.
\end{example} 

We illustrate the effectiveness of the local dampening mechanism by applying it to three diverse problems: (i) Percentile selection problem, where the task is to release the label of the $p$-th percentile element; (ii) Influential Node analysis, which searches for central nodes in a graph database. Given a centrality/influence metric, we release the label of the top-k most central nodes while preserving the privacy of the relationships between nodes in the graph; (iii) We also provide an application on tabular data that is a private adaptation to the ID3 algorithm to build a decision tree from a given tabular dataset based on the information gain for each attribute.

The contributions of this work are summarized as follows:
\begin{itemize}
    \item We adapt the local sensitivity definition to the non-numeric setting (Section \ref{section:differential_privacy}) and we introduce a new definition of local sensitivity that measures sensitivity per element (Section \ref{section:local_dampening}).
    \item We introduce the Local Dampening mechanism, a novel differentially private mechanism to answer non-numeric queries that applies local sensitivity to attenuate the utility function to increase the signal-to-sensitivity ratio to reduce noise (Section \ref{section:local_dampening}).
    \item We present a second version of our approach which we call the Shifted Local Dampening mechanism, which can effectively use the element local sensitivity to improve accuracy (Section \ref{section:shifted_local_dampening}).
    \item We develop a theoretical and empirical accuracy analysis where we enumerate some conditions under which the local dampening mechanism benefits from the local sensitivity notions. Under those conditions, we show that the exponential mechanism is an instance of the local dampening mechanism, and it is the worst instance of the local dampening mechanism in terms of accuracy. Also, we discuss the scenario where those conditions are not met and how we can still have good accuracy (Section \ref{section:accuracy_analysis}).
    \item We tackle the Percentile Selection problem where a private mechanism should report the label of the $p$-th percentile element. Empirical results show that the local dampening mechanism can improve up to $73\%$ over global sensitivity based approaches (Section \ref{section:percentile_selection}).
    \item We apply the local dampening mechanism to construct differentially private algorithms for a graph problem called Influential Node Analysis using Egocentric Betweenness Centrality as the influence metric, and we show how to compute local sensitivity for this application. Experimental results show that our approach could be as accurate as global sensitivity-based mechanisms using $2$ to $4$ orders of magnitude less privacy budget than global sensitivity-based approaches (Section \ref{section:influential_node_analysis}). 
    \item We address the application of building private algorithms for decision tree induction as an example data-mining application for tabular data. We present a differentially private adaptation of the entropy-based ID3 algorithm using the local dampening mechanism, and we provide a way to compute the local sensitivity efficiently. We can improve accuracy up to $12\%$ compared to previous works (Section \ref{section:decision_tree_induction}).
\end{itemize}
Section \ref{section:differential_privacy} presents the relevant formal definitions for Differential Privacy. Related work is presented in Section \ref{sec:related_work}, \ref{sec:related_work_influential} and \ref{section:related_work_decision_tree_induction}. Finally, Section 8 concludes the paper. \iftoggle{vldb}{We defer the proofs of the lemmas and theorems to our technical report \cite{ourtechnicalreport}, except Theorem \ref{theorem:priv_local_dampening} where its proof is provided in this paper.}{We defer all the proofs to the appendix, except Theorem \ref{theorem:priv_local_dampening} where its proof is provided in this paper.}


\section{Differential Privacy}
\label{section:differential_privacy}

Let $x$ be a sensitive database and $f$ a function to be evaluated on $x$. The database is represented as vector $x \in \mathcal{D}^n$ where each entry represents an individual tuple. The output $f(x)$ must be released without leaking much information about the individuals. For that, we need to design a randomized algorithm $\mec(x)$ that adds noise to $f(x)$ such that it satisfies the definition of differential privacy stated below.

\begin{definition}
    ($\epsilon$-Differential Privacy \cite{dwork2006our,dwork2006calibrating}). A randomized algorithm $\mathcal{M}$ satisfies $\epsilon$-differential privacy, if for any
    two databases $x$ and $y$ satisfying $d(x,y) \leq 1$ and for any
    possible output $O$ of $\mathcal{M}$, we have
    $$Pr [\mec(x) = O] \leq \exp(\epsilon) Pr [\mec(y) = O],$$
    where $Pr[\cdot]$ denotes the probability of an event and $d$ denotes the hamming distance between the two databases, i.e. the number of tuples of individuals that changed value, i.e., $d(x,y)=|\{ i \ | \ x_i \neq y_i \}|$. We refer to $d$ as the distance between two given databases.
    \label{def:dp}
\end{definition}

The parameter $\epsilon$ dictates how close the distribution of the outputs differs between the databases $x$ and $y$. Small values of $\epsilon$ means that those two distributions must be really close which hurts accuracy but provides a better indistinguishability, i.e., a better privacy level. For large $\epsilon$, the opposite happens, the two distributions can differ more which means a better accuracy and a lower level of privacy. The data querier may want to submit multiple queries to the database. The parameter $\epsilon$ is also called as the \textit{privacy budget} since each query submitted to the database consumes a part of the budget (refer to Section \ref{sec:composition}).

\subsection{The Non-Numeric Setting}
\label{sec:non_numeric_setting}

We consider two settings when building a private mechanism: the numeric setting and the non-numeric setting. The numeric setting is when one needs to construct a private mechanism for a function $f: \domain \rightarrow \range$ where $f$ outputs a vector of numeric values, i.e., $\range = \real^d$. In this case, the Laplace Mechanism applies \cite{dwork2006calibrating}.

In this work, we address the non-numeric setting. We aim to build a private mechanism for a function $f: \domain \rightarrow \range$ where $f$ outputs non-numeric values, i.e., $\range$ denotes a discrete set of outputs $\range = \{ r_1, r_2, r_3, \dots \}$.

In the non-numeric setting, the data querier needs to provide a utility function $u: \mathcal{D}^n \times \mathcal{R} \rightarrow \mathbb{R}$ that takes a database $x$ and an output $r \in \mathcal{R}$ and produces an \textit{utility score} $u(x,r)$. The utility function is application-specific and each application requires its own utility function.

The utility score represents how good an output $r$ is for the dataset $x$. This means that, for a given input database $x$, the analyst prefers that the mechanism outputs the elements with high utility score. Thus a mechanism for answering $f$ needs to output a high utility output with higher probability.

\subsection{Sensitivity in the Non-Numeric Setting}
\label{sec:sensitivity_non_numeric_setting}

Differentially private mechanisms usually perturb the true output with noise. The amount of noise added to the true output of a non-numeric function $f: \domain \rightarrow \range$ is proportional to the \textit{sensitivity} of the utility function $u: \mathcal{D}^n \times \mathcal{R} \rightarrow \mathbb{R}$. The various notions of sensitivity used in this work are presented in this section.


\textbf{Global sensitivity}. The global sensitivity of $u$ is defined as the maximum possible difference of utility scores at all possible pairs of database inputs $x, y$ and all possible elements $r \in \range$:

\begin{definition}
    (Global Sensitivity $\Delta u$ \cite{mcsherry2007mechanism}).
    Given a utility function $u: \mathcal{D}^n \times \mathcal{R} \rightarrow \mathbb{R}$ that takes as input a database $x \in \domain$ and an element $r \in \range$ and outputs a numeric score for $r$ in $x$. The global sensitivity of $u$ is defined as:
    $$ \Delta u = \max_{r \in \mathcal{R}} \max_{x,y | d(x,y) \leq 1} |u(x,r) - u(y,r)|.$$ 
    \label{def:global_sensitivity}
\end{definition}

Intuitively, the global sensitivity measures the maximum change in the utility score over all $r \in \range$ between any two neighboring databases in the universe of all databases.


\textbf{Local Sensitivity}. In non-numeric mechanims, a larger global sensitivity $\Delta u$ implies a lower accuracy. Thus, when constructing private solutions, we seek mechanisms with low sensitivity. The concept of local sensitivity $LS(x)$ \cite{nissim2007smooth} captures the sensitivity locally on the input database $x$ instead of searching for the sensitivity in the universe of databases $\mathcal{D}^n$. 

Local sensitivity tends to be smaller than global sensitivity for many problems \cite{blocki2013differentially,karwa2011private,kasiviswanathan2013analyzing,lu2014exponential,nissim2007smooth,zhang2015private}. In those problems, the real-world databases are very different from the worst-case scenario of the global sensitivity and they have a low observed local sensitivity. We present an adapted version of the definition of local sensitivity for the non-numeric setting, as given in \cite{nissim2007smooth}.

\begin{definition}
    (Local Sensitivity, adapted from \cite{nissim2007smooth}). Given a utility function $u(x,r)$ that takes as input a database $x$ and an element $r$ and outputs a numeric score, the local sensitivity of $u$ is defined as
    $$ LS^{u}(x) = \max_{r \in \mathcal{R}} \max_{y| d(x,y) \leq 1}|u(x,r)-u(y,r)|$$
\end{definition}

Observe that the global sensitivity is the maximum local sensitivity over the set of all databases, $\Delta u = \max_{x} LS^u(x)$. 

However, using solely the local sensitivity to build a mechanism is not enough to satisfy differential privacy. Thus, like the smooth sensitivity framework \cite{nissim2007smooth}, a part of our solution is based on local sensitivity at distance $t$. We adapt the notion of local sensitivity at distance $t$ \cite{nissim2007smooth} to the non-numeric setting for use in this work:

\begin{definition}\label{def:locsen_dist_t}
	(Local Sensitivity at distance $t$, adapted from \cite{nissim2007smooth}). Given a utility function $u: \mathcal{D}^n \times \mathcal{R} \rightarrow \mathbb{R}$ that takes as input a database $x \in \domain$ and an element $r \in \range$ and outputs a numeric score for $r$ in $x$, the local sensitivity at distance $t$ of $u$ is defined as
	$$LS^{u}(x,t) = \max_{y| d(x,y) \leq t} LS^{u}(y) $$
\end{definition}

Local sensitivity at distance $t$, $LS^u(x,t)$, measures the maximum local sensitivity $LS^u(y)$ over all databases $y$ at maximum distance $t$, i.e., we allow $t$ modifications on the database before computing its local sensitivity. Note that $LS^{u}(x,0) = LS^{u}(x)$. An example of the local sensitivity at distance $1$ is given in Example \ref{example:example_local_sensitivity_dist_1}.

\begin{example}

(Local sensitivity at distance $1$) Consider the graph $G$ of Figure \ref{fig:graphc}. The local sensitivity at distance $t$ allows $t$ extra modifications before measuring local sensitivity. As discussed in Example \ref{example:example1}, the local sensitivity of $G$ is $3$ (at distance $0$): $LS^{EBC}(G,0)=3$.

To compute local sensitivity at distance $1$, we need to find which edge to add or remove in order to compute the maximum local sensitivity at distance $1$. This case is found by removing edge $(a,v_0)$ as shown in Figure \ref{fig:graphd} obtaining $G'$. Then the local sensitivity of $G'$ is $5$ where node $b$ increases by $5$ units when adding edge $(b,v_0)$ ($1$ for each pair $\{v_0, v_2\}$, $\{v_0,v_3\}$, $\{v_0,v_4\}$, $\{v_0,v_5\}$ and $\{v_0,a\}$). This means that $LS^{EBC}(G,1)=5$

\tikzstyle{vertex}=[circle,minimum size=15pt,inner sep=0pt, draw=black]
\tikzstyle{edge} = [draw,-]
\tikzstyle{dashed edge} = [draw,-,densely dashed]

\begin{figure}[ht]

    \begin{subfigure}[b]{0.23\textwidth}
        \centering
        \begin{tikzpicture}[]
            \foreach \pos/\name in {{(1.11,0)/a}, {(2.22,0)/b}, {(3.33,-1.5)/v_5}, {(2.67,-1.5)/v_4}, {(2,-1.5)/v_3}, {(1.34,-1.5)/v_2}, {(0.67,-1.5)/v_1}, {(0,-1.5)/v_0}}
                \node[vertex] (\name) at \pos {$\name$};
    
            \foreach \source/ \dest in {a/b, a/v_0, a/v_1, a/v_2, a/v_3, b/v_2, b/v_3, b/v_4, b/v_5, v_0/v_1, v_4/v_5}
                \path[edge] (\source) -- (\dest);

        \end{tikzpicture}
        \caption{Original Graph $G$}
        \label{fig:graphc}
    \end{subfigure}
    \begin{subfigure}[b]{0.23\textwidth}
        \centering
        \begin{tikzpicture}[]
            \foreach \pos/\name in {{(1.11,0)/a}, {(2.22,0)/b}, {(3.33,-1.5)/v_5}, {(2.67,-1.5)/v_4}, {(2,-1.5)/v_3}, {(1.34,-1.5)/v_2}, {(0.67,-1.5)/v_1}, {(0,-1.5)/v_0}}
                \node[vertex] (\name) at \pos {$\name$};
    
            \foreach \source/ \dest in {a/b, a/v_1, a/v_2, a/v_3, b/v_2, b/v_3, b/v_4, b/v_5, v_0/v_1, v_4/v_5}
                \path[edge] (\source) -- (\dest);    
                
            \path[dashed edge] (v_0) -- (b);
        \end{tikzpicture}
        
    \caption{$G'$ at distance 1 from $G$}
    \label{fig:graphd}
    \end{subfigure}        

    \caption{Local Sensitivity at distance $1$}

\end{figure}
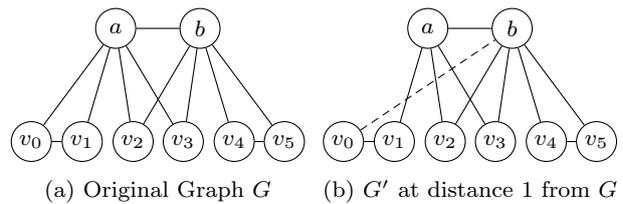
\label{example:example_local_sensitivity_dist_1}
\end{example}

\subsection{Composition}
\label{sec:composition}

The data querier can pose several queries to the database to compose complex differentially private algorithms. In this work, we focus on two types of composition: sequential composition and parallel composition. 

The sequential composition applies when a set of mechanisms is executed against a dataset. This implies that the privacy budget used on each computation sums up:

\begin{theorem}
    (Sequential composition \cite{mcsherry2007mechanism,mcsherry2009privacy}) Let $\mec_i : \domain \rightarrow \range_i$ be an $\epsilon_i$-differentially private algorithm for $i \in [k]$. Then $\mec(x)=(\mec_1(x),\cdots,\mec_k(x))$ is ($\sum_{i=1}^{k}$)-differentially private.
    \label{theorem:sequential_composition}
\end{theorem}


The parallel composition applies if queries are applied to disjoint subsets of the database. Suppose that the data querier issues many $\epsilon_i$-differentially private mechanisms over disjoint subsets of the database. It composes to a $\max_i \epsilon_i$-differentially private mechanism. This scenario has a lower privacy cost compared to the sequential composition.

\begin{theorem}
    (Parallel composition \cite{mcsherry2009privacy}) Let $\mec_i : \domain \rightarrow \range_i$ be an $\epsilon_i$-differentially private algorithm for $i \in [k]$ and $x_1,...,x_k$ be pairwise disjoint subsets of $\domain$. Then $\mec(x)=(\mec_1(x_1),\cdots,\mec_k(x_k))$ is ($\max_i \epsilon_i$)-differentially private.
    \label{theorem:parallel_composition}
\end{theorem}

\section{Local Dampening Mechanism}
\label{section:local_dampening}

We present the \textit{local dampening mechanism}, an output perturbing differentially private mechanism for the non-numeric setting that uses local sensitivity to reduce the noise injected to the true answer. Our approach uses the local sensitivity notions described in the Section \ref{sec:sensitivity_non_numeric_setting}. Additionally, we introduce a new notion of sensitivity called \textit{element local sensitivity}. It measures the worst impact on the sensitivity for a given element $r \in \mathcal{R}$ when adding or removing an individual from the input database $x$, i.e., the largest difference $|u(x,r)-u(y,r)|$ for all neighbors $y$ of $x$.

More broadly, we coin the notion of \textit{sensitivity function} that generalizes the all local sensitivity definitions. This concept is specially useful when computing the sensitivity is not possible or efficient but computing an upper bound is simpler, as it can be NP-hard \cite{nissim2007smooth,zhang2015private}.

The local dampening mechanism employs a sensitivity function to dampen the utility function $u$ and construct its dampened version, referred to $D_{u, \delta^u}$. Specifically, we attenuate $u$ such that the signal-to-sensitivity ratio (i.e. u/sensitivity) is larger which results in higher accuracy. A sensitivity function is a function that computes one of the notions of sensitivity or an upper bound on the sensitivity.

\subsection{Element Local Sensitivity}
\label{sec:els}

The local sensitivity at distance $t$, $LS^u(x, t)$, quantifies the maximum sensitivity of $u$ over all elements $r \in \mathcal{R}$ for an input database $x$ with $t$ modifications (Definition \ref{def:locsen_dist_t}). That gives a high-level description of the variation of $u$ in neighboring databases. However, if just one element in $\mathcal{R}$ has a high value of sensitivity (close to $\Delta u$), $LS^u(x, t)$ will be equally large. That is ineffective in a scenario where most of the elements have low sensitivity and just few have high sensitivity, which makes $LS^u(x, t)$ large and consequently hurts accuracy.

We introduce a more specialized definition of local sensitivity named element local sensitivity, denoted as $LS^u(x, t, r)$, which measures the sensitivity of $u$ for a given $r \in \mathcal{R}$ for an input database $x$ at distance $t$ (definition \ref{def:item_local}). This allows us to grasp the sensitivity of $u$ for a single element. 

\begin{definition}
	(Element Local Sensitivity at distance $t$). Given a utility function $u(x,r)$ that takes as input a database $x$ and an element $r$ and outputs a numeric score for $x$, the element local sensitivity at distance $t$ of $u$ is defined as 
	$$ LS^{u}(x,t,r) = \max_{y \in \domain| d(x,y) \leq t, z \in \domain | d(y,z) \leq 1}|u(y,r)-u(z,r)|,$$	
	where $d(x,y)$ denotes the distance between two databases.
	\label{def:item_local}
\end{definition}

Note that we can obtain $LS^u(x, t)$ from this definition: $ LS^{u}(x,t) = \max_{r \in \mathcal{R}} LS^{u}(y,t,r)$ as $LS^{u}(x,t,r) = \max_{y| d(x,y) \leq t} LS^{u}(y,0,r)$.


\begin{example}
(Element local sensitivity) We illustrate this definition with the same setup from previous examples. Let $G$ be the graph from Figure \ref{fig:graphc}. Suppose we want to compute the element local sensitivity for $v_4$, $LS^{u}(G,0,v_4)$. We measure only the worst impact of the addition or removal of an edge on the value of the EBC score for $v_4$. This is obtained by adding the edge $(v_0, v_4)$ (Figure \ref{fig:graphe}). The EBC score increases by $2$ ($1$ for each pair $\{b, v_0\}$ and $\{ v_0, v_5 \}$). Thus $LS^{u}(G,0,v_4) = 2$ which is smaller than local sensitivity $LS^{u}(G,0) = 3$  and $\Delta u = 7.5$ as discussed in Example \ref{example:example1}.

\tikzstyle{vertex}=[circle,minimum size=15pt,inner sep=0pt, draw=black]
\tikzstyle{edge} = [draw,-]
\tikzstyle{dashed edge} = [draw,-,densely dashed]

\begin{figure}[ht]

\centering
\caption{Element Local Sensitivity for $v_4$}

	\begin{tikzpicture}[]
		\foreach \pos/\name in {{(1.11,0)/a}, {(2.22,0)/b}, {(3.33,-1.5)/v_5}, {(2.67,-1.5)/v_4}, {(2,-1.5)/v_3}, {(1.34,-1.5)/v_2}, {(0.67,-1.5)/v_1}, {(0,-1.5)/v_0}}
			\node[vertex] (\name) at \pos {$\name$};

		\foreach \source/ \dest in {a/b, a/v_0, a/v_1, a/v_2, a/v_3, b/v_2, b/v_3, b/v_4, b/v_5, v_0/v_1, v_4/v_5}
			\path[edge] (\source) -- (\dest);    
			
			\draw [-, dashed] (v_4) to [bend left=40]  (v_0);

	\end{tikzpicture}

\label{fig:graphe}
\end{figure}
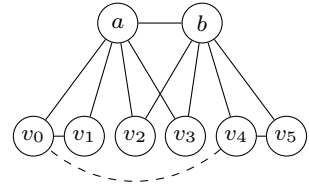

\label{example:example3}
\end{example}

\subsection{Sensitivity Functions}
\label{sec:sensitivity_functions}

Computing local sensitivity $LS^{u}(x,t)$ or element local sensitivity $LS^{u}(x,t,r)$ is not always feasible, as it can be NP-hard \cite{nissim2007smooth,zhang2015private}. To navigate this problem, we can relax the need for the computation of $LS^{u}(x,t)$ or $LS^{u}(x,t,r)$ and build a computationally efficient function $\delta^u(x,t,r)$ that computes an upper bound for $LS^{u}(x,t)$ or $LS^{u}(x,t,r)$ that is still smaller than $\Delta u$. We refer to  $\delta^u$ as a sensitivity function that has the following signature $\delta^u: \domain \times \mathbb{N}^0 \times \range \rightarrow \real$. Note that $\delta^u(x,t,r)=\Delta u$, $\delta^u(x,t,r)=LS^{u}(x,t)$ or $\delta^u(x,t,r)=LS^{u}(x,t,r)$ are sensitivity functions.

We define a classification of sensitivity functions based on four properties: admissibility, boundedness, monotonicity and stability.

\subsubsection{Admissibility}
\label{sec:admissibility}

The sensitivity function $\delta^u$ needs to have some properties to be admissible in the local dampening mechanism to guarantee differential privacy:

\begin{definition}
	(Admissibility). A sensitivity function $\delta^{u}(x, t, r)$ is \textit{admissible} if:
	\begin{enumerate}
		\item $\delta^{u}(x, 0, r) \geq LS^{u}(x, 0, r)$, for all $x \in \domain$ and all $r \in \range$
		\item $\delta^{u}(x, t+1, r) \geq \delta^{u}(y, t, r)$, for all $x,y$ such that $d(x,y) \leq 1$ and all $t \geq 0$		
	\end{enumerate}
	\label{def:admissible_function}
\end{definition}

The global sensitivity $\Delta u$ is admissible as $\Delta u \geq LS^{u}(x, 0, r)$, for all $x$ and a constant value would satisfy the second requirement of Definition \ref{def:admissible_function}. We also show that the function $LS^u(x, t, r)$ itself is admissible.

\begin{restatable}{lemma}{lemmalocaladmissible}
    \sloppy The element local sensitivity $LS^u(x, t, r)$ is admissible.
    \label{lemma:admissiblels}
\end{restatable}


A intermediate result shows (Lemma \ref{lemma:max_admissible}) that $LS^u(x,t)$ is an admissible function once $LS^u(x,t)=max_{r \in \mathcal{R}} LS^u(x,t,r)$ and $LS^u(x,t,r)$ is an admissible function (Theorem \ref{lemma:admissiblels}).

\begin{lemma}
	\sloppy Let $\delta_1(x,t,r), \ldots, \delta_p(x,t,r)$ be admissible functions. Then $\delta(x,t,r)$ defined as $\delta(x,t,r)=\max(\delta_1(x,t,r), \ldots, \delta_p(x,t,r))$ is an admissible sensitivity function.
	\label{lemma:max_admissible}
\end{lemma}

The proof of Lemma \ref{lemma:max_admissible} is immediate given by the admissibility of $\delta_1(x,t,r), \ldots, \delta_p(x,t,r)$.

\subsubsection{Boundedness}
\label{sec:bounded_functions}

Some sensitivity functions, such as $LS^u(x,t)$ and $LS^u(x,t,r)$, converge to $\Delta u$, by design, as $t$ grows. This follows from the fact that the maximum distance of two databases is at most $n$ by the hamming distance definition. Thus when $t=n$, $LS^u(x,t)$ and $LS^u(x,t,r)$ measure sensitivity in all possible databases. We refer to those functions as \textit{bounded functions}. 


Note that one can force a given function $\delta^u(x,t,r)$ to be bounded by simply replacing it by its bounded version $\min(\delta^u(x,t,r), \Delta u)$. We show that $\min(\delta^u(x,t,r), \Delta u)$ is admissible and bounded.

\begin{restatable}{lemma}{lemmadeltaconvergence}
	\sloppy If $\delta^u(x, t, r)$ is admissible, then $min(\delta^u(x, t, r), \Delta u)$ is admissible and bounded.
    \label{lemma:delta_convergence}
\end{restatable}

\begin{definition}
	(Boundedness) A sensitivity function $\delta^u(x,t,r)$ is said to be bounded if $\delta^u(x,t,r)=\Delta u$ for all $t \geq n$.
\end{definition}

\begin{figure}[!ht]
	\centering	
	
	\begin{tikzpicture}
		\draw [fill] (0.5,0.5) circle [radius=0.08];
		\draw [fill] (1,1) circle [radius=0.08];
		\draw [fill] (1.5,1.5) circle [radius=0.08];
		\draw [fill] (2,2) circle [radius=0.08];
		\draw [fill] (2.5,2.5) circle [radius=0.08];
		\draw [fill] (3,3) circle [radius=0.08];
		\draw [fill] (3.5,3) circle [radius=0.08];
		\draw [fill] (4,3) circle [radius=0.08];
		\draw [fill] (4.5,3) circle [radius=0.08];
		\draw [fill] (5,3) circle [radius=0.08];
		\draw [fill] (5.5,3) circle [radius=0.08];
		\draw [fill] (6,3) circle [radius=0.08];
		\draw [dashed] (3,0) -- (3,3);
		\draw [dashed] (0,3) -- (3,3);

		\node at (3, -0.3) { $n$};
		\node at (-0.4, 3) { $\Delta u$};
		
		\draw [->, thick] (0,0) -- (0,4.5);
		\draw [->, thick] (0,0) -- (6.4,0);
		\node at (0,4.8) {$\delta^u(x,t,r)$};
		\node at (6,.3) {$t$};
	\end{tikzpicture}    	
    \caption{Boundedness - $\delta^u (x,t,r)$ converges to $\Delta u$ when $t \geq n$ where $n$ is the size of the database.}
\end{figure}
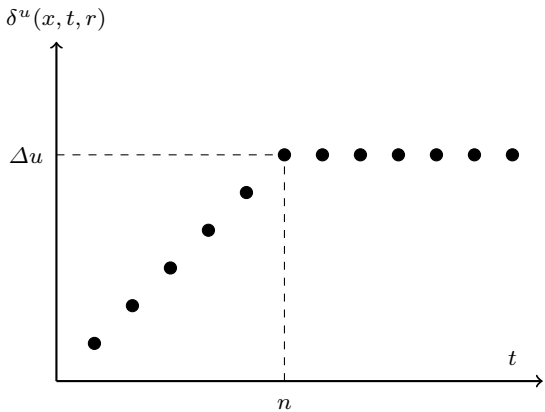

Thus, we can replace $\delta^u(x,t,r)$ with $\min(\delta^u(x,t,r), \Delta u)$ since its admissible. In terms of accuracy, this replacement is beneficial. We have that $\delta^u(x,t,r) \geq \min(\delta^u(x,t,r), \Delta u)$ for all database $x$, $t \geq 0$ and $r \in \range$. Thus $\delta^u(x,t,r)$ is always larger than $\min(\delta^u(x,t,r), \Delta u)$ meaning that local dampening injects less noise as sensitivity is proportional to the noise. This means that we can impose boundedness for any function and, beyond that, we have gains in accuracy as it injects less noise.

\subsubsection{Monotonicity}
\label{sec:monotonic}

We introduce the notion of monotonicity in our context. When the utility score $u(x,r)$ is a monotonic function of $\delta^u(x,t,r)$ over $r \in \range$, we say that $\delta^u(x,t,r)$ is monotonic. We have three classifications for monotonicity.

\begin{definition}
	(Non-decreasing Monotonicity) Let $u(x,r)$ be a utility function and $\delta^u(x,t,r)$ be a sensitivity function. $\delta^u(x,t,r)$ is said to be monotonically non-decreasing if $\delta^u(x,t,r) \geq \delta^u(x,t,r')$ for all $x \in \domain$, $r,r' \in \range$, $t \geq 0$ such that $u(x,r) \geq u(x,r')$.
	\label{def:non_decreasing_monotonicity}
\end{definition}

And its symmetric definition is:

\begin{definition}
	(Non-increasing Monotonicity) Let $u(x,r)$ be a utility function and $\delta^u(x,t,r)$ be a sensitivity function. $\delta^u(x,t,r)$ is said to be monotonically non-increasing if $\delta^u(x,t,r) \geq \delta^u(x,t,r')$ for all $x \in \domain$, $r,r' \in \range$, $t \geq 0$ such that $u(x,r) \leq u(x,r')$.
	\label{def:non_increasing_monotonicity}
\end{definition}

Also, a sensitivity can be \textit{flat}:

\begin{definition}
	(Flat Monotonicity) Let $u(x,r)$ be a utility function and $\delta^u(x,t,r)$ be a sensitivity function. $\delta^u(x,t,r)$ is said to be flat if $\delta^u(x,t,r) = \delta^u(x,t,r')$ for all $x \in \domain$, $r,r' \in \range$, $t \geq 0$.
	\label{def:flat_monotonicity}
\end{definition}

We refer to a \textit{monotonic function} as a function that is either flat, monotonically non-decreasing or monotonically non-increasing.

Note that flat sensitivity functions are independent on $r$ and they are both monotonic non-increasing and monotonic non-decreasing. The global sensitivity $\Delta u$ and the local sensitivity $LS^u(x,t)$ are flat sensitivity functions since they do not depend on $r$.

Additionally, given a utility function $u$ and an sensitivity function $\delta^u(x,t,r)$, one can build a function $\hat{\delta}^u(x,t,r)$ from $\delta^u(x,t,r)$ such that $\hat{\delta}^u(x,t,r)$ is flat.

$$
	\hat{\delta}^u(x,t,r) = \max_{r' \in \range} \delta^u(x,t,r')
$$.


A drawback of using $\hat{\delta}^u(x,t,r)$ is that $\hat{\delta}^u(x,t,r) \geq \delta^u(x,t,r)$, for all $x$, $t \geq 0$ and $r \in \range$ meaning that $\hat{\delta}^u(x,t,r)$ returns a large upper bound for sensitivity and, consequently, hurts accuracy. Lemma \ref{lemma:max_admissible} ensures that $\hat{\delta}^u(x,t,r)$ is admissible if $\delta^u$ is admissible. 

\subsubsection{Stability}

\begin{definition}
	(Stability) A sensitivity function $\delta^u(x,t,r)$ is stable if $\delta^u$ is admissible, bounded and monotonic.
	\label{def:regular_function}
\end{definition}

Satisfying all three requirements (admissibility, boundedness and monotonicity) for designing a stable function may sound very restrictive. However, for all  definitions of sensitivity, two of them are naturally stable: global sensitivity $\Delta u$ and local sensitivity $LS^u(x,t)$. Only the element local sensitivity $LS^u(x,t,r)$ can be non-monotonic and, consequently, non-stable. Nevertheless, in Section \ref{sec:relaxing_monotonicity}, we argue that the requirement of strict monotonicity can be relaxed and an admissible bounded function with ``weak'' monotonicity can perform well in the local dampening mechanism. 

Besides, for any function, the requirement of boundedness can be easily imposed as shown in Section \ref{sec:bounded_functions} while still providing lower sensitivity.

\subsection{Dampening Function}
\label{sec:damp}

A crucial part of our mechanism is the \textit{dampening function}. We now define the dampening function $D_{u, \delta^u}(x,r)$, which uses an admissible sensitivity function $\delta^u(x, t, r)$ to return a dampened and scaled version of the original utility function.

\begin{figure}
	\centering
	\begin{tikzpicture}
		\draw [->, thick] (0,-2.5) -- (0,2.5);
		\draw [->, thick] (-4,0) -- (4,0);

		\draw [thick] (0,0) -- (0.6, 0.8) -- (1.8, 1.6) -- (3.7, 2.1) -- (4, 2.15);
		\draw [thick] (0,0) -- (-0.6,-0.8) -- (-1.8, -1.6) -- (-3.7, -2.1) -- (-4, -2.15);
		
		\foreach \x/\y [count=\xi from 1] in {0.8/0.6, 1.6/1.8, 2.1/3.7}{
			\draw [ultra thick] (-.1,\x) -- (.1,\x);
			\draw [ultra thick] (-.1,-\x) -- (.1,-\x);
			\node at (-.3,\x) {\xi};
			\node at (.4,-\x) {$-\xi$};
			\draw [fill] (\y,\x) circle [radius=0.08];
			\draw [fill] (-\y,-\x) circle [radius=0.08];
			\draw [dashed] (0,\x) -- (\y,\x);
			\draw [dashed] (0,-\x) -- (-\y,-\x);
			\draw [dashed] (\y, 0) -- (\y,\x);
			\draw [dashed] (-\y, 0) -- (-\y,-\x);
			\draw [-, thick] (\y,0.1) -- (\y,-0.1);
			\draw [-, thick] (-\y,0.1) -- (-\y,-0.1);

		}
		\node at (0.7, -0.35) { $b(x,1,r)$};
		\node at (2.1, -0.35) { $b(x,2,r)$};
		\node at (3.2, -0.35) { $\cdots$};

		\node at (-0.75, 0.35) { $b(x,-1,r)$};
		\node at (-2.25, 0.35) { $b(x,-2,r)$};
		\node at (-3.4, 0.35) { $\cdots$};

		\node at (3.9, -0.35) {$u$};

		\node at (.55,2.5) {$D_{u, \delta^u}$};
		
	\end{tikzpicture}    
    \caption{Dampening function $D_{u, \delta^u}$}
    \label{fig:dampening}
\end{figure}
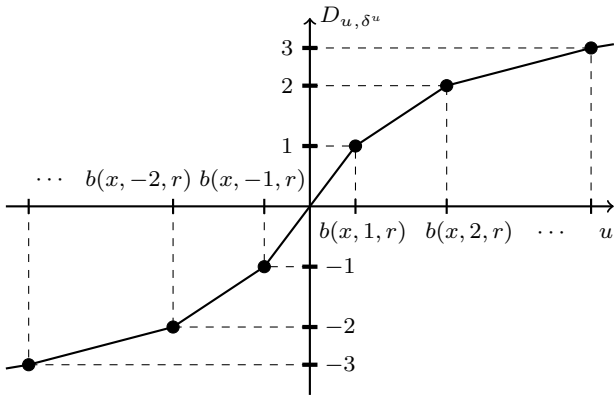

\begin{definition}
	(Dampening function). Given a utility function $u(x, r)$ and an admissible function $\delta^u(x, t, r)$, the dampening function $D_{u,\delta^u}(x,r)$ is defined as a piecewise linear interpolation over the points:	    
	$$< \ldots,(b(x,-1,r), -1),(b(x,0,r), 0),(b(x,1,r), 1), \ldots >$$
	where $b(x, i, r)$ is given by:		    	    
	\begin{equation*}	
		b(x, i, r) \coloneqq
		\begin{cases}
			\sum^{i-1}_{j=0} \delta^u(x, j, r) & \text{if} \ i > 0 \\
			0 & \text{if} \ i = 0 \\
			- b(x, -i, r)                    & \text{otherwise}     
		\end{cases}
	\end{equation*}
	Therefore,	    
	\begin{equation*}
		D_{u,\delta^u}(x, r) =  \frac{u(x, r)-b(x,i,r)}{b(x, i+1, r)- b(x, i, r)} + i
	\end{equation*}					
	where $i$ is defined as the smallest integer such that $u(x, r) \in  \left[ b(x,i,r), b(x,i+1,r)  \right)$.
\end{definition}

Figure \ref{fig:dampening} shows the general scheme of $D_{u, \delta^u}$. A crucial property of $D_{u, \delta^u}$ is that it scales $u$ so that the sensitivity of $D_{u, \delta^u}$ is bounded to $1$ (Lemma \ref{lemma:principal}).

\begin{restatable}{lemma}{lemmaprincipal}
	$|D_{u,\delta^u}(x,r) - D_{u,\delta^u}(y,r)| \leq 1$ for all $x,y$ such that $d(x,y) \leq 1$ and all $r \in R$ if $\delta^u$ is admissible.
	\label{lemma:principal}
\end{restatable}

\subsection{Local Dampening Mechanism}

We now state the local dampening mechanism a generic non-numeric differentially private mechanism. It takes a database $x$, the privacy budget $\epsilon$, the utility function $u$, an admissible sensitivity function $\delta^u$ and the range $\range$ of the function to be sanitized.

It samples an element from $r \in \range$ based on its dampened utility score $D_{u, \delta^u}(x,r)$. The larger the score, the higher the probability of sampling it.

\begin{definition}
	(Local dampening mechanism). The local dampening mechanism $\mecld(x, \epsilon, u, \delta^{u}, \mathcal{R})$ selects and outputs an element $r \in \mathcal{R}$ with probability proportional to $\exp \big(  \frac{\epsilon \; D_{u,\delta^u}(x, r)}{2}\big) $.
\end{definition}

This version of the local dampening mechanism is specially effective when the sensitivity function is flat. In the following example, we demonstrate the process operation of the local dampening mechanism.

\begin{example}
	(Local dampening mechanism) This example explores the local dampening mechanism using the local sensitivity definition while the element local sensitivity is addressed in Section \ref{section:shifted_local_dampening}. Let $G$ be the graph of Figure \ref{fig:graphc}. As we have discussed in Example \ref{example:example_local_sensitivity_dist_1}, we have that $LS^{EBC}(G,0)=3$ and $LS^{EBC}(G,1)=5$. The EBC scores for the vertices are $EBC(a)=EBC(b)=6.5$ and $EBC(v_i)=0$, for $0 \leq i \leq 5$. Their dampened EBC scores are:
	\begin{gather*}
		D_{EBC, LS^{EBC}}(G, a) = D_{EBC, LS^{EBC}}(G, b) = 1.7	\\	
		D_{EBC, LS^{EBC}}(G, v_i) = 0 \text{, for } 0 \leq i \leq 5
	\end{gather*}
	For instance, assuming $\epsilon = 2.0$, the probability for each node to be selected is:
	\begin{gather*}
		Pr[\text{a is selected}] = Pr[\text{b is selected}] \propto \exp(1.7) = 5.47 \\
		Pr[v_i \text{ is selected}] \propto \exp(0) = 1.0  \text{, for } 0 \leq i \leq 5			
	\end{gather*}
	\sloppy Normalizing, we have that $Pr[\text{a is selected}] = Pr[\text{b is selected}] = 0.32$ and $Pr[v_i \text{ is selected}] = 0.06$. Thus the local dampening mechanism samples a element with those probabilities.

\label{example:dampening}
\end{example}


\subsection{Privacy Guarantee}
\label{sec:priv}

We now prove that the local dampening mechanism $\mecld$ ensures $\epsilon$-differential privacy (Theorem \ref{theorem:priv_local_dampening}). The privacy correctness proof follows from the exponential mechanism correctness \cite{mcsherry2007mechanism} and Lemma \ref{lemma:principal}.

\begin{restatable}{theorem}{privlocaldampening}
	$\mecld$ satisfies $\epsilon$-Differential Privacy if $\delta$ is admissible.
	\label{theorem:priv_local_dampening}
\end{restatable}

\begin{proof}    
Given two neighboring databases $x,y \in D^n$ (i.e., $d(x,y)\leq 1$) and an output $r \in \mathcal{R}$. We show that the ratio of the probability of $r$ being produced by local dampening mechanism on database $x$ and $y$ is bounded by $\exp(\epsilon)$.

\begin{align*}	
		\frac{P_x(r)}{P_y(r)} & = \frac{P[\mecld(x,u,\mathcal{R}) = r]}{P[\mecld(y,u,\mathcal{R}) = r]} \\		
        & = \frac{ \left( \frac{\exp(\frac{\epsilon D_{u,\delta}(x,r)}{2})}{\sum_{r' \in \mathcal{R}} \exp(\frac{\epsilon D_{u,\delta}(x,r')}{2})} \right)}
        { \left( \frac{\exp(\frac{\epsilon D_{u,\delta}(y,r)}{2})}{\sum_{r' \in \mathcal{R}} \exp(\frac{\epsilon D_{u,\delta}(y,r')}{2})} \right)} \\
        & = \left( \frac{\exp(\frac{\epsilon D_{u,\delta}(x,r)}{2})}{\exp(\frac{\epsilon D_{u,\delta}(y,r)}{2})} \right) \cdot \left( \frac{\sum_{r' \in \mathcal{R}} \exp(\frac{\epsilon D_{u,\delta}(y,r')}{2})}{\sum_{r' \in \mathcal{R}} \exp(\frac{\epsilon D_{u,\delta}(x,r')}{2})} \right) \\
        & \leq \text{exp}\left( \frac{\epsilon (D_{u,\delta}(x,r') - D_{u,\delta}(y,r'))}{2} \right)\\
        & \quad \cdot \left( \frac{\sum_{r' \in \mathcal{R}} \exp(\frac{\epsilon (D_{u,\delta}(x,r')+1)}{2})}{\sum_{r' \in \mathcal{R}} \exp(\frac{\epsilon D_{u,\delta}(x,r')}{2})} \right) \\
        & \leq \text{exp}\left(\frac{\epsilon}{2}\right) \cdot \text{exp}\left(\frac{\varepsilon}{2}\right) \cdot \left( \frac{\sum_{r' \in \mathcal{R}} \exp(\frac{\varepsilon D_{u,\delta}(x,r')}{2})}{\sum_{r' \in \mathcal{R}} \exp(\frac{\varepsilon D_{u,\delta}(x,r')}{2})} \right) \\
        & = \text{exp}(\varepsilon)                        
\end{align*}
The two inequalities follow from lemma \ref{lemma:principal}. By symmetry, $\frac{P[\mecld(x,u,\mathcal{R}) = r]}{P[\mecld(y,u,\mathcal{R}) = r]} \geq \exp(-\epsilon)$ holds.
\end{proof}

\subsection{Related Work}
\label{sec:related_work}

There is a vast literature on differential privacy for numeric queries, and we refer the interested reader to \cite{DBLP:conf/sigmod/Machanavajjhala17} for a recent survey. Our work was previously published in \cite{farias2020local}. This paper advances on the understanding of the local dampening mechanism and on the related work. We discuss and carry out experiment with more recent related work, the permute-and-flip mechanism. We provide a new application to the local dampening mechanisms and compare to the related work. We give a new theoretical accuracy analysis where we show how to compare two instances of the local dampening mechanism. Besides that, we show that, under some conditions, our approach is never worse than the exponential mechanism.

Given that, in this section, we discuss also the two available differential privacy approaches for the non-numeric, also known as the selection problem, setting in the literature, the exponential mechanism and the permute-and-flip mechanism.

\subsubsection{Exponential Mechanism}

The exponential mechanism $\mecexp$ \cite{mcsherry2007mechanism} is the most used approach for providing differential privacy to the non-numeric setting. It uses a notion of global sensitivity $\Delta u$ (Definition \ref{def:global_sensitivity}).

The exponential mechanism privately answers a function $f: \domain \rightarrow \range$ applied to database $x$ by sampling an element $r \in \range$ with probability proportional to its utility score $u(x,r)$. It uses the exponential distribution to assign probabilities for each $r \in \range$. The exponential mechanism is stated as follows:

\begin{definition}
	(Exponential Mechanism \cite{mcsherry2007mechanism}). The exponential mechanism $\mecexp(x, \epsilon, u, \mathcal{R})$ selects and outputs an element $r \in \mathcal{R}$ with probability proportional to $\exp \big(\frac{\epsilon \; u(x, r)}{2 \Delta u}\big) $.
	\label{def:exponential_mechanism}
\end{definition}


The exponential mechanism satisfies $\epsilon$-differential privacy \cite{mcsherry2007mechanism}. 

In Section \ref{section:shifted_local_dampening}, we show that, under some conditions, the exponential mechanism is never worse than the local dampening in terms of accuracy. Additionally, we carry out an experimental evaluation with the two applications that we tackle in this work: influential node analysis (Sections \ref{section:influential_node_analysis}) and decision tree induction (Section \ref{section:decision_tree_induction}).

\begin{example}
	(Comparison local dampening mechanism with exponential mechanism)
	We make a simple comparison of the probabilities of the local dampening mechanism  with the exponential mechanism in Example \ref{example:dampening}. 

	In Example \ref{example:dampening}, we have that $Pr[\text{a is selected}] = Pr[\text{b is selected}] = 0.32$ and $Pr[v_i \text{ is selected}] = 0.06$. While, according to Definition \ref{def:exponential_mechanism}, the exponential mechanism obtained that $Pr[\text{a is selected}] = Pr[\text{b is selected}] = 0.22$ and $Pr[v_i \text{ is selected}] = 0.09$. Thus local dampening yields a higher probability of choosing the node with highest score.
	\label{example:exponential_mechanism}
\end{example}

\subsubsection{Permute-and-Flip}

The \textit{permute-and-flip} mechanism,$\mecpf$, \cite{mckenna2020permute} is recent work that also addresses differential privacy for the non-numeric setting. It is defined as an iterative algorithm that employs the exponential distribution to assign probabilities for each element $r$.

Algorithm \ref{algo:permute_and_flip} formally defines permute-and-flip. 

\begin{algorithm}[htp]
    \SetAlgoLined\DontPrintSemicolon
    \SetKwFunction{algo}{$\mecpf$(Database $x$, Privacy Budget $B$, utility function $u$, Range set $\range$)}
    \SetKwProg{myalg}{Procedure}{}{}
    \myalg{\algo}{
	$u^* = \max_{r \in \range} u(x,r)$\;    
    \For{$r \in RandomPermutation(\range)$ }{        
        $p_r = \exp \left( \frac{\epsilon}{2 \Delta u} (u(x,r) - u^*) \right)$ \;
        \If{$Bernoulli(p_r)$}{     	   
        	\KwRet{r}\;
    	}
    } 
	}{}
    \addtocontents{loa}{\protect\addvspace{-9pt}}
    \caption{Permute-and-Flip}
    \label{algo:permute_and_flip}

\end{algorithm}

Basically, the algorithm iterates over a random permutation of the elements $r \in \range$ and then flip a biased coin with probability $\frac{\epsilon}{2 \Delta u} (u(x,r) - u^*)$. $u^*$ is the maximum utility observed over all elements in the range set $\range$ given the input database $x$. Thus, the closer $u(x,r) - u^*$, more likely is $r$ to be outputted. The mechanism is guaranteed to terminate with a result because if $u(x,r) = u^*$, then the probability of heads is $1$.

McKenna et al \cite{mckenna2020permute} show that their approach is also never worse than the exponential mechanism in terms of accuracy. We conduct an empirical comparison of permute-and-flip mechanism to local dampening in Sections \ref{section:influential_node_analysis} and \ref{section:decision_tree_induction}.

\section{Shifted Local Dampening Mechanism}
\label{section:shifted_local_dampening}

In this section, we present a second version of the local dampening mechanism name \textit{shifted local dampening} mechanism $\mecsld$. This version is designed for non-flat monotonic sensitivity functions which is the most usual case in our experiments.

We develop an insightful discussion on accuracy of the shifted local dampening mechanism (Section \ref{section:accuracy_analysis}). We provide tools to compare two instances of the shifted local mechanism in terms of accuracy. Also, these tools guide on the design of good sensitivity functions that provide accurate $\mecsld$ instances (Section \ref{sec:acc_stable_function}). We show that, with a stable sensitivity function, the local dampening mechanism is never worse than the exponential mechanism (Section \ref{sec:comparison_exponential_mechanism}). Additionally, even if the stability condition is not met, we discuss how to construct good sensitivity functions (Section \ref{sec:relaxing_monotonicity}).

\subsection {Inversion problem}
\label{sec:inversion_problem}

First, we exemplify an issue that happens when the sensitivity function is not monotonic.


Consider the scenario where we dampen the utility scores of the elements $r \in \range$ with the sensitivity function $\delta^u$ that is not monotonic. This might be the case when we use $\delta^u(x,t,r)$ as the element local sensitivity, $\delta^u(x,t,r) = LS^u(x,t,r)$.


In this situation, Example \ref{example:inversion} illustrates a case where the local dampening change the relative order of the dampened utility scores compared to the original utility scores. We refer to this problem as the \textit{inversion problem}.

\begin{example}
(Inversion problem) Consider the following setup: $\mathcal{R} = \{r_1, r_2\}$, $\delta^u(x,0,r_1)=1$, $\delta^u(x,1,r_1)=2$, $\delta^u(x,0,r_2)=4$, $u(x,r_1)=3$ and $u(x,r_2)=4$. When applying $D_{u, \delta^u}$ to $r_1$ and $r_2$, we obtain $D_{u, \delta^u}(x, r_1)=2$ and $D_{u, \delta^u}(x, r_2)=1$. Originally, $r_2$ is more useful than $r_1$ but after dampening it inverts. This hurts accuracy since the local dampening mechanism will choose $r_1$ with higher probability.
\label{example:inversion}
\end{example}

\subsection{Shifted Local Dampening}
\label{sec:shifted_local_dampening}

The key idea for this extension is the use of shifting in the utility score to take advantage of non-flat monotonic sensitivity functions $\delta^u$. The discussion in this section is focused on non-flat monotonic sensitivity functions. However, we show later that the shifted local dampening also performs well for non strictly monotonic functions.

Example \ref{example:shifting} shows a case where shifting increases the probability of high utility elements to be chosen (i.e. improves accuracy) when $\delta^u$ is monotonically non-decreasing.

\begin{example}
\sloppy (Utility function shifting) Consider the graph $G$ from figure \ref{fig:original_graph}. For nodes $a$ and $b$, their measured element local sensitivities are: $LS^{EBC}(G,0,a)=LS^{EBC}(G,0,b)=3$ and $LS^{EBC}(G,1,a)=LS^{EBC}(G,1,b)=5$. For a node $v_i$, for $0 \leq i \leq 5$, its measured sensitivity is $LS^{EBC}(G,0,v_i)=2$. We observe the non-decreasing monotonicity of $LS^{EBC}$, since the EBC scores are $EBC(a)=EBC(b)=6.5$ and $EBC(v_i)=0$, for $0 \leq i \leq 5$.

\tikzstyle{vertex}=[circle,minimum size=15pt,inner sep=0pt, draw=black]
\tikzstyle{edge} = [draw,-]
\tikzstyle{dashed edge} = [draw,-,densely dashed]

\begin{figure}[ht]
	\caption{Original Graph $G$}
	\centering

    \begin{tikzpicture}[]
        \foreach \pos/\name in {{(1.11,0)/a}, {(2.22,0)/b}, {(3.33,-1.5)/v_5}, {(2.67,-1.5)/v_4}, {(2,-1.5)/v_3}, {(1.34,-1.5)/v_2}, {(0.67,-1.5)/v_1}, {(0,-1.5)/v_0}}
            \node[vertex] (\name) at \pos {$\name$};
    
        \foreach \source/ \dest in {a/b, a/v_0, a/v_1, a/v_2, a/v_3, b/v_2, b/v_3, b/v_4, b/v_5, v_0/v_1, v_4/v_5}
            \path[edge] (\source) -- (\dest);    
                    
    \end{tikzpicture}

\label{fig:original_graph}
\end{figure}
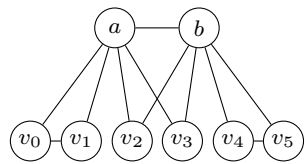

For instance, shifting the $EBC$ scores by $-7$, we get that $EBC'(a)=EBC'(b)=-0.5$ and $EBC'(v_i)=-7$, for $0 \leq i \leq 5$. Then we compute their dampened $EBC'$ scores:
\begin{gather*}
	D_{EBC', LS^{EBC}}(G, a) = D_{EBC', LS^{EBC}}(G, b) = 0.1 \\
	D_{EBC', LS^{EBC}}(G, v_i) = -2 \text{, for } 0 \leq i \leq 5	
\end{gather*}	
Let $\epsilon = 2.0$. The probability for each node to be selected is:
\begin{gather*}
	Pr[\text{a is selected}] = Pr[\text{b is selected}] \propto \exp(0.1) = 0.44 \\
	Pr[v_i \text{ is selected}] \propto \exp(-2) = 0.13  \text{, for } 0 \leq i \leq 5
\end{gather*}
Normalizing, we have that $Pr[\text{a is selected}] = Pr[\text{b is selected}] = 0.472$ and $Pr[v_i \text{ is selected}] = 0.0046$. Recall that, the exponential mechanism obtained that $Pr[\text{a is selected}] = Pr[\text{b is selected}] = 0.22$ and $Pr[v_i \text{ is selected}] = 0.09$ (Example \ref{example:exponential_mechanism}) and, for the unshifted local dampening mechanism, (Example \ref{example:dampening}), we have that $Pr[\text{a is selected}] = Pr[\text{b is selected}] = 0.32$ and $Pr[v_i \text{ is selected}] = 0.06$. The nodes with highest score increase probability compared to the unshifted local dampening and the exponential mechanism.
\label{example:shifting}
\end{example}

For the sake of argument, suppose that $\delta^u(x,t,r)$ is monotonically non-decreasing. We design the shifting in a way that it rearranges the utilities scores such that the distribution of the utility scores is more spread. 

The idea is the following: we shift left enough so that all utility scores are negative. The elements with larger utility score are the elements with smallest absolute value after shifting. Thus, these shifted scores are dampened with large $\delta^u(x,t,r)$ (by assumption of non-decreasing monotonicity). This implies that large utility scores are dampened closer to $0$ and the opposite happens to elements with small utility scores. The elements with small utility scores are dampened with small $\delta^u(x,t,r)$ and, consequently, the scores are less attenuated and far from $0$. This results in more spread distribution of utility scores.



Hereby we propose to replace the original utility function $u$ with its shifted version $u^s$ where $s$ is the utility score shift and 
$$
	u^s(x,r) = u(x,r) - s.
$$

One could design a private query, consuming part of the privacy budget, to choose $s$ such that it minimizes some loss function to optimize accuracy. In this work, we set $s$ to a value that does not depend on private data, $s \rightarrow \infty$. In what follows, the shifted local dampening mechanism is stated as follows:

\begin{definition}
	(Shifted Local Dampening Mechanism - non-decreasing sensitivity function). The shifted local dampening mechanism $\mecsld(x, \epsilon, u, \delta^{u}, \mathcal{R})$ outputs an element $r \in \mathcal{R}$ with probability equals to 
	
	$$
		\lim_{s \to \infty} \left( 
			\frac{ 
				\exp \left( \frac{\epsilon \, D_{u^{s}, \delta^u}(x,r)}{2} \right) 
			}{ 
				\sum_{r' \in \mathcal{R}} \exp \left( \frac{\epsilon \, D_{u^{s}, \delta^u}(x,r')}{2} \right) 
			} 
		\right).
	$$
	\label{def:shifted_local_dampening_non_decreasing}
\end{definition}

When $\delta^u$ is monotonically non-increasing the following definition of the shifted local dampening mechanism applies:

\begin{definition}
	(Shifted Local Dampening Mechanism - non-increasing sensitivity function). The shifted local dampening mechanism $\mecsld(x, \epsilon, u, \delta^{u}, \mathcal{R})$ outputs an element $r \in \mathcal{R}$ with probability equals to 
	
	$$
		\lim_{s \to -\infty} \left( 
			\frac{ 
				\exp \left( \frac{\epsilon \, D_{u^{s}, \deltau}(x,r)}{2} \right) 
			}{ 
				\sum_{r' \in \mathcal{R}} \exp \left( \frac{\epsilon \, D_{u^{s}, \deltau}(x,r')}{2} \right)
			} 
		\right).
	$$
	\label{def:shifted_local_dampening_non_increasing}
\end{definition}

For the case of functions that do not depend on $r$, both versions of the shifted local dampening mechanism are applicable.

\subsection{Privacy Guarantee and implementation}

We now prove that the shifted local dampening mechanism $\mecsld$ ensures $\epsilon$-differential privacy. For the privacy guarantee, the sensitivity function $\delta^u$ just needs to be admissable and bounded but not necessarily monotonic. Recall that boundedness can be easily achieved (Section \ref{sec:bounded_functions}).

We first show an intermediate result:


\begin{restatable}{lemma}{lemmaconvergence}
	If $\delta^u$ is an admissible and bounded sensitivity function then 
	$ \frac{ 
		\exp(\epsilon \; D_{u^s, \deltau}(x, r)/2)
	}{ 
		\sum_{r' \in \range} \exp( \epsilon \; D_{u^s, \deltau}(x, r')/2)
	} =
	\frac{ \exp(\epsilon \; D_{u^{s_0}, \deltau}(x, r)/2)
	}{ 
		\sum_{r' \in \range} \exp( \epsilon \; D_{u^{s_0}, \deltau}(x, r')/2)
	}$ 
	for $s \geq s_0$ where $ s_0 = n  \Delta u + \max_{r' \in R} u(x,r')$ and $n$ is the size of the input database.
	\label{lemma:convergence}
\end{restatable}

Lemma also \ref{lemma:convergence} gives hint about the implementation. It suffices to shift by $n  \Delta u + \max_{r' \in R} u(x,r')$ to meet the definition of the shifted local dampening. Also, from Lemma \ref{lemma:convergence}, it follows directly (Corollary \ref{cor:existence}).


\begin{corollary} 
	$ \lim_{s \to \infty} \left( 
		\frac{ 
			\exp \left( \frac{\epsilon \, D_{u^{s}, \deltau}(x,r)}{2} \right) 
		}{ 
			\sum_{r' \in R} \exp \left( 
				\frac{\epsilon \, D_{u^{s}, \deltau}(x,r')}{2} 
			\right)
		} 
	\right)  $ exists and is equal to $\frac{ 
		\exp(\epsilon \; D_{u^s, \deltau}(x, r)/2)
	}{ 
		\sum_{r' \in \range} \exp( \epsilon \; D_{u^s, \deltau}(x, r')/2)
	}$ 
	for $s \geq s_0$ where $ s_0 = n  \Delta u + \max_{r' \in R} u(x,r')$ and $n$ is the size of the input database.
	\label{cor:existence}
\end{corollary}

The privacy correctness proof follows from the exponential mechanism correctness \cite{mcsherry2007mechanism}, Lemma \ref{lemma:principal} and Corollary \ref{cor:existence}. In this proof we use the non-decreasing admissable function version of the local dampening (Definition \ref{def:shifted_local_dampening_non_decreasing}). The non-increasing version (Definition \ref{def:shifted_local_dampening_non_increasing}) privacy guarantee proof is symmetric.

\begin{theorem}
	$\mecsld$ satisfies $\epsilon$-Differential Privacy if $\delta^u$ is admissible and bounded.
	\label{theorem:shifted_priv_local_dampening}
\end{theorem}

\subsection{Accuracy Analysis}
\label{section:accuracy_analysis}

In this section, we provide theoretical analysis on the accuracy. We aim to answer to the following questions: i) How to compare two instances of the local dampening with two different admissible functions?; ii) Under which conditions does the local dampening performs more accurately than the exponential mechanism?; iii) If those conditions are not met, how to build good admissible functions? and iv) How does local dampening compare to the exponential mechanism in terms of accuracy?.

We evaluate the accuracy of a given mechanism $\mec$ by studying the error random variable $\error$. $\error$ gives how much the element sampled by $\mec$ differ from the optimal element in terms of utility.

$$ \error(\mec, x) = u^* - u(x, \mec(x))$$
where $u^*$ is the optimal utility score, $u^* = \max_{r \in \range} u(x,r)$.

To compare two instances of the local dampening for the same problem, we need to analyse the features of the function $\delta^u$. We develop a discussion on accuracy guarantees for stable functions where we show how to compare two stable functions and show that, using a stable function, the local dampening mechanism is never worse than the exponential mechanism in terms of accuracy.

\subsubsection{Accuracy Analysis for Stable Sensitivity Functions}
\label{sec:acc_stable_function}

Two instances of the local dampening mechanism can be compared by their stable sensitivity functions. As lower sensitivity means higher accuracy, a stable sensitivity function that produces lower values implies in higher accuracy. For that analysis we establish a relation of dominance between two stable sensitivity functions:


\begin{definition}
	(Dominance) Let $\delta^u(x,t,r)$ and $\bar{\delta}^u(x,t,r)$ be two stable sensitivity functions and $x$ be a database. Let $\alpha(x,t,r)$ refer to the gap between $\delta^u(x,t,r)$ and $\bar{\delta}^u(x,t,r)$: $\alpha(x,t,r)  = \bar{\delta}^u(x,t,r) - \delta^u(x,t,r)$. Assume that $\range = \{r_1,...,r_q\}$ is ordered such that $u(x,r_1) \geq \cdots \geq u(x,r_q)$. If $\alpha(x,t,r_1) \geq \alpha(x,t,r_2) \geq \cdots \geq \alpha(x,t,r_q) \geq 0$ for all $t \geq 0$, then $\delta^u(x,t,r)$ dominates $\bar{\delta}^u(x,t,r)$.
	\label{definition:dominance}
\end{definition}

Given that, we can affirm that an instance of the local dampening mechanism using $\delta^u(x,t,r)$ is never worse than an instance using the dominated $\deltabar^u(x,t,r)$:

\begin{restatable}{lemma}{lemmaLocalDampeningAccuracy}
	(Shifted Local Dampening Accuracy) Let $\delta^u(x,t,r)$ and $\bar{\delta}^u(x,t,r)$ be two stable functions and $x$ be a database. If $\delta^u(x,t,r)$ dominates $\bar{\delta}^u(x,t,r)$ then:
	\begin{enumerate}
		\item $Pr[\error(\mecsld,x) \geq \theta] \leq Pr[\error(\mecsldbar,x) \geq t]$ for all $\theta \geq 0$,
		\item $\expected[\error(\mecsld, x)] \leq \expected[\error(\mecsldbar, x)]$,
	\end{enumerate}
	where $\mecsld$ represents an instance of the shifted local dampening mechanism using $\delta^u$ as the sensitivity function while $\overline{\mecsld}$ is an instance using $\bar{\delta}^u$.
	\label{lemma:local_dampening_accuracy}
\end{restatable}


We can use Lemma \ref{lemma:local_dampening_accuracy} as a tool understand the accuracy of the local dampening mechanism. It suggests that a sensitivity function should be as inclined as possible, i.e., a higher difference between two gaps $\alpha(x,t,r_i)$ and $\alpha(x,t,r_{i+1})$ implies higher accuracy. Also, the gaps $\alpha(x,t,r_i)$ should be as large as possible. To illustrate this, we provide an example:

\begin{example} (Shifted Local Dampening Accuracy) Let $\delta^{u}_{\beta}(x,t,r)=(u(x,r).t)/\beta, \Delta u$ be a sensitivity function where $\beta$ is a parameter to control magnitude of $\delta^{u}$ and, consequently, to control the accuracy of the shifted local dampening mechanism based on Definition \ref{definition:dominance} and Lemma \ref{lemma:local_dampening_accuracy}. 
	
Note that $\delta^{u}_{\beta}$ is monotonically non-decreasing and can be bounded by the process depicted in Section \ref{sec:bounded_functions}. Also, assume that $\delta^{u}_{\beta}$ is admissable for $u$, thus it is a stable sensitivity function.
	
Without loss of generality, suppose that  $\range = \{r_1,...,r_q\}$ is ordered such that $u(x,r_1) \geq \cdots \geq u(x,r_q)$. The gaps $\alpha(x,t,r_i)$ (Definition \ref{definition:dominance}) between two sensitivity functions $\delta^{u}_{\beta_1}$ and $\delta^{u}_{\beta_2}$ are given by $\alpha(x,t,r_i) = \delta^{u}_{\beta_1}(x,t,r_i) - \delta^{u}_{\beta_2}(x,t,r_i) = (u(x,r_i).t)/\beta_1 - (u(x,r_i).t)/\beta_2 =  u(x,r_i).t/(1/\beta_1-1/\beta_2)$, for $1 \leq i \leq q$ and $t \geq 0$. In what follows, when $\beta_1 \geq \beta_2$, $\delta^{u}_{\beta_1}$ dominates $\delta^{u}_{\beta_2}$ as $\alpha(x,t,r_i) = u(x,r_i).t/(1/\beta_1-1/\beta_2) \geq u(x,r_{i+1}).t/(1/\beta_1-1/\beta_2) \geq \alpha(x,t,r_i)$ because $u(x,r_{i}) \geq u(x,r_{i+1})$ and $(1/\beta_1-1/\beta_2) < 0$ for $t \geq 0$.

This way, the parameter $\beta$ also control the gaps $\alpha(x,t,r_i)$ since it is proportional to the term $(1/\beta_1-1/\beta_2)$. This way, we illustrate this by varying the parameter $\beta$ while fixing $\epsilon = 1.0$, $\Delta u = 100$, $\range = \{ r_1, r_2, r_3\}$ $u(x,r_1)=10$, $u(x,r_2)=20$ and $u(x,r_3)=30$. We report the expected error $\expected[]\error(\mec, x)]$.

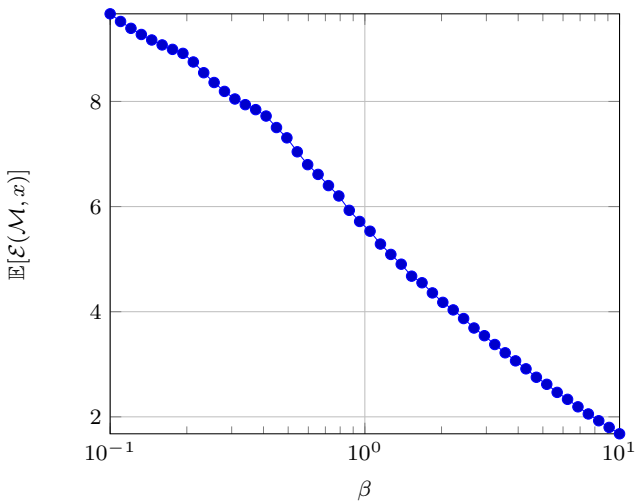
\begin{figure}[h!]
    
\center

\begin{subfigure}[b]{0.5\textwidth}
    \resizebox{\linewidth}{!}{
        \begin{tikzpicture}
            \begin{semilogxaxis}[grid=major,
                enlargelimits=false,
                xlabel={$\beta$},
                ylabel={$\expected[\error(\mec, x)]$},
                xticklabel style={
                    /pgf/number format/fixed,
                    /pgf/number format/precision=5,
                }]

                \addplot table[x=beta,y=error] {./data/dummy/dummy.dat};

            \end{semilogxaxis}
        \end{tikzpicture}            
    }
    
\end{subfigure}
\caption{Expected error for shifted local dampening while varying the paramenter $\beta$.}
\end{figure}
\end{example}

\subsubsection{Comparison to the Exponential Mechanism}
\label{sec:comparison_exponential_mechanism}

A very useful property of both versions of the local dampening mechanism is that the exponential mechanism is an instance of the local dampening mechanism. The exponential mechanism is obtained by setting $\delta^u(x,t,r)=\Delta u$ in an instance of the shifted local dampening (Lemma \ref{lemma:local_dampening_exponential}). 

\begin{restatable}{lemma}{lemmalocaldampeningexponential}
	$Pr[\mecexp(x, \epsilon, u, \mathcal{R})=r] =Pr[\mecld(x, \epsilon, u, \delta^{u}, \mathcal{R})= r]= Pr[\mecsld(x, \epsilon, u, \delta^{u}, \mathcal{R}) = r]$ for $\delta^u(x,t,r)=\Delta u$, for all $x \in \domain$, for all $\epsilon \in \real$ and for all  $r \in \range$.
	\label{lemma:local_dampening_exponential}
\end{restatable}

Thus we can use Lemma \ref{lemma:local_dampening_accuracy} to compare any instance of the exponential mechanism using a given stable function $\delta^u(x,t,r)$ against the exponential mechanism. Note by the assumption of boundedness of the stable sensitivity function $\delta^u(x,t,r)$ we have that $\delta^u(x,t,r) \leq \Delta u$, for all $x$, $t \geq 0$ and $r \in \range$. It implies that $\delta^u(x,t,r)$ dominates $\Delta u$. Thus the following corollary holds:

\begin{corollary}
	\sloppy Let $\delta^u$ be a stable function. The shifted local dampening mechanism $\mecsld(x, \epsilon, u, \delta^{u}, \mathcal{R})$ is never worse than the exponential mechanism $\mecexp(x, \epsilon, u, \mathcal{R})$. That is:
	\begin{enumerate}
		\item $Pr[\error(\mecsld,x) \geq t] \leq Pr[\error(\mecexp,x) \geq t]$ for all $t \geq 0$,
		\item $\expected[\error(\mecsld, x)] \leq \expected[\error(\mecexp, x)]$.
	\end{enumerate}
\end{corollary}

Note that this result implies that we can use local sensitivity $LS^u(x,t)$ safely since $LS^u(x,t)$ is a stable sensitivity function, i.e., using the shifted local dampening mechanism with $LS^u(x,t)$ as the sensitivity function is never worse than the exponential mechanism. Yet, it suggests that the larger the difference between $LS^u(x,t)$ and $\Delta u$, the more accurate it is in relation to the exponential mechanism.

This result also suggests that using the $\Delta u$ as a sensitivity function is the worst-case stable function. Given that, what would be the best stable function? The element local sensitivity $LS^u(x,t,r)$ function is a good candidate. As shown before, $LS^u(x,t,r)$ is admissable and bounded. However, $LS^u(x,t,r)$ is not necessarily monotonic. We demonstrate that $LS^u(x,t,r)$ is minimum admissable, i.e. it dominates all admissable functions:

\begin{restatable}{lemma}{lemmaminimumls}
    \sloppy $LS^u(x, t, r)$ is minimum admissable, i.e. $LS^u(x,t, r)$ dominates any admissible sensitivity function $\delta^u(x, t, r)$.
    \label{lemma:minimumls}
\end{restatable}

Even if $LS^u(x,t,r)$ happens to be non monotonic, we devote the next subsection to discuss that only a ``weak'' monotonicity is enough for our mechanism. Also, we discuss other cases where the local dampening mechanism also performs well.

\subsection{Relaxing Monotonicity}
\label{sec:relaxing_monotonicity}

We have shown theoretical guarantees for the accuracy of the shifted local dampening mechanism using stable functions. For sensitivity functions like the global sensitivity $\Delta u$ and the local sensitivity $LS^u(x,t)$ we have strong accuracy guarantees.

Strict monotonicity may be a complex goal to achieve. In the applications and datasets analysed in our experimental section, none of them satisfy the strict monotonicity requirement. Yet, the shifted local dampening mechanism outperforms the exponential mechanism in our experiments.

For those sensitivity functions that violates monotonicity, we use the results in Section \ref{sec:acc_stable_function} as guide to construct a good sensitivity function. The same analysis also works here, Lemma \ref{lemma:local_dampening_accuracy} suggests that a sensitivity function should be as inclined as possible, i.e., a higher difference between two gaps $\alpha(x,t,r_i)$ and $\alpha(x,t,r_{i+1})$ implies in higher accuracy. Also, the gaps $\alpha(x,t,r_i)$ should be as large as possible.


For the running example of this paper (Example \ref{example:example1}), we designed an admissible function $\delta^{EBC}$ stated in Definition \ref{def:delta_ebc} for the use of the shifted local dampening mechanism. Figure \ref{fig:monotonicity} displays the value of $u(x,r)$ on the x-axis against $\delta^u(x,0,r)$ for the Enron graph database \cite{snapnets}. This example clearly violates strict monotonicity. But it shows a weak monotonicity for the sensitivity function $\delta^EBC$ in the sense that $EBC(x,r)$ is still positively correlated with $\delta^{EBC}(x,t,r)$ with respect to $r$. 

\begin{figure}[h]
	\caption{Correlation between $EBC(x,r)$ and $\delta^{EBC}(x,0,r)$ for EBC metric for Enron Dataset.}
	
    \includegraphics[width=0.5\textwidth]{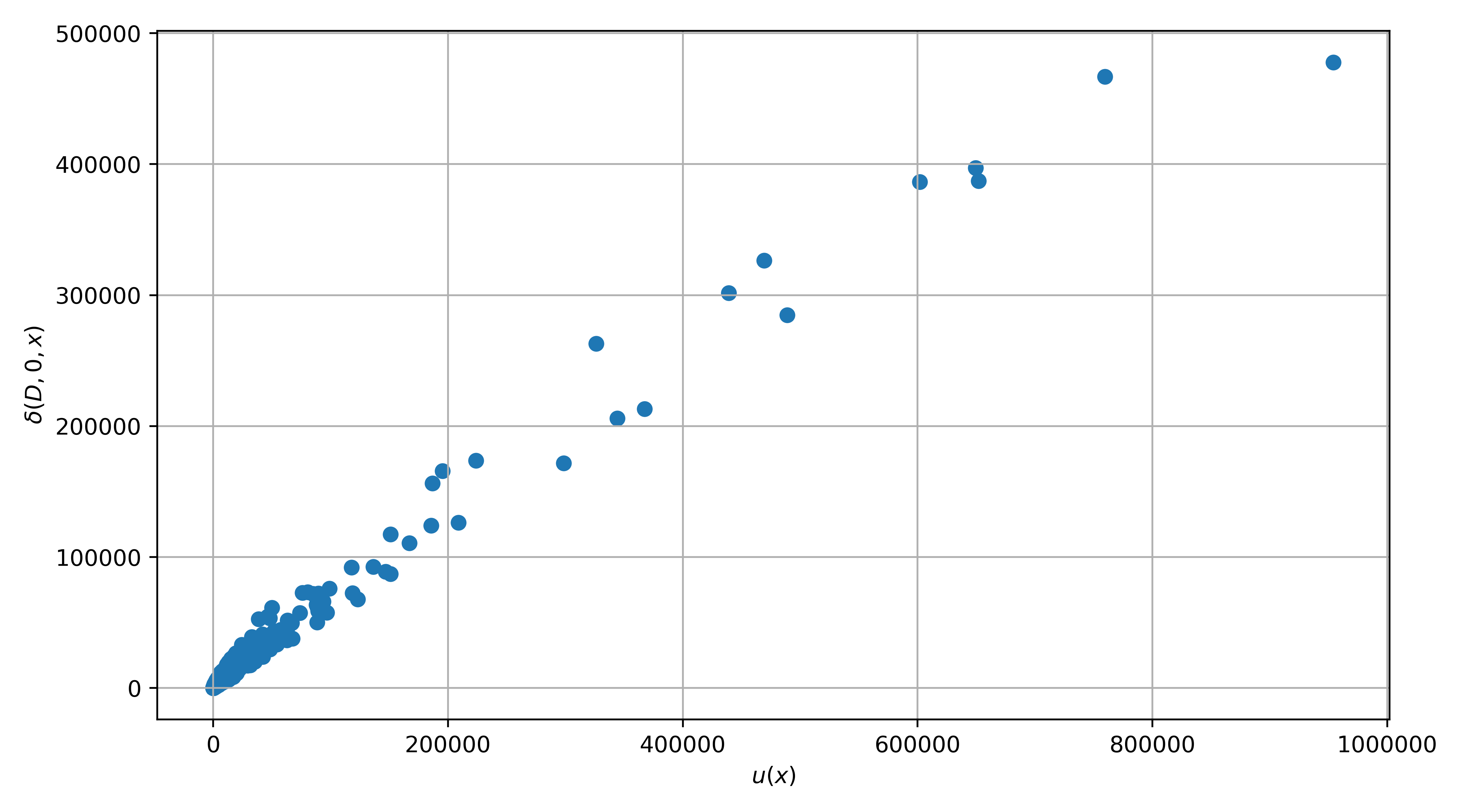}
	
	\label{fig:monotonicity}
\end{figure}

We argue that this kind of behavior is enough for a good performance of the shifted local dampening. Our empirical results corroborates with this argument.

Additionally, Lemma \ref{lemma:local_dampening_accuracy} also suggests that functions that do not exhibit correlation but have lower value than $\Delta u$ also perform well which is the case for local sensitivity $LS^u(x,t)$.

\section{Application 1: Percentile Selection}
\label{section:percentile_selection}

In this Section, we address the percentile selection problem. The goal is to output the label of the element with the closest value to the p-th percentile from a set of real numbers. 

Nissim et al \cite{nissim2007smooth} and McKenna et al \cite{mckenna2020permute} have tackled reduced versions of this problem that work only with medians. Nissim et al \cite{nissim2007smooth} tackled the numeric version of the median selection problem where the task is to return the median value itself and not the label of the median element. McKenna et al \cite{mckenna2020permute} addressed binned version of the median selection problem where the data is binned in $k$ buckets and the goal is to return the median bin. The latter version has a low global sensitivity.


\subsection{Problem Statement}

Given a database $x \in \real^n$ represented as a vector of real numbers $<x_1,\cdots,x_n>$. For simplicity's sake, assume that every database $x$ is ordered such that $x_1 \leq \cdots \leq x_n$. Suppose that all the values lies in $[0,\Lambda]$, $0 \leq x_1 \leq \cdots \leq x_n \leq \Lambda$.

The task is to return the index $i$ where its element $x_i$ is as close as possible to the p-th percentile element $x_k$ where $k = \ceil{\frac{p (n+1)}{100}}$. Note that $\range = \{1, \cdots, n\}$. The utility function for a given index $i$ is the distance from $x_i$ to $x_k$ multiplied by $-1$ so that closer elements have higher utility score:
\begin{definition}
    (Utility function for percentile selection problem).
    $u_p(x,i)=-|x_k-x_i|$, where $k = \ceil{\frac{p (n+1)}{100}}$.
\end{definition}

\subsection{Private Mechanism and Sensitivity Analysis}

This problem is solved by a single call to a non-numeric mechanism. Here we use the exponential mechanism and the permute-and-flip mechanism to compare to the local dampening mechanism. The exponential mechanism and the permute-and-flip mechanism require the computation of the global sensitivity $\Delta u_p$ while the local dampening mechanism requires the computation of the element local sensitivity $u_p$. 

\subsubsection{Global Sensitivity}


The global sensitivity $\Delta u_{p}$ is set by the following scenario: let $p=50$ implying that $k=\ceil{\frac{n+1}{2}}$. Let $x \in \real^n$, for $n>2$, be a database where $x_1=x_2=\cdots=x_{n-1}=0$ and $x_n=\Lambda$. Let $y = <y_1,\cdots,y_n> \in \real^n$ be a neighboring database of $x$ obtained from $x$ by changing the value of $x_{n}$ to $0$, $y_{n}=0$. Thus we have that $y_1,...,y_{n}=0$. In what follows, $u_p(x,n)=\Lambda$ and $u_p(y,n)=0$ as $x_k = y_k = 0$ which implies that $u(x,n) - u(y,n) = \Lambda$.


This happens to be the largest possible $|u_p(x,i) - u_p(y,i)|$ since $|u_p(x,i) - u_p(y,i)| \leq \Lambda$ for any $x,y \in \real^n$ and $i \in [1,n]$. The latter follows from the fact that the distance from $x_i$ to $x_k$ is positive and smaller than $\Lambda$, $0 \leq u_p(x,i) \leq \Lambda$ and $0 \leq u_p(y,i) \leq \Lambda$. 

\begin{lemma}
    (Percentile Selection Global Sensitivity) $\Delta u_p = \Lambda$.
    \label{lemma:percentile_global}
\end{lemma}

\subsubsection{Element local sensitivity}

\textbf{Element local sensitivity at distance 0}. Before calculating the element local sensitivity of $u_p$ at distance $t$, we discuss how to compute the element local sensitivity at distance $0$ $LS^{u_p}(x,0,i)$.

Observe that a naive computation of $LS^{u_p}(x,0,i)$ is infeasible. It needs to iterate over each neighboring database $y$ of $x$ and take $|u(x,i)-u(y,i)|$. The number of neighboring databases is infinite because we can set a given $x_i$ to any real value in $[0,\Lambda]$.

$$ LS^{u_p}(x,0,i) = \max_{y | d(x,y) \leq 1}|u_p(x,i)-u_p(y,i)|$$

Thus we provide a way to efficiently compute $LS^{u_p}(x,0,i)$ in $O(1)$ time complexity (Lemma \ref{lemma:percentile_local_sensitivity_distance_0}).

\begin{restatable}{lemma}{lemmalocalsensitivitydistancezero}
    (Percentile Selection Element Local Sensitivity at distance 0)

    \begin{align*}
        LS^{u_p}(x,0,i)  = & \max(|x_k-x_i|, x_{k+1} - x_{k},  \\
        & x_k - x_{k-1}, p(x,i), q(x,i)), \\
    \end{align*}    
    where
    \begin{equation*}
        p(x,i) = \max
		\begin{cases}
			\Lambda - x_i & \text{if} \ i > k \\
            \Lambda - x_{k+1}& \text{if} \ i = k \\
			\Lambda + x_i - 3 x_k +x_{k+1} & i < k
        \end{cases}, 
    \end{equation*}
    \begin{equation*}
        q(x,i) = \max
        \begin{cases}
            x_i & \text{if} \ i > k \\
            x_{k-1} & \text{if} \ i = k \\
            3 x_k - x_i - x_{k-1} & i < k
        \end{cases},
    \end{equation*}
    $0 \leq x_1 \leq \cdots \leq x_n \leq \Lambda$
    and $k = \ceil{\frac{p (n+1)}{100}}$
    \label{lemma:percentile_local_sensitivity_distance_0}
\end{restatable}



\textbf{Element local sensitivity at distance $t$}. Now we proceed to compute $LS^{u_{med}}(x,t,r)$. 

$$ LS^{u_{med}}(x,t,r) = \max_{y| d(x,y) \leq t} LS^{u_{med}}(y,0,r)$$

Given a distance $t$, our task is to compute $LS^{u}(y,0,r)$ over all $y$ such that $d(x,y) \leq t$. A naive brute force approach would be infeasible since there are infinite databases at distance $t$ from $x$ as discussed previously for the computation of element local sensitivity at distance 0. 

However, there is a small subset, referred as $candidates(x,t,r)$, of $\{y| d(x,y) \leq t\}$ where we can evaluate $LS^{u}(y,0,r)$ only on the databases of $candidates(x,t,r)$ to obtain $LS^{u}(x,t,r)$. The databases $\{y| d(x,y) \leq t\} - candidates(x,t,r)$ are safe to discard, i.e., it exists a database $y \in candidates(x,t,r)$ where $LS^{u}(x,t,r) = LS^{u}(y,0,r)$. So we rewrite the element local sensitivity of $\umed$ as:

\begin{restatable}{lemma}{lemmalocalsensitivitydistancetmedian}
    (Element local sensitivity at distance $t$ for percentile selection)
    \begin{equation*}        
        LS^{u_{med}}(x,t,r) = \max_{candidates(x,t,r)} LS^{u_{med}}(y,0,r).
    \end{equation*}
    \label{lemma:local_sensitivity_distance_t_median}
\end{restatable}

Algorithm \ref{algo:candidates_median} depicts how to compute $candidates(x,t,r)$. The $CopyAndReplace(x, i, v)$ function used in this algorithm makes a copy $x'$ of the input database $x$, sets the element $x'_i$ to the value $v$ and returns the ordered dataset $x'$. The algorithm $candidates(x,t,r)$ returns a subset of only $6$ databases of $\{y| d(x,y) \leq t\}$ that maximizes $\max_{y| d(x,y) \leq t} LS^{u_{med}}(y,0,r)$. A sketch of the proof of Lemma \ref{lemma:local_sensitivity_distance_t_median} can be found in \cite{mythesis}.

\begin{algorithm}[htp]
    \SetAlgoLined\DontPrintSemicolon
    \SetKwFunction{algo}{Candidates(Dataset $x$, distance $t$, range element $r$)}
    \SetKwProg{myalg}{Procedure}{}{}
    \myalg{\algo}{

    \If{$t = 0$}{        
        \KwRet{$(x)$}\;
    }    
    $k$=$\ceil{\frac{n+1}{2}}$ \;
    \If{$t = 1$}{        
        $x_1'$ = $CopyAndReplace(x,r,\Lambda)$\;
        $x_2'$ = $Copy(x)$\;
        $x_3'$ = $CopyAndReplace(x,r,0)$\;
        $x_4'$ = $Copy(x)$\;
        $x_5'$ = $CopyAndReplace(x,k,\Lambda)$\;
        $x_6'$ = $CopyAndReplace(x,k,0)$\;
        \KwRet{$(x_1', x_2', x_3', x_4', x_5', x_6')$}\;
    }
    $x_1, x_2, x_3, x_4, x_5, x_6 = Candidates(x, t-1, r)$ \;
    $x_1'$ = $CopyAndReplace(x_1,k,0)$ \;
    $x_2'$ = $CopyAndReplace(x_2,k,\Lambda)$ \;
    $x_3'$ = $CopyAndReplace(x_3,k,0)$ \;
    $x_4'$ = $CopyAndReplace(x_4,k,\Lambda)$ \;
    $x_5'$ = $CopyAndReplace(x_5,k,0)$ \;
    $x_6'$ = $CopyAndReplace(x_6,k,\Lambda)$ \;
    \KwRet{$(x_1', x_2', x_3', x_4', x_5', x_6')$}\;}{}
    
    \addtocontents{loa}{\protect\addvspace{-9pt}}
    \caption{Candidates Algorithm}
    \label{algo:candidates_median}
\end{algorithm}

\subsection{Experimental Evaluation}

\textbf{Datasets.} We tested most of the datasets from Hay et al \cite{hay2016principled}. Our approach is specially accurate when the local sensitivity is far from the global sensitivity. We report the results for two datasets where the local sensitivity is reasonably smaller from the global sensitivity: PATENT dataset and HEPTH dataset. Also we show one dataset where the local sensitivity is very close to the global sensitivity to show our approach behaves on this scenario: INCOME dataset.

\textbf{Methods.} We test three approaches for the percentile selection problem: i) EMPercentileSelection, the exponential mechanism using global sensitivity; ii) PFPercentileSelection, the permute-and-flip mechanism using the global sensitivity and iii) LDPercentileSelection, the local dampening mechanism using local sensitivity.

\textbf{Evaluation.} We measure the error: $|retrieved\_percentile\_value - true\_percentile\_value|$. For the EMPercentileSelection and LDPercentileSelection methods, we report the expected error and for the PFPercentileSelection we report the mean error over $100,000$ runs. We vary $\epsilon \in [10^{-1},10^2]$ and $p=50,95,99$.

\pgfplotsset{
cycle list={%
{color=black,solid},
{color=blue,dotted, very thick},
{color=red,loosely dashed, very thick},
{color=green,loosely dashed, very thick},
}}
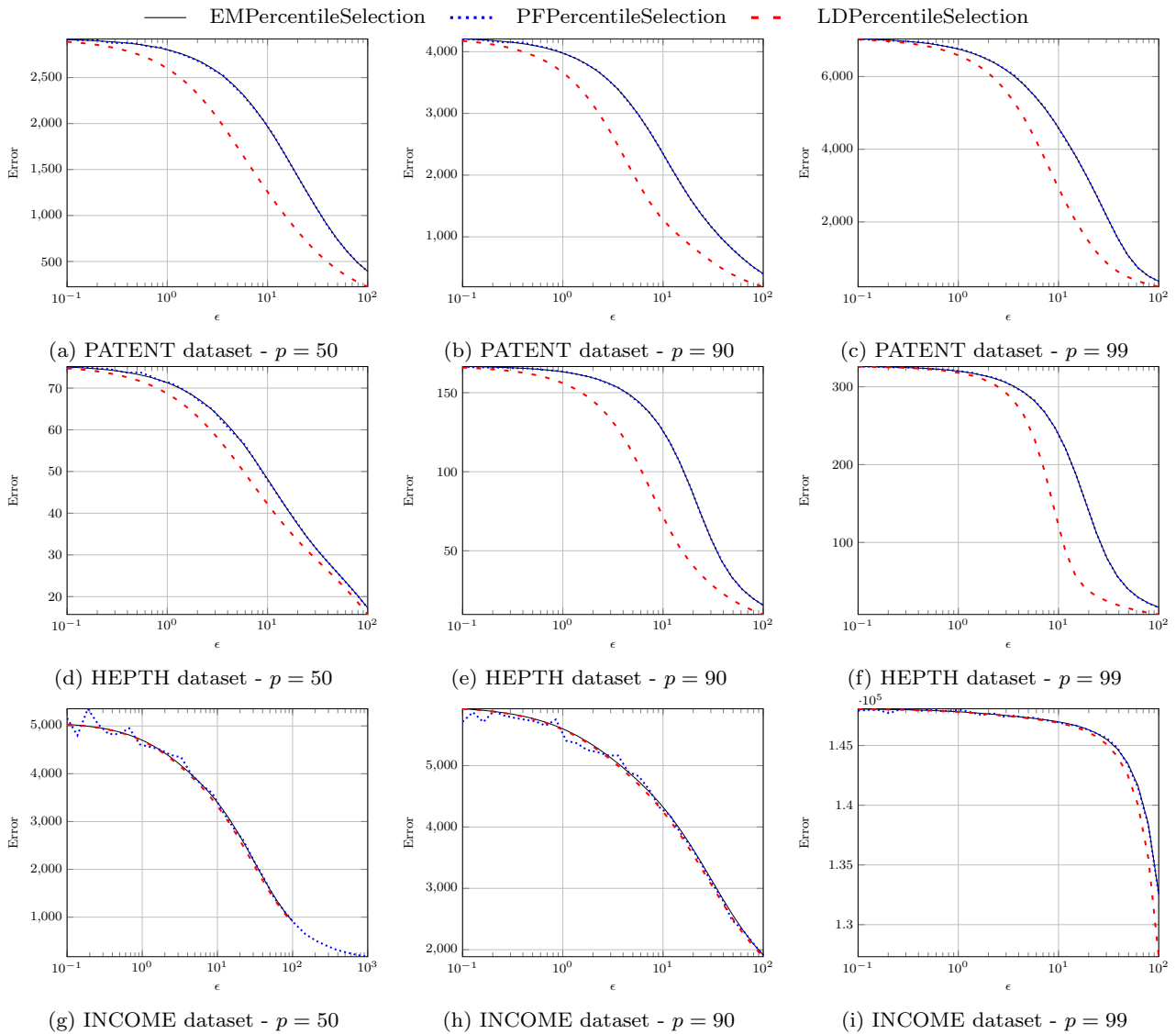
\begin{figure*}[h!]
    \caption{Expected error for the EMPercentileSelection and LDPercentileSelection methods and mean error over $100,000$ runs for PFPercentileSelection where $\epsilon \in [10^{-1},10^2]$. X axis is in log scale.}
    \center
    \begin{tikzpicture}
        \begin{customlegend}[legend columns=3,legend style={align=left,draw=none,column sep=2ex},legend entries={
            EMPercentileSelection, 
            PFPercentileSelection, 
            LDPercentileSelection}]
        \addlegendimage{color=black,solid}
        \addlegendimage{color=blue,solid, dotted, very thick}
        \addlegendimage{color=red,solid, loosely dashed, very thick}        
        \end{customlegend}
    \end{tikzpicture}


\begin{subfigure}[b]{0.32\textwidth}
    \resizebox{\linewidth}{!}{
        \begin{tikzpicture}
            \begin{semilogxaxis}[grid=major,
                enlargelimits=false,
                xlabel={$\epsilon$},
                ylabel={Error},
                xticklabel style={
                    /pgf/number format/fixed,
                    /pgf/number format/precision=5,
                }]

                \addplot table[x=eps,y=error] {./data/percentiles/patent_em_50.dat};

                \addplot table[x=eps,y=error] {./data/percentiles/patent_pf_50.dat};

                \addplot table[x=eps,y=error] {./data/percentiles/patent_ld_50.dat};

            \end{semilogxaxis}
        \end{tikzpicture}            
    }
    \caption{PATENT dataset - $p=50$}
\end{subfigure}
\begin{subfigure}[b]{0.32\textwidth}
    \centering
    \resizebox{\linewidth}{!}{
        \begin{tikzpicture}
            \begin{semilogxaxis}[grid=major,
                enlargelimits=false,
                xlabel={$\epsilon$},
                ylabel={Error},
                xticklabel style={
                    /pgf/number format/fixed,
                    /pgf/number format/precision=5,
                }]

                \addplot table[x=eps,y=error] {./data/percentiles/patent_em_90.dat};

                \addplot table[x=eps,y=error] {./data/percentiles/patent_pf_90.dat};
    
                \addplot table[x=eps,y=error] {./data/percentiles/patent_ld_90.dat};                   
            
            \end{semilogxaxis}
        \end{tikzpicture}            
    }
    \caption{PATENT dataset - $p=90$}
\end{subfigure}
\begin{subfigure}[b]{0.32\textwidth}
    \centering
    \resizebox{\linewidth}{!}{
        \begin{tikzpicture}
            \begin{semilogxaxis}[grid=major,
                enlargelimits=false,
                xlabel={$\epsilon$},
                ylabel={Error},
                xticklabel style={
                    /pgf/number format/fixed,
                    /pgf/number format/precision=5,
                }]

                \addplot table[x=eps,y=error] {./data/percentiles/patent_em_99.dat};

                \addplot table[x=eps,y=error] {./data/percentiles/patent_pf_99.dat};
    
                \addplot table[x=eps,y=error] {./data/percentiles/patent_ld_99.dat};                             

            \end{semilogxaxis}
        \end{tikzpicture}            
    }
    \caption{PATENT dataset - $p=99$}
\end{subfigure}

\begin{subfigure}[b]{0.32\textwidth}
    \resizebox{\linewidth}{!}{
        \begin{tikzpicture}
            \begin{semilogxaxis}[grid=major,
                enlargelimits=false,
                xlabel={$\epsilon$},
                ylabel={Error},
                xticklabel style={
                    /pgf/number format/fixed,
                    /pgf/number format/precision=5,
                }]

                \addplot table[x=eps,y=error] {./data/percentiles/hepth_em_50.dat};

                \addplot table[x=eps,y=error] {./data/percentiles/hepth_pf_50.dat};

                \addplot table[x=eps,y=error] {./data/percentiles/hepth_ld_50.dat};                    

            \end{semilogxaxis}
        \end{tikzpicture}            
    }
    \caption{HEPTH dataset - $p=50$}
\end{subfigure}
\begin{subfigure}[b]{0.32\textwidth}
    \centering
    \resizebox{\linewidth}{!}{
        \begin{tikzpicture}
            \begin{semilogxaxis}[grid=major,
                enlargelimits=false,
                xlabel={$\epsilon$},
                ylabel={Error},
                xticklabel style={
                    /pgf/number format/fixed,
                    /pgf/number format/precision=5,
                }]

                \addplot table[x=eps,y=error] {./data/percentiles/hepth_em_90.dat};

                \addplot table[x=eps,y=error] {./data/percentiles/hepth_pf_90.dat};
    
                \addplot table[x=eps,y=error] {./data/percentiles/hepth_ld_90.dat};                   
            
            \end{semilogxaxis}
        \end{tikzpicture}            
    }
    \caption{HEPTH dataset - $p=90$}
\end{subfigure}
\begin{subfigure}[b]{0.32\textwidth}
    \centering
    \resizebox{\linewidth}{!}{
        \begin{tikzpicture}
            \begin{semilogxaxis}[grid=major,
                enlargelimits=false,
                xlabel={$\epsilon$},
                ylabel={Error},
                xticklabel style={
                    /pgf/number format/fixed,
                    /pgf/number format/precision=5,
                }]

                \addplot table[x=eps,y=error] {./data/percentiles/hepth_em_99.dat};

                \addplot table[x=eps,y=error] {./data/percentiles/hepth_pf_99.dat};
    
                \addplot table[x=eps,y=error] {./data/percentiles/hepth_ld_99.dat};                             

            \end{semilogxaxis}
        \end{tikzpicture}            
    }
    \caption{HEPTH dataset - $p=99$}
\end{subfigure}

\begin{subfigure}[b]{0.32\textwidth}
    \resizebox{\linewidth}{!}{
        \begin{tikzpicture}
            \begin{semilogxaxis}[grid=major,
                enlargelimits=false,
                xlabel={$\epsilon$},
                ylabel={Error},
                xticklabel style={
                    /pgf/number format/fixed,
                    /pgf/number format/precision=5,
                }]

                \addplot table[x=eps,y=error] {./data/percentiles/income_em_50.dat};

                \addplot table[x=eps,y=error] {./data/percentiles/income_pf_50.dat};

                \addplot table[x=eps,y=error] {./data/percentiles/income_ld_50.dat};                    

            \end{semilogxaxis}
        \end{tikzpicture}            
    }
    \caption{INCOME dataset - $p=50$}
\end{subfigure}
\begin{subfigure}[b]{0.32\textwidth}
    \centering
    \resizebox{\linewidth}{!}{
        \begin{tikzpicture}
            \begin{semilogxaxis}[grid=major,
                enlargelimits=false,
                xlabel={$\epsilon$},
                ylabel={Error},
                xticklabel style={
                    /pgf/number format/fixed,
                    /pgf/number format/precision=5,
                }]

                \addplot table[x=eps,y=error] {./data/percentiles/income_em_90.dat};

                \addplot table[x=eps,y=error] {./data/percentiles/income_pf_90.dat};
    
                \addplot table[x=eps,y=error] {./data/percentiles/income_ld_90.dat};                   
            
            \end{semilogxaxis}
        \end{tikzpicture}            
    }
    \caption{INCOME dataset - $p=90$}
\end{subfigure}
\begin{subfigure}[b]{0.32\textwidth}
    \centering
    \resizebox{\linewidth}{!}{
        \begin{tikzpicture}
            \begin{semilogxaxis}[grid=major,
                enlargelimits=false,
                xlabel={$\epsilon$},
                ylabel={Error},
                xticklabel style={
                    /pgf/number format/fixed,
                    /pgf/number format/precision=5,
                }]

                \addplot table[x=eps,y=error] {./data/percentiles/income_em_99.dat};

                \addplot table[x=eps,y=error] {./data/percentiles/income_pf_99.dat};
    
                \addplot table[x=eps,y=error] {./data/percentiles/income_ld_99.dat};                             

            \end{semilogxaxis}
        \end{tikzpicture}            
    }
    \caption{INCOME dataset - $p=99$}
\end{subfigure}

\label{fig:median_experiment}
\end{figure*}

Figure \ref{fig:median_experiment} displays the results. The results show that EMPercentileSelection and PFPercentileSelection mechanisms have very similar error in all the cases. For PATENT dataset, we observe up to $44\%$, $52\%$, $59\%$, for $p=50,90,99$ respectively, of reduction in the error by the LDPercentileSelection in relation to both the EMPercentileSelection and the PFPercentileSelection over the tested values for $\epsilon$. For HEPTH dataset, the error reduction is at most $12\%$, $52\%$, $73\%$, for $p=50,90,99$ respectively.


For the INCOME dataset, we show that the LDPercentileSelection has the same error as the EMPercentileSelection and the PFPercentileSelection in most scenarios. However, we notice a slight reduction of error of at most $3\%$ for $p=99$. This behavior is similar to the other datasets from Hay et al \cite{hay2016principled} not presented here where the local sensitivity is near to the global sensitivity.






\section{Application 2: Influential Node Analysis}
\label{section:influential_node_analysis}

Identifying influential nodes in a network is an important task for social network marketing \cite{ma2008mining}. This analysis has great value for making a more effective marketing campaign since influential nodes have great capacity to diffuse a message through the network. 

\subsection{Problem statement}

The influential node analysis problem is a query over an input graph database $G=(V,E)$ that releases the labels of $k$ nodes that maximize a given influence metric. In this work, we use the Egocentric Betweenness Centrality metric (EBC, Definition \ref{definition:ebc}). Using EBC as an influence measure allows to identify influential nodes that are important in different loosely connected parties.

\begin{definition}(Egocentric Betweenness Centrality (EBC) \cite{everett2005ego,freeman1978centrality,marsden2002egocentric}) 
    $$ EBC(c) = \sum_{u,v \in N_c | u \neq v}  \frac{p_{uv}(c)}{q_{uv}(c)},$$
    \label{def:ebc}
    
    where $N_c = \{ v \in V | \{c,v\} \in E \}$ is the set of neighbors of the central node $c$, $q_{uv}(c)$ is the number of geodesic paths connecting $u$ and $v$ on the induced subgraph $G[N_c \cup \{ c\}]$ and $p_{uv}(c)$ is the number of those paths that include $c$.
    \label{definition:ebc}
\end{definition}

\subsection{Private Mechanism}

We propose \textit{PrivTopk}, a top-k algorithm template that chooses iteratively $k$ nodes that maximize EBC (Algorithm \ref{algo:privtopk}). In each iteration, the algorithm makes a call to a non-numeric mechanism (line 5) that returns a node that maximizes EBC that was not previously chosen.

We experiment with four instances of this algorithm template: i) \textit{EMPrivTopk}, where we replace line $5$ with an exponential mechanism call; ii) \textit{PFPrivTopk}, where we replace line $5$ with a permute-and-flip mechanism call; iii) \textit{LDPrivTopk} where we replace line $5$ with a local dampening call; iv) \textit{SLDPrivTopk} where we replace line $5$ by a shifted local dampening mechanism.

\begin{algorithm}[htp]
    \SetAlgoLined\DontPrintSemicolon
    \SetKwFunction{algo}{PrivTopk(Graph $G = (V,E)$, Privacy Budget $\epsilon$, Integer $k$)}
    \SetKwProg{myalg}{Procedure}{}{}
    \myalg{\algo}{
    $\epsilon' = \epsilon/k$ \;
    $\Omega = \emptyset$ \;
    \For{$j\leftarrow 1$ \KwTo $k$}{        
        $v = MEC(G, \epsilon', EBC, V) $ \tcp{Non-numeric mechanism call}
        $\Omega = \Omega \cup \{ v\}$ \;
    }
    \KwRet{$\Omega$}\;}{}
    
    \addtocontents{loa}{\protect\addvspace{-9pt}}
    \caption{PrivTopk}
    \label{algo:privtopk}
\end{algorithm}


Our algorithm issues $k$ calls to a private mechanism with privacy budget $\epsilon' =\epsilon/k$. By the sequential composition theorem (Theorem \refeq{theorem:sequential_composition}) the total privacy budget consumed in the entire algorithm is $\epsilon' \times k =  (\epsilon / k) \times k = \epsilon$. Thus Algorithm \ref{algo:privtopk} satisfies $\epsilon$-differential privacy.

\subsection{Sensitivity Analysis}

\textbf{Global Sensitivity}. We need to provide the global sensitivity for EBC to the Exponential Mechanism and the permute-and-flip mechanism:

\begin{restatable}{lemma}{lemmaebcglobalsensitivity}
    (EBC global sensitivity). The global sensitivity $\Delta EBC$ for $EBC$ is given by
    $$\Delta EBC = \max{\left( \frac{\Delta(G)(\Delta (G)-1)}{4}, \Delta(G) \right)}$$,
    \label{lemma:global_ebc}
\end{restatable}
where $\Delta(G)$ is the maximum degree of the input graph $G$. In this work, we assume the maximum degree is a public information or that we have an upper bound for it.

\textbf{Element Local Sensitivity}. For the local dampening call, we provide an upper bound to the element local sensitivity using the sensitivity function $\delta^{EBC}$:

\begin{definition}(Sensitivity function $\delta^{EBC}(G,t,v)$). The sensitivity function $\delta^{EBC}$ for $EBC$ is defined as
    $${\textstyle \delta^{EBC}(G,t,v) = \max{\left( \frac{(d^{G}(v)+t)(d^{G}(v)+t-1)}{4}, d^{G}(v)+t \right)}}$$
    \label{def:delta_ebc}
\end{definition}
where $d^{G}(v)$ denotes the degree of $v$ in $G$, i.e., $d^{G}(v)=|N^{G}_v|$. We also show that $\delta^{EBC}(G,t,v)$ is admissible (Lemma \ref{lemma:admissible_delta_ebc}).

\begin{restatable}{lemma}{lemmaebclocalsensitivity}
    The sensitivity function $\delta^{EBC}(G,t,v)$ is admissible.
    \label{lemma:admissible_delta_ebc}
\end{restatable}

In terms of correlation between $EBC(v)$ and $deg^{G}(v)+t$, note that for a given node $v$ with degree $d^{G}(v)$, there are ${{d^{G}(v)}\choose{2}}=(deg^{G}(v) \cdot (deg^{G}(v)-1))/2$ terms in the $EBC$ equation for $v$ (Definition \ref{def:ebc}), i.e., pairs $(u,z) \in N^{G}_{v}$. As each term contributes at most $1$ to $EBC$, it suggests that there is a correlation between $EBC(v)$ and $deg^{G}(v)+t$ and consequently, between $EBC(v)$ and $\delta^{EBC}(G,t,v)$. Empirical observation of the datasets confirmed that correlation. For this reason, the shifted local dampening mechanism call in SLDPrivTopk is suitable.

\subsection{Related Work}
\label{sec:related_work_influential}

The literature provides some work to release statistics on graphs which are presented in this section. Also, we present some work on releasing linear statistics over relational databases that can be used to release EBC and compare experimentally to our approach.

\subsubsection{Differentially Private Graph Analysis}

Kasiviswanathan et al \cite{kasiviswanathan2013analyzing} observed that the sensitivity of many graph problems is a function of the maximum degree of the input graph $G$, so they proposed a generic projection that truncates the maximum degree of $G$. This projection is built upon smooth sensitivity framework \cite{nissim2007smooth} but the target query is answered using the global sensitivity on the truncated graphs which may still be high. Moreover, it satisfies a weaker definition of privacy: $(\epsilon, \delta)$-differential privacy.

There is a number of works that aims to publish subgraph counting as k-triangles, k-stars and k-cliques. Kasiviswanathan et al  \cite{kasiviswanathan2013analyzing} also addresses this problem applying a linear programming-based approach that release those counts for graphs that satisfies $\alpha$-decay.

A new notion of sensitivity called restricted sensitivity was introduced by Blocki et al \cite{blocki2013differentially} to answer subgraph counts. In this setting, the querier may have some belief about the structure of the input graph, so the restricted sensitivity measures sensitivity only on the subset of graphs which are believed to be inputs to the algorithm. However, this work satisfies only $(\epsilon, \delta)$-differential privacy.

Blocki et al \cite{blocki2013differentially} introduced the \textit{ladder functions} to answer to subgraph counts. A latter function is a structure built upon the local sensitivity of the subgraph count. It rank the possible outputs of the subgraph count query and sample a given output using the exponential mechanism with low sensitivity.

The work presented in \cite{karwa2011private} is a directed application of the smooth sensitivity framework \cite{nissim2007smooth} also for answering to subgraph count queries. The authors provide the bounds of the local sensitivity of k-triangles and k-stars and show that they are more accurate than related work.

About centrality metrics. A recent work, Laeuchli et al \cite{laeuchli2021analysis} analyzes three centrality measures on graphs: eigenvector, laplacian and closeness centrality. The result is that releasing those metrics using either the laplace mechanism (based on global sensitivity) or the smooth sensitivity framework (based on local sensitivity) is infeasible. To show that, they demonstrate that the local sensitivity is unbounded or, even it is bounded, it is too large and it results in overwhelming addition of noise. 

To the best of out knowledge, in the literature, none of the works on graph analysis tackles top-k or EBC release.

\subsubsection{Releasing Linear Statistics over Relational Databases}

A body of work is available in the literature on answer linear queries, i.e. queries that can be answered by linear aggregation, over relational databases using SQL unde differential privacy. 

As will be shown in Section \ref{section:privatesql}, the EBC metric can be computed by issuing a set of count SQL queries with cyclic joins and GROUP BY clauses over a graph stored in relational database. Thus we survey works that can possibly answer to this kind of query.

Local sensitivity has been used to answer full acyclic join queries \cite{tao2020computing}. This approach lacks generality as it cannot compute SQL queries with cyclic joins and GROUP BY clauses which it is required for EBC queries. The recursive mechanism \cite{chen2013recursive} can answer linear queries with unrestricted joins with GROUP BY clauses, however it requires the target function $f$ to be monotonic, i.e., inserting a new individual in the database always causes $f$ to increase (or always decrease). This condition is not satisfied by EBC.


PrivateSQL \cite{kotsogiannis2019privatesql} is an approach that can answer linear queries with cyclic joins and correlated subqueries with GROUPBY clauses. 



This approach requires a representative workload $Q$ as input, a primary relation in the database. The representative workload is used in the VSelector module to identify a set of views over the base relations that support the analyst queries and then generates a set of views that can answer all the the analyst's queries. The VRewrite module rewrite using truncation operators and semijoin operators to bound the sensitivity. Therefore, the sensitivity for each query is computed in the SensCalc module and a fraction of the budget is given to each query in the BudgetAllocator module. The PrivSynGen generates a synopsis for each view. A synopsis is a compact representation of a view that capture important information from it and can approximate the answer of a given query on the view. The private synopses are released under differential privacy and queries can be executed over them without using privacy budget. The synopsis generation is based on non-negative least squares inference \cite{li2015matrix}. The sensitivity computation for each view is based on Flex \cite{johnson2018towards} augmented with truncation operators.

Then when a new query is issued to the system, this query is rewritten as linear queries over the private views' synopsis in the MapQuery module. Since graph databases can be modeled as a table Node(id) and a table Edge(source, target), we carry out an experimental evaluation in Section \ref{section:privatesql}.

\subsection{Experimental Evaluation}

\textbf{Datasets.} We use three real-world graph datasets: 1) \textit{Enron} is a network of email communication obtained from around half million emails. Each node is an email address and an edge connects a pair of email addresses that exchanges emails ($|V|=36,692$ , $|E|=183,831$ and $\Delta(G) = 1,383$); 2) \textit{DBLP} is a co-authorship network where two authors (nodes) are connected if they published at least one paper together ($|V|=317,080$, $|E|=1,049,866$ and $\Delta(G) = 343$); 3) \textit{Github} is a network of developers with at least $10$ stars on the platform. Developers are represented as nodes and an edge indicates that two developers follow each other ($|V|=37,700$, $|E|=289,003$ and $\Delta(G) = 9,458$). $V$ is the set of vertices, $E$ is the set of edges and $\Delta(G)$ is the maximum degree of a graph $G$. All datasets can be found on Stanford Network Dataset Collection \cite{snapnets}.

\textbf{Methods}. We compare the four versions of PrivTopk (algorithm \ref{algo:privtopk}): 1) EMPrivTopk, using the exponential mechanism, 2) PFPrivTopk, using permute-and-flip mechanism,  3) LDPrivTopk using local dampening mechanism and 4) SLDPrivTopk using shifted local dampening mechanism.

\textbf{Evaluation.} We evaluate the accuracy by the percentage of common nodes to the retrieved top-k set and the true top-k set, i.e., $(|\text{retrieved\_topk} \cap \text{true\_topk}|)/k$. We report the mean accuracy in $100$ simulations. We set $k \in \{ 5,10,20\}$ and a range for privacy budget $\epsilon \in [10^{-3}, 10^{4}]$.

Figure \ref{fig:topk} displays the results. We first note that the global sensitivity based approaches, EMPrivTopk and PFPrivTopk, exhibit low accuracy for low values of $\epsilon$. This is due to the high global sensitivity: $\Delta EBC = 22,361,076.5$ for github dataset, $\Delta EBC = 477,826.5$ for DBLP dataset and $\Delta EBC = 29,326.5$ for Enron dataset. Also, the LDPrivTopk algorithm suffers from the inversion problem (Section \ref{sec:inversion_problem}) while SLDPrivTopk could exploit the correlation between $EBC$ and $\delta^{EBC}$ to fix this problem. 

\input{./images/figure_topk5.tex}

We observe a clear pattern where the methods perform worse as $k$ grows. This is explained by the fact that each call to the non-numeric mechanism uses $\epsilon/k$ of the total privacy budget $\epsilon$ (Algorithm \ref{algo:privtopk}). Thus, larger $k$ implies that less of the privacy budget is used in each non-numeric mechanism call which hurts accuracy. 


For SLDPrivTopk, we note that we need different level of privacy budget for each dataset reasonable accuracy. This is explained by a number of factors. For Github dataset, the distribution of EBC is heavy tailed thus the nodes with high EBC have a higher probability to be correctly picked with low privacy budget. On the other hand, for the DBLP dataset we need need more privacy budget as it has roughly $10$ times more nodes than the other datasets, which dilute the probability of the nodes with higher EBC.

Our approach SLDPrivTopk achieves the same level of accuracy with privacy values $3$ to $4$ orders of magnitude less than EMPrivTopk and $2$ to $3$ orders of magnitude less than PFPrivTopk.

\subsubsection{Comparison to PrivateSQL}
\label{section:privatesql}

We carry out an experimental comparison of local dampening mechanism to the PrivateSQL on Influential Node Analysis problem. For this application, PrivateSQL is not particularly scalable (as discussed further) so we performed the experiments with smaller datasets and test with fewer values of privacy budgets.

PrivateSQL addresses the Influential Node Analysis problem by computing the counts $q_{uv}(c)$ and $p_{uv}(c)$ for all $u,v \in N_c$ (Definition 11) to compute EBC and takes the top-k nodes with highest EBC score. Note that it considers only for the terms $p_{uv}(c)/q_{uv}(c)$ where the distance from $u$ to $v$ in $G[N_c \cup \{ c\}]$ is $2$ which is the maximum possible distance as $u, v \in N_c$. If their distance is $1$ (i.e. $u$ and $v$ are neighbors), the term $p_{uv}(c)$ is $0$ since the geodesic path of length $1$ from $u$ to $v$ ($u,v \neq c$) cannot contain $c$. Thus, for PrivateSQL, we pose private queries $\mathcal{Q}(u,v,c)$ \nottoggle{vldb}{(see Appendix \ref{ap:privateSQL})} that returns 1) $1$ if $u$ and $v$ are neighbors, 2) $0$ if the distance from $u$ to $v$ is larger than $2$ and 3) $q_{uv}(c)$, otherwise.

Therefore, when $\mathcal{Q}(u,v,c)$ is not equal to $0$, it means that the term $p_{uv}(c)/q_{uv}(c)$ should be accounted in the EBC definition. In that case, we obtain a noisy estimate for $q_{uv}(c)$ from $\mathcal{Q}(u,v,c)$. A noisy estimate for $p_{uv}(c)$ can be derived from noisy $q_{uv}(c)$ by setting $p_{uv}$ to $0$ if $q_{uv}(c)=0$ or to $1$ if $q_{uv}(c) > 0$. The rationale is that exactly one of the paths of length $2$ from $u$ to $v$ contains $c$ as $u, v \in N_c$.

The set $N_c$ is itself private information. Hence, to compute $EBC(c)$ for every $c \in V(G)$ we need to compute $\mathcal{Q}(u,v,c)$ for every $u,v \in V(G)$. This results in a total number of $O(n^3)$ queries which poses a scalability problem. For this reason, we perform experiments with samples $S$ of the graphs which are obtained by choosing a node sample $S_n$ in breadth-first search fashion with a random seed node and then we set $S = G[S_n]$.

Table \ref{table:privatesql} displays the mean accuracy for $10$ runs on $10$ sample graphs with $|S_n|=50$ nodes with $k \in \{1,2,3\}$ for each $\epsilon \in \{0.1,0.5,1.0,5.0,10.0\}$. We compare the best local dampening based algorithm SLDPrivTopk (PTK) with the PrivateSQL based approach (PSQL).

PrivateSQL approach generated one private view for each node in the graph. Thus, the privacy budget needs to be divided by the number of nodes $n$ which implies that accuracy is hurt as $n$ grows. Moreover, the sensitivity for each view is high, e.g, sensitivity is $1448$ when $\Delta(G)=10$. This entails in a poor performance for the PrivateSQL based approach.

\begin{table}[ht]
    \caption{Mean accuracy for SLDPrivTopk (PTK) and PrivateSQL (PSQL) over $10$ runs on $10$ sample graphs with $50$ nodes with $k\in\{1,2,3\}$ for each $\epsilon \in \{0.1,0.5,1.0,5.0,10.0\}$.}

    \begin{tabular}{c|c|c|c|c|c|c|}
        \cline{2-7}
        & \multicolumn{2}{c|}{Enron} & \multicolumn{2}{c|}{DBLP} & \multicolumn{2}{c|}{Github} \\ \hline
        \multicolumn{1}{|c|}{$\epsilon$} & PTK & PSQL & PTK & PSQL & PTK & PSQL \\ \hline
        \multicolumn{1}{|c|}{$0.1$} & 0.06 & 0.01 & 0.05 & 0.02 & 0.07 & 0.02 \\ \hline
        \multicolumn{1}{|c|}{$0.5$} & 0.45 & 0.01 & 0.27 & 0.02 & 0.49 & 0.02 \\ \hline
        \multicolumn{1}{|c|}{$1.0$} & 0.60 & 0.16 & 0.44 & 0.05 & 0.69 & 0.02 \\ \hline
        \multicolumn{1}{|c|}{$5.0$} & 0.84 & 0.20 & 0.87 & 0.11 & 0.86 & 0.03 \\ \hline
        \multicolumn{1}{|c|}{$10.0$} & 0.88 & 0.21 & 0.92 & 0.21 & 0.91 & 0.07 \\ \hline
    \end{tabular}            
    
    \label{table:privatesql}
\end{table}

\section{Application 3: ID3 Decision Tree Induction}
\label{section:decision_tree_induction}

Classification based on decision tree is an important tool for data mining \cite{kotsiantis2007supervised}. Specifically, decision trees are a set of rules that are applied to the input attributes to decide to which class a given instance belongs. 


	

Creating a decision tree manually is a burden. Thus many approaches for automatically building decision trees were proposed. One of the most known tree induction algorithms is the ID3 algorithm \cite{quinlan1986induction}. A tree induction algorithm receives a dataset and outputs a decision tree. 


    

The ID3 algorithm \cite{quinlan1986induction} starts with the root node containing the original set. Then the algorithm greedly chooses an unused attribute to split the set and generate child nodes. The selection criterion is Information Gain (IG), given by the entropy before splitting minus the entropy after splitting. It expresses how much entropy was gained after the split. This process continues recursively for the child node until splitting does not reduce entropy or the maximum depth is reached.

\subsection{Problem Statement}

A decision tree induction algorithm takes as input a dataset $\TT$ with attributes $\AAA = \{A_1, \dots, A_d \}$ and a class attribute $C$ and produces a decision tree. The task is to build a decision tree in a differentially private manner. Specifically, we base our approach in one of the most known tree induction algorithms, the ID3 algorithm. 

The notation for this section is summarized in Table \ref{table:notation_tree}. All logarithms are in base $2$. When it is clear from the context, we drop the superscript $\TT$ from the notations.

\begin{table}[ht]
    \caption{Notation table for private decision tree induction}
    \centering
        
        \begin{tabular}{|l|l|}
            \hline
            \textbf{Variable} & \textbf{Definition}   \\ \hline    
            $IG$ & Information Gain  \\ \hline
            $\TT$ & Dataset \\ \hline
            $\AAA$ & Attribute set \\ \hline
            $A_i$ & i-th attribute  \\ \hline
            $C$ & Class attribute  \\ \hline
            $\ttt^\TT$ & Cardinality of a dataset $\TT$: $\ttt^\TT = |\TT|$ \\ \hline    
            $r_A$ & Values of an attribute $A$ in a record $r$ \\ \hline
            $r_C$ & Values of the class attribute $C$ in a record $r$ \\ \hline
            $\TT_j^A$ & \begin{tabular}{@{}l@{}}Set of records $r \in \TT$ where attribute $A$ \\ takes value $j$: $\TT_j^A =\{ r \in \TT : r_A = j \}$\end{tabular}\\ \hline
            $\tjx{\TT}$ & Cardinality of $\TT_j^A$: $\tjx{\TT} = |\TT_j^A|$\\ \hline
            $\ttt_c^{\TT}$ & \begin{tabular}{@{}l@{}} Number of records $r \in \TT$ where class \\ attribute $C$  takes value $c$: \\ $\ttt_c^{\TT} = |r \in \TT: r_C = c|$ \end{tabular} \\ \hline
            $\tjcx{\TT}$ & \begin{tabular}{@{}l@{}} Number of records $r \in \TT$ where attribute $A$ \\ takes value $j$ and class attribute $C$ takes \\ value $c$:  $\tjcx{\TT} = |r \in \TT : r_A = j \wedge r_c = c|$ \end{tabular} \\ \hline
        \end{tabular}

    \label{table:notation_tree}
\end{table}


\subsection{Related Work}
\label{section:related_work_decision_tree_induction}

Blum et al \cite{blum2005practical} presented the SuLQ framework that provides primitives for data mining algorithms. As an application, they introduced a differentially private adaptation for ID3 algorithm where it computes the information gain based on the noisy counts provided by SuLQ. This approach applies two operators from SuLQ: 1) \textbf{NoisyCount} that uses Laplace mechanism to return private estimate of a count query, and 2) \textbf{Partition} that splits the dataset in disjoint subsets so that the privacy budgets for the queries over each subset do not sum up (parallel composition, Theorem \ref{theorem:parallel_composition}) meaning that we can make a more efficient use of the privacy budget. A major drawback of this algorithm is that it requires to query the noisy counts (via \textbf{NoisyCounts}) for each attribute, so the privacy budget needs to be split to those queries. This entails in a small budget per query and in a larger noise magnitude.

To overcome this, Friedman et al \cite{friedman2010data} introduced a variation of the SuLQ algorithm (Algorithm \ref{algo:DiffPID3}). This algorithm replaces the many \textbf{NoisyCount} calls for a single call to the exponential mechanism. Line $12$ (Algorithm \ref{algo:DiffPID3}) is the call for the exponential mechanism that was previously several calls to \textbf{NoisyCount}.

The procedure $Build\_DiffPID3$ in algorithm \ref{algo:DiffPID3} starts by checking for the construction of a leaf node (Line 5-7). It verifies if there the dataset has any remaining attribute and if there is enough instances to make new splits. In lines 8-10, it partitions the data based on the class attribute $C$, it privately queries the count of each partition and creates a leaf with the class label of the largest partition. The Lines 12-15 creates a new decision rule recursively. It starts by privately choosing the attribute with largest IG using the exponential mechanism then it partition the dataset using the chosen attribute and call $Build\_DiffPID3$ with each partition to create the sub-trees.


\begin{algorithm}[ht]
    \SetAlgoLined\DontPrintSemicolon
    \SetKwFunction{algo}{GlobalDiffPID3(Dataset $\TT$, Attribute Set $\AAA$, Class attribute $C$, Depth $d$, Privacy Budget $\epsilon$)}
    \SetKwProg{myalg}{Procedure}{}{}
    \myalg{\algo}{
    $\epsilon' = \epsilon/(2(d+1))$ \;
    \KwRet{\upshape Build\_DiffPID3($\TT$, $\AAA$, $C$, $d$, $\epsilon'$)}\;}{}

    \SetKwFunction{algoo}{Build\_DiffPID3($\TT$, $\AAA$, $C$, $d$, $\epsilon$)}
    \myalg{\algoo}{
    $t = \max_{A \in \AAA} |A|$  \;
    $N_{\TT} = \text{NoisyCount}_{\epsilon}(\TT)$ \;
    \If{$\AAA = \emptyset$ or $d=0$ or $\frac{N_{\TT}}{t|C|} < \frac{\sqrt{2}}{2}$}{
        $\TT_c = \text{Partition}(\TT, \forall c \in C : r_c = c)$ \;
        $\forall c \in C : N_c  = \text{NoisyCount}_{\epsilon} (\TT_{c})$ \;
        \KwRet{\upshape a leaf labeled with $\arg\max_c(N_c)$}\;
    }
    $\bar{A} = \mathcal{E}(\TT,\epsilon,IG,\AAA)$ \tcp{Exp. mechanism call}
    $\TT_i = \text{Partition}(\TT, \forall i \in \bar{A} : r_{\bar{A}} = i)$ \;
    $\forall i \in \bar{A} : \text{Subtree}_i = \text{Build\_DiffPID3}(\TT_i, \AAA \setminus \bar{A}, C, d-1, \epsilon) $ \;
    \KwRet{\upshape a tree with a root node labeled $\bar{A}$ and edges labeled $1$ to $\bar{A}$ each going to $\text{Subtree}_i$}\;}{}

    \addtocontents{loa}{\protect\addvspace{-9pt}}
    \caption{GlobalDiffPID3}
    \label{algo:DiffPID3}
\end{algorithm} 

To the best of our knowledge, our work is the first to apply local sensitivity to single trees which was an open question pointed out in a recent survey on private decision trees \cite{fletcher2019decision}.

Many other works address the private construction of random forests \cite{fletcher2015bdifferentially,fletcher2015differentially,fletcher2017differentially,jagannathan2009practical,patil2014differential,rana2015differentially}. Interestingly, local sensitivity was used for building random forests \cite{fletcher2015differentially,fletcher2017differentially} using smooth sensitivity. This shows a promising future direction of our work which is applying local dampening to construct random forests. 

\subsection{Private Mechanism}

We use the algorithm GlobalDiffPID3 \cite{friedman2010data} (Algorithm \ref{algo:DiffPID3}) as a template. We aim to adapt it for the use of the local dampening mechanism and to the shifted local dampening mechanism producing the \textit{LocalDiffPID3} and \textit{ShiftedLocalDiffPID3}, respectively. In the following, we make a discussion about the split criterion, the global sensitivity of the split criterion for the exponential mechanism and the element local sensitivity for the local dampening.

\textbf{Split criterion}. In this work, we address one of the most traditional split criterion, \textit{information gain} (IG). It is given by the entropy of the class attribute $C$ in $\TT$ minus the obtained entropy of $C$ splitting the tuples according to an attribute $A \in \AAA$. 


\begin{equation*}
    IG(\TT, A) = H_C(\TT) - H_{C|A}(\TT),
\end{equation*}
where $H_C(\TT)$ is entropy with respect to the classe attribute $C$

\begin{equation*}
    H_C(\TT) = -\sum_{c \in C} \frac{\ttt_c}{\ttt} \log \frac{\ttt_c}{\ttt},
\end{equation*} 
and $H_{C|A}(\TT)$ is the entropy obtained by splitting the instances according to attribute $A$

\begin{equation*}
    H_{C|A}(\TT) = \sum_{j \in A} \frac{\ttt_j^A}{\ttt} . H_C(\TT^A_j).
\end{equation*}

Since $H_C(\TT)$ does not depend on $A$, we can further simplify the utility function $IG$:

\begin{align}
    IG(\TT, A) &= - \ttt . H_{C|A}(\TT) \\
    & = - \sum_{j \in A} \sum_{c \in C} \tjc . \log(\frac{\tjc}{\tj}). \label{line:utility_ig}
\end{align}

\textbf{Global sensitivity}. The exponential mechanism requires the computation of the global sensitivity for $IG$. It is given by $\Delta IG = \log(N+1) + 1/ \ln 2$ \cite{friedman2010data} where $N$ is the size of the dataset $\TT$. The global sensitivity case can be achieved by $\TT$ and $\TT'$ where:
\begin{enumerate}
    \item $\TT$ has all tuples with values for $A$ equal to a single value $j \in A$ and all tuples class attribute $C$ are set to a value different from a given value $c \in C$ (i.e. $\tj = \ttt$ and $\tjc = 0$);
    \item  $\TT'$ is obtained from $\TT$ by adding a tuple $r$ where $r_A = j$ and $r_C = c$.
\end{enumerate}

\textbf{Element local sensitivity}. In our experiments, we observed that this mentioned case for the global sensitivity is not frequent in real datasets. For those datasets, a local measurement of the sensitivity can be about one order of magnitude lower than the global sensitivity. To this matter, we replace line $12$ of algorithm \ref{algo:DiffPID3} for a local dampening mechanism call producing the algorithm LocalDiffPID3. 

\textbf{Element Local Sensitivity at distance 0}. To use the local dampening mechanism, we provide means to efficiently compute the element local sensitivity for $IG$ (Lemma \ref{lemma:local_Infogain}). The element local sensitivity at distance $t$ measures $LS^{IG}(\TT', 0, A)$ for all datasets $\TT'$ such that $d(\TT, \TT') \leq t$. We first show how to obtain $LS^{IG}(\TT', 0, A)$:

\begin{restatable}{lemma}{lemmalocalinfogain}
    (Element local sensitivity at distance 0 for IG). Given a dataset $\TT$ and the attribute set $A$, $LS^{IG}(\TT,0,A)$ produces the element local sensitivity for IG at distance $0$:
    \begin{equation*}
        LS^{IG}(\TT,0,A) = \max_{j \in A, c \in C} h(\tjx{\TT}, \tjc{\TT}),
    \end{equation*}
    where
    \begin{align*}
        h(a,b) &= \max(f(a)-f(b), g(b)-g(a) ),\\
        g(x) &= x.log((x-1)/x) - log(x-1),\\
        f(x) &= x.log((x+1)/x) + log(x+1).
    \end{align*}
    \label{lemma:local_Infogain0}
\end{restatable}
Assume that $g(x)=0$ for $x \leq 1$ and $f(x)=0$ for $x \leq 0$. Note that, the expression $g(\tjcx{\TT})-g(\tjx{\TT})$ measures the impact of the removal of a tuple $r$ such that $r_A = j$ and $r_C = c$ and the expression $f(\tjx{\TT})-f(\tjcx{\TT})$ measures the addition of tuple $r$. Thus to obtain $LS^{IG}(\TT,0,A)$, we need to measure, for all $j \in A$ and $c \in C$, the addition or removal of the tuple $r$ where $r_A = j$ and $r_C = c$, i.e. $h(\tjx{\TT},\tjcx{\TT}) = \max(f(\tjx{\TT})-f(\tjcx{\TT}), g(\tjcx{\TT})-g(\tjx{\TT}))$.

\sloppy \textbf{Element local sensitivity at distance t}. We use a similar idea to compute $LS^{IG}(\TT,t,A)$. $LS^{IG}(\TT,t,A)$ searches for the largest $LS^{u}(\TT',0,A)$  over all datasets $\TT'$ where $d(\TT,\TT') \leq t$:

\begin{align*}
    LS^{IG}(\TT,t,A) & = \max_{\TT'| d(\TT,\TT') \leq t} LS^{u}(\TT',0,A) \\
    & =  \max_{c \in C, j \in A} \max_{\TT'| d(\TT,\TT') \leq t} h(\tjx{\TT'}, \tjcx{\TT'}). \\
\end{align*}

Exhaustively iterating over all $\TT'$ to compute $h(\tjx{\TT'}, \tjcx{\TT'})$ is not feasible since the number of datasets $\TT'$ grows exponentially with respect to $t$. However, we can restrict the number of evaluations of $h$ by discarding some of the datasets $\TT'$.

To this end, we introduce the algorithm \textit{Candidates}$(\TT, t, j, c)$ (Algorithm $\ref{algo:candidates}$) that produces a subset of the set of the pairs ($\tjx{\TT'}$, $\tjcx{\TT'}$) of all datasets $\TT'$ such that $d(\TT, \TT') = t$, i.e.,  \textit{Candidates}$(\TT, t, j, c) \subseteq \{(\tjx{\TT'}, \tjcx{\TT'}) \; | \; d(\TT,\TT') = t \}$.

\begin{algorithm}[htp]
    \SetAlgoLined\DontPrintSemicolon
    \SetKwFunction{algo}{Candidates(Dataset $\TT$, distance $t$, attribute value $j$, class attribute value $c$)}
    \SetKwProg{myalg}{Procedure}{}{}
    \myalg{\algo}{

    \If{$t = 0$}{        
        \KwRet{$\{ (\tj, \tjc) \}$}\;
    }    
    $candidates = \emptyset$ \;
    \For{each pair $(a, b) \in $Candidates($\TT, t-1, j, c$) }{
        \If{$a > 0$ and $b > 0$}{
            $candidates = candidates \cup \{ (a-1, b-1) \}$ \;
        }    
        \If{$a < \ttt$}{
            $candidates = candidates \cup \{ (a+1, b) \}$ \;
        }
    }
    \KwRet{$candidates$}\;}{}
    
    \addtocontents{loa}{\protect\addvspace{1pt}}
    \caption{Candidates Algorithm}
    \label{algo:candidates}
\end{algorithm} 

The Candidates algorithm has two important properties: 
\begin{enumerate}
    \item \textit{Candidates}$(\TT, t, j, c)$ contains the pair $(\tjx{\TT'}$, $\tjcx{\TT'})$ such that $h(\tjx{\TT'}$, $\tjcx{\TT'})$ is maximum, i.e., $h(\tjx{\TT'}$, $\tjcx{\TT'}) = \max_{\TT'| d(\TT,\TT') = t} h(\tjx{\TT'}, \tjcx{\TT'})$ (Lemma \ref{lemma:local_Infogain});
    \item It is cacheable, when computing $D_{IG, \delta^{IG}}$, we evaluate $LS^{IG}(\TT',t,A)$ several times in increasing order of $t$ then one can cache calls to $Candidates(\TT, t-1, j, c)$ to execute $Candidates(\TT, t, j, c)$ (line 6) efficiently. 
\end{enumerate}

Thus $LS^{IG}(\TT,t,A)$ is given by:


\begin{restatable}{lemma}{lemmalocainfogaint}
    (Element local sensitivity at distance t for IG) Given an input table $\TT$, a distance $t$ and an attribute set $A$, $LS^{IG}(\TT,t,A)$ produces the element local sensitivity at distance $t$ for $IG$.
    \begin{equation*}
        LS^{IG}(\TT,t,A) = \mathop{\max_{j \in A, c \in C,}}_{0 \leq t' \leq t} \; \max_{(a,b) \in Candidates(\TT, t', j, c)} h(a,b).
    \end{equation*}
    \label{lemma:local_Infogain}
\end{restatable}
In turn, the computation of $LS^{IG}(\TT,t,A)$ is also cacheable. One can store a previous call to $LS^{IG}(\TT,t-1,A)$ to compute $LS^{IG}(\TT,t,A)$ as:

\begin{align*}
    LS^{IG}(\TT,t,A) = &\max\vast(\mathop{\max_{j \in A, c \in C,}}_{(a,b) \in Candidates(\TT, t, j, c)} h(a,b), \\
    &  \quad \quad \quad LS^{IG}(\TT,t-1,A)\vast)
\end{align*}

In the datasets used in our experiments, $LS^{IG}$ and $IG$ shows correlation. Consequently, we also replace the exponential mechanism call on line $12$ in algorithm \ref{algo:DiffPID3} by a call to the Shifted local dampening with $LS^{IG}$, which produces the ShiftedLocalDiffPID3.


\textbf{Continuous attributes support}. An important feature introduced in C4.5 algorithm \cite{salzberg1993c4} is the support for continuous attributes. To support continuous attributes, we use a simpler approach that performed well in our experiments and it is also used in \cite{friedman2010data}. We discretize the continuous attributes into $b$ evenly spaced bins on the dataset and use them as discrete attributes.

\subsection{Experimental Evaluation}

\textbf{Datasets.} We use of three tabular datasets: 1) \textit{National Long Term Care Survey (NLTCS)} \cite{manton2010national} is a dataset that contains $16$ binary attributes of $21,574$ individuals that participated in the survey, 2) \textit{American Community Surveys (ACS)} dataset \cite{series2015version} includes the information of $47,461$ rows with $23$ binary attributes obtained from 2013 and 2014 ACS sample sets in IPUMS-USA and 3) \textit{Adult} dataset \cite{blake1998uci} contains $45,222$ records (excluding records with missing values) with $12$ attributes where $8$ are discrete and $4$ are continuous.

\textbf{Methods}. We compare the three versions of the DiffPID3 (algorithm \ref{algo:DiffPID3}): 1) \textit{GlobalDiffPID3} using the exponential mechanism, 2) \textit{LocalDiffPID3} using the local dampening mechanism and 3) \textit{ShiftedLocalDiffPID3} using the local dampening mechanism with shifting.

\textbf{Evaluation.} We evaluate the accuracy of the approach by reporting the mean accuracy across the $10$ runs of a $10$-fold cross validation. We set $depth \in \{2,5\}$ and $\epsilon \in \{0.01, 0.05, 0.1, 0.5, 1.0, 2.0\}$.

\input{./images/figure_tree1.tex}

Figure \ref{fig:tree_exp} presents the results. We observe that the LocalDiffPID3 improves on the GlobalDiffPID3 in almost every privacy budget value, up to $5\%$. While ShiftedLocalDiffPID3 improves a little more in relation to the LocalDiffPID3, up to $1\%$.

For the ACS dataset, the inversion problem (Section \ref{sec:inversion_problem}) appears. Specifically, for $depth=2$, the second and the third attributes with largest IG become the third and second attributes, respectively, with larger Dampened IG. As a consequence, as $\epsilon$ grows, LocalDiffPID3 tends to pick the first and the third attributes with largest Information which is sub-optimal. ShiftedLocalDiffPID3 is less prone to suffer from the inversion problem in larger depths, i.e. depth=5, since it can pick, in deeper levels, those attributes that loose rank (see Figure \ref{fig:tree_acs_d5}). ShiftedLocalDiffPID3 improves at most $8\%$ on GlobalDiffPID3.

Fot the Adult Dataset, the LocalDiffPID3 improves a little over the GlobalDiffPID3. However, ShiftedLocalDiffPID3 improves over GlobalDiffPID3 up to $4\%$.

\section{Conclusion}
In this paper, we introduced the Local Dampening mechanism, a novel framework to provide Differential Privacy for non-numeric queries using local sensitivity. We have shown that using local sensitivity on non-numeric queries reduces the magnitude of the noise added to achieve Differential Privacy which makes the answer of those queries more useful. We evaluated our approach on three applications: 1) Percentile Selection problem that our approach outperformed global sensitivity based approaches; 2) Influential node analysis which benefited greatly from the use of local sensitivity; and 3) Decision Tree induction which improves on approaches that use the exponential mechanism for this task based on Information Gain.

Our paper has laid the foundations for providing DP for non-numeric queries using local sensitivity. We have achieving a deeper theoretical understanding of the Local Dampening mechanism to understand the class of problems for which it can provide significant gains over the Exponential mechanism. There are many interesting directions of future work. First, as discussed in Section 7, any problem in the literature that has used the Exponential mechanism for non-numeric queries to guarantee DP is a candidate problem that could potentially benefit from using our local dampening mechanism instead, and worthy of future work. Second, it would be interesting to tackle other graph influence/centrality metrics for Influential Node analysis, such as PageRank. Third, applying the local dampening mechanism for private evolutionary algorithms is a promising future direction.

\bibliographystyle{spmpsci}      
\bibliography{references}   

\nottoggle{vldb}{
  \appendix                                     
\section{Proofs}


\subsection{Proof of Lemma \ref{lemma:admissiblels}}

\lemmalocaladmissible*

\begin{proof}
                
	We need to satisfy the two conditions of the admissibility of functions.
	\begin{enumerate}
		\item $LS^u(x, 0, r) \geq LS^u(x, 0, r)$ 
		\item Since $\{ y | d(x,y) \leq t\} \subset \{ y| d(x',y) \leq t+1\}$ for any neighboring databases $x,x'$, we have that
	\end{enumerate}
			
	\begin{align*}
		LS^u(x, t, r) & = \max_{y| d(x,y) \leq t, z| d(y,z) \leq 1}|u(y,r)-u(z,r)|       \\
					& \leq \max_{y| d(x',y) \leq t+1, z| d(y,z) \leq 1}|u(y,r)-u(z,r)| \\
					& = LS^u(x', t+1, r)
		\\
	\end{align*}
				
	Thus $LS^u(x, t, r)$ is an admissible function. 
	
\end{proof}


\subsection{Proof of Lemma \ref{lemma:delta_convergence}}

\lemmadeltaconvergence*

\begin{proof}
                        
	\sloppy We show that $\min(\delta^u(x, t, r), \Delta u)$ is admissible. First we show that $\min(\delta^u(x, 0, r), \Delta u) \geq LS^u(x, 0, r)$. Thus, as $\delta^u(x, 0, r) \geq LS^u(x, 0, r)$ and $\Delta u \geq LS^u(x, 0, r)$ we have that $ \min(\delta^u(x, 0, r), \Delta u) \geq LS^u(x, 0, r)$. 
	
	Now, Suppose that $t>0$, let $y$ be a neighboring database of $x$. We have that $\delta^u(x, t+1, r), \Delta u) \geq \delta^u(y, t, r), \Delta u)$ as $\delta^u$ is admissible. This, $\min(\delta^u(x, t+1, r), \Delta u) \geq \min(\delta^u(y, t, r), \Delta u)$ holds. Thus $\min(\delta^u(x, t, r), \Delta u)$ is admissible.                        
					
	We move to show that $\min(\delta^u(x, t+1, r), \Delta u)$ is bounded. Suppose that $t \geq n$ The maximum hamming distance between two datasets is at most $n$. Thus $\{y| d(x,y) \leq n\} = D^n$. So we have:

	\begin{align*}
		LS^u(x, t, r) & = \max_{y| d(x,y) \leq t } LS^u(y, 0, r) = \max_{y \in D^n } LS^u(y, 0, r) = \Delta u \\
	\end{align*}                                
					
	\sloppy Therefore, we have that $\delta^u(x, t, r) \geq LS^u(x, t, r)$ since $\delta^u$ is admissible. Thus it implies that $\delta^u(x, t, r) = \Delta u$. Finally, $\min(\delta^u(x, t, r), \Delta u) = \Delta u$ for any $t>n$.
\end{proof}

\subsection{Proof of Lemma \ref{lemma:principal}}

\lemmaprincipal*

\begin{proof}
	Fix a database $x \in D^{n}$ and let $y \in D^{n}$ be a neighbor of $x$ such that $d(x,y) \leq 1$. Assume $u(x,r)$ lies in $[b(x,i,r), b(x,i+1,r))$ for some $i \in \mathbb{Z}$. We first show that $D_{u,\delta^u}(x,r) - D_{u,\delta^u}(y,r) \leq 1$. We analyse it in two cases: (1) $u(x,r) \geq 0$ and (2) $u(x,r)<0$.
			
	\sloppy \textbf{Case (1)}. Assume $u(x,r) \geq 0$. By construction of the dampening function $D_{u,\delta}$, $i\geq 0$ holds. Thus, one can find bounds for $u(y,r)$ using the definition of $LS^u$ and the admissibility of $\delta^u$.

	\begin{align*}
		u(y,r) & \geq u(x,r) - LS^u(x, 0, r) \\
				& \geq b(x,i,r) - \delta^u(x, 0, r) \\
				& = \sum^{i-1}_{j=1} \delta^u(x, j, r) \geq \sum^{i-2}_{j=0} \delta^u(y, j, r) \\
				& = b(y,i-1,r)         
	\end{align*}
	
	and 
	
	\begin{align*}
		u(y,r) & \leq u(x,r) + LS^u(x, 0, r) \\
				& \leq b(x,i+1,r) + \delta^u(x, 0, r) \\
				& = \sum^{i}_{j=0} \delta^u(x, j, r) + \delta^u(x, 0, r) \\
				& \leq \sum^{i+1}_{j=1} \delta^u(y, j, r) + \delta^u(x, 0, r) = b(y,i+2,r)         
	\end{align*}
	
	\sloppy Thus $u(y,r) \in [b(y,i-1,r), b(y,i+2,r))$. We split the argument in three subcases: (1.1) $u(y,r) \in [b(y,i-1,r), b(y,i,r))$. (1.2) $u(y,r) \in [b(y,i,r), b(y,i+1,r))$ and (1.3) $u(y,r) \in [b(y,i+1,r), b(y,i+2,r))$.
	
		\textit{Case (1.1)}. Assume that $u(y,r) \in [b(y,i-1,r), b(y,i,r))$. Thus we get the following:		
		\begin{align}
			& D_{u,\delta^u}(x,r) - D_{u,\delta^u}(y,r) \label{line:proof_dampeninga}\\
			& = \frac{u(x,r)-b(x,i,r)}{b(x,i+1,r) - b(x,i,r)} +i  - \nonumber\\ 
			& \quad - \frac{u(y,r)-b(y,i-1,r)}{b(y,i,r) - b(y,i-1,r)} -i+1  \label{line:proof_dampeningb}\\
			& \leq \frac{u(x,r)-b(x,i,r)- u(y,r)+b(y,i-1,r)}{b(y,i,r) - b(y,i-1,r)}+1 \label{line:proof_dampeningc}\\
			& \leq \frac{u(x,r)-b(x,i,r) - u(x,r)+\delta(x,0,r)+b(y,i-1,r)}{b(y,i,r) - b(y,i-1,r)}+1 \label{line:proof_dampeningd} \\
			& \leq \frac{-b(x,i,r)+b(x,i,r)}{b(y,i,r) - b(y,i-1,r)} + 1 \leq 1
	\end{align}
	\sloppy The rationale for the equations above is the following: Equation (\ref{line:proof_dampeningb}) follows from the application of the definition of $D_{u,\delta^u}$; Equation (\ref{line:proof_dampeningc}), as $b(x,i+1,r) - b(x,i,r) = \delta_u(x,i,r) \geq \delta_u(y,i-1,r) = b(y,i,r) - b(y,i-1,r)$, we have that:
	\begin{equation*}
		\frac{u(x,r)-b(x,i,r)}{b(x,i+1,r) - b(x,i,r)}  \leq \frac{u(x,r)-b(x,i,r)}{b(y,i,r) - b(y,i-1,r)}
	\end{equation*}
	
	Equation (\ref{line:proof_dampeningd}) holds since $u(y,r) \geq u(x,r) - LS^u(x,0,r) \geq u(x,r) - \delta^u(x,0,r)$. That $D_{u,\delta^u}(y,r) - D_{u,\delta^u}(x,r) \leq 1$ follows by symmetry.
	
	\textit{Case (1.2)}. Assume that $u(y,r) \in [b(y,i,r), b(y,i+1,r))$ which entails that  $D_{u,\delta^u}(y,r) \in [i, i+1)$. Likewise, as $u(x,r) \in [b(x,i,r), b(x,i+1,r))$ by assumption, it holds that $D_{u,\delta^u}(x,r) \in [i, i+1)$. In what follows, we get that:
	\begin{equation*}
		|D_{u,\delta}(x,r) - D_{u,\delta}(y,r)| \leq 1
	\end{equation*}
	
	\textit{Case (1.3)}. Assume that $u(y,r) \in [b(y,i+1,r), b(y,i+2,r))$. This case follows similar reasoning to the case (1.1), so we omit this part of the proof.
				
	\textbf{Case (2)}. Assume $u(x,r) < 0$. This case is symmetric to the case (1) as $D_{u,\delta}$ is symmetric.
	
	Given that, $|D_{u,\delta}(x,r)-D_{u,\delta}(y,r)|\leq 1$ holds for all pairs of neighboring databases $x,y \in D^n$ where $d(x,y) \leq 1$ and for all $r \in R$.
\end{proof}


\subsection{Proof of Lemma \ref{lemma:convergence}}

\lemmaconvergence*

\begin{proof}
	Let $r \in \mathcal{R}$ be an output element. By definition of $u^s$, observe that 
	
	\begin{align*}
		u^s(x,r) = u(x,r) - s \leq u(x,r) - n  \Delta u - \max_{r' \in R} u(x,r') = u^{s_0}(x,r)
	\end{align*}
			
	since $s \geq s_0$. And also, as $n \geq 0$, we get that
	
	\begin{align*}
		u^{s_0} = u(x,r) - n \Delta u - \max_{r' \in R} u(x,r') \leq -n \Delta u \leq 0  
	\end{align*}

	This means that $u^{s_0}(x,r)$ is non-positive and, by consequence, $u^{s}(x,r)$ is also non-positive for all $s>s_0$ and all $r \in R$. Therefore, by the construction of $D_{\deltau, u}$, we have that $u^{s_0} \in [b(x,i,r), b(x, i+1, r))$ for some $i \leq 0 $ since $u^{s_{0}}(x,r)\leq 0$. As $\deltau$ is bounded:
	
	\begin{align*}
		& b(x,i,r) \leq u^{s_0}(x,r) \\
		& \Rightarrow -\sum_{j=0}^{-i-1} \deltau(x,j,r) \leq u(x,r) -n \Delta u - \max_{r' \in R} u(x,r') \\
		& \Rightarrow (i+1) \Delta u \leq -n \Delta u  \\
		& \Rightarrow  i+1 \leq -n \\
		& \Rightarrow n \leq -i-1
	\end{align*}

	This last fact, the admissibility and convergence ($\deltau$ is bounded) of $\deltau$ lead us to show that the difference $b(x,k,r) - b(x,i,r)$ is equal to $(k-i)\Delta u$ for all $k \leq i < 0$. We will use this fact posteriorly.

	\begin{align}
		b(x,k,i)& - b(x,i,r) = \\
		& = -\sum_{j=0}^{-k-1}\deltau(x,-j,r) + \sum_{j=0}^{-i-1}\deltau(x,-j,r)  \\
		& = -\sum_{j=-i}^{-k-1}\deltau(x,-j,r) \\ 
		& = -\sum_{j=-i}^{-k-1}\deltau(x,n,r)\\
		& = (k-i) \Delta u \label{line:theo5_b}
	\end{align}
	
	Note that $u^{s} \in [b(x,k,r), b(x,k+1,r))$ for some $k \leq i \leq0 $ since $u^{s}(x,r) \leq u^{s_{0}}(x,r) \leq 0 $. Given that, we calculate the difference:   
	 
	\begin{align}
		& D_{u^{s_0},\deltau}(x,r) - D_{u^s,\deltau}(x,r) \\
		& = \frac{u^{s_0}(x,r) - b(x,i,r)}{b(x,i+1,r) - b(x,i,r)} + i \nonumber \\
		& \quad - \frac{u^{s}(x,r) - b(x,k,r)}{b(x,k+1,r) - b(x,k,r)} - k  \\
		& = \frac{u^{s_0}(x,r) - b(x,i,r) - u^{s}(x,r) + b(x,k,r)}{\Delta u} \quad - k + i \\
		& = \frac{u^{s_0}(x,r) - b(x,i,r) - u^{s_0}(x,r) - s_{0} + s + b(x,k,r)}{\Delta u} \nonumber \\
		& \quad  - k + i  \label{line:theo5_a}\\
		& = \frac{s - s_{0}}{\Delta u} + \frac{b(x,k,r) - b(x,i,r) }{\Delta u} - k + i  \label{line:theo5_c}\\
		& = \frac{s - s_{0}}{\Delta u} + \frac{(k-i)\Delta u }{\Delta u} - k + i  = \frac{s - s_{0}}{\Delta u}
	\end{align}
	Equation (\ref{line:theo5_a}) holds since $u^s(x,r) = u(x,r) - s = u^{s_0}(x,r) + s_o - s$ and equation \ref{line:theo5_c} follows from equation \ref{line:theo5_b}.
	
	Finally,
	\begin{align*}
		& \frac{ 
			\exp(\epsilon \; D_{u^{s}, \deltau}(x,r)/2)
		}{ 
			\sum_{r' \in R} \exp( \epsilon \; D_{u^s, \deltau}(x,r')/2)
		} \\ 
		& = \frac{ 
			\exp(\epsilon \; (D_{u^{s_0}, \deltau}(x,r)-(s - s_{0})/\Delta u)/2)
		}{ 
			\sum_{r' \in R} \exp( \epsilon \; (D_{u^{s_0}, \deltau}(x,r')-(s - s_{0})/\Delta u)/2)
		} \\
		& = \frac{
			\exp(-\epsilon(s - s_{0})/2 \Delta u).\exp(\epsilon \; D_{u^{s_0}, \deltau}(x,r)/2)
		}{ 
			\exp(-\epsilon(s - s_{0})/2 \Delta u). \sum_{r' \in R} \exp( \epsilon \; D_{u^{s_0}, \deltau}(x,r')/2)
		} \\
		& = \frac{
			\exp(\epsilon \; D_{u^{s_0}, \deltau}(x,r)/2)
		}{ 
			\sum_{r' \in R} \exp( \epsilon \; D_{u^{s_0}, \deltau}(x,r')/2)
		}
	\end{align*}
\end{proof}

\subsection{Proof of Lemma \ref{lemma:local_dampening_accuracy}}

\lemmaLocalDampeningAccuracy*

\begin{proof}
    We first prove point 1. Let $t$ be a real number larger than $0$. we need to show that the following expression if non-positive. 
    
    \begin{align}
        & Pr[\error(\mecsld,x) \geq \theta] - Pr[\error(\mecsldbar,x) \geq \theta] \\ 
        & = \sum_{r \in \range | u^* - u(x,r) \geq \theta} 
        Pr[\mecsld(x) = r] - Pr[\mecsldbar(x) = r]\\
        & = \sum_{r \in \range | u^* - u(x,r) \geq \theta} 
        \lim_{s \to \infty} \left( 
            \frac{ 
                \exp \left( \frac{\epsilon \, D_{u^{s}, \delta}(x,r)}{2} \right) 
            }
            { 
                \sum_{r' \in \mathcal{R}} \exp \left( \frac{\epsilon \, D_{u^{s}, \delta}(x,r')}{2} \right)
            } 
        \right)  \nonumber \\
        & \quad - \lim_{s \to \infty} \left( 
            \frac{ 
                \exp \left( \frac{\epsilon \, D_{u^{s}, \deltabar}(x,r)}{2} \right) 
            }{ 
                \sum_{r' \in \mathcal{R}} \exp \left( \frac{\epsilon \, D_{u^{s}, \deltabar}(x,r')}{2} \right)
            }
        \right)  \\	
        & = \sum_{r \in \range | u^* - u(x,r) \geq \theta} 
        \vast( 
            \frac{ 
                \exp \left( \frac{\epsilon \, D_{u^{s_0}, \delta}(x,r)}{2} \right) 
            }{ 
                \sum_{r' \in \mathcal{R}} \exp \left( \frac{\epsilon \, D_{u^{s_0}, \delta}(x,r')}{2} \right)
            } \nonumber \\
        & \quad 
            - \frac{ 
                \exp \left( \frac{\epsilon \, D_{u^{s_0}, \deltabar}(x,r)}{2} \right) 
            }{ 
                \sum_{r' \in \mathcal{R}} \exp \left( \frac{\epsilon \, D_{u^{s_0}, \deltabar}(x,r')}{2} \right)
            } 
        \vast) \label{line:accuracy_s0}\\	
        %
        %
        & = \frac{
            \sum_{r | u^* - u(x,r) \geq \theta} 
            \sum_{r' \in \mathcal{R}}
            \exp \left( 
                \frac{\epsilon \, (D_{u^{s_0}, \delta}(x,r) + D_{u^{s_0}, \deltabar}(x,r'))}{2}  
            \right)
        }{ 
            \sum_{r', r'' \in \mathcal{R}} 
            \exp \left( 
                \frac{\epsilon \, D_{u^{s_0}, \delta}(x,r')}{2} 
            \right) 
            \exp \left( 
                \frac{\epsilon \, D_{u^{s_0}, \deltabar}(x,r'')}{2}
            \right)
        } \\
        & \quad - \frac{
            \sum_{r | u^* - u(x,r) \geq \theta} 
            \sum_{r' \in \mathcal{R}}
            \exp \left( \frac{\epsilon \, (D_{u^{s}, \deltabar}(x,r)+D_{u^{s}, \delta}(x,r'))}{2}  \right) 		
        }{ 
            \sum_{r', r'' \in \mathcal{R}} \exp \left( \frac{\epsilon \, D_{u^{s}, \delta}(x,r')}{2} \right) \exp \left( \frac{\epsilon \, D_{u^{s}, \deltabar}(x,r'')}{2} \right)
        }\\	
        %
        & = \frac{ 
            \sum_{
                \begin{subarray}{l}r | u^* - u(x,r) \geq \theta \\ r'| u^* - u(x,r') < \theta \end{subarray}
            } 
            \left(
                \exp \left( 
                    \frac{\epsilon \, (D_{u^{s_0}, \delta}(x,r)+D_{u^{s}, \deltabar}(x,r'))}{2}  
                \right)
            \right)
        }{ 
            \sum_{r', r'' \in \mathcal{R}} 
            \exp \left( 
                \frac{\epsilon \, D_{u^{s_0}, \delta}(x,r')}{2} 
            \right) 
            \exp \left( 
                \frac{\epsilon \, D_{u^{s_0}, \deltabar}(x,r'')}{2} 
            \right)
        } \nonumber \\
        & \quad - \frac{ 
            \sum_{
                \begin{subarray}{l}r | u^* - u(x,r) \geq \theta \\ r'| u^* - u(x,r') < \theta \end{subarray}
            } 
            \left(
                \exp \left( 
                    \frac{\epsilon \, (D_{u^{s_0}, \deltabar}(x,r)+D_{u^{s}, \delta}(x,r'))}{2}  
                \right)
            \right)
        }{ 
            \sum_{r', r'' \in \mathcal{R}} 
            \exp \left( 
                \frac{\epsilon \, D_{u^{s_0}, \delta}(x,r')}{2} 
            \right) 
            \exp \left( 
                \frac{\epsilon \, D_{u^{s_0}, \deltabar}(x,r'')}{2} 
            \right)
        } \\
        & \leq 0 \label{line:accuracy_last_line}
    \end{align}

    The Line \ref{line:accuracy_s0} follows from the Lemma \ref{lemma:convergence} where $s_0 = \Delta u + u*$. The last inequality (Line \ref{line:accuracy_last_line}) follows from the following:
    
    \begin{align}
        & D_{u^{s}, \delta}(x,r') - D _{u^{s}, \delta}(x,r) \\
        & = \frac{
            u^{s_0}(x,r')-\sum_{t=0}^{n-1} \delta(x,t,r')
        }{
            \Delta u
        } \nonumber \\
        & \quad + n 
        - \frac{
            u^{s_0}(x,r)-\sum_{t=0}^{n-1} \delta(x,t,r)
        }{
            \Delta u
        } 
        - n \label{line:definition_so} \\	
        & = \frac{
                u^{s_0}(x,r')+\sum_{t=0}^{n-1} (\deltabar(x,t,r')+\alpha_{x,t,r'})
            }{
                \Delta u
            } \nonumber \\
        & \quad - \frac{
                u^{s_0}(x,r)+\sum_{t=0}^{n-1}(\deltabar(x,t,r)+\alpha_{x,t,r})
            }{
                    \Delta u
            } \\
        & = \frac{
            u^{s_0}(x,r')
            +\sum_{t=0}^{n-1} \deltabar(x,t,r')
            +\sum_{t=0}^{n-1} \alpha_{x,t,r'}
        }{
            \Delta u
        } \nonumber \\
        & \quad - \frac{
            u^{s_0}(x,r)
            +\sum_{t=0}^{n-1} \deltabar(x,t,r)
            +\sum_{t=0}^{n-1} \alpha_{x,t,r}
        }{\Delta u} \\
        & = D_{u^{s}, \deltabar}(x,r') +
         \frac{
            \sum_{t=0}^{n-1} \alpha_{x,t,r'}
        }{
            \Delta u
        } \quad \\ 
        & \quad - D_{u^{s}, \deltabar}(x,r)
        - \frac{          
            \sum_{t=0}^{n-1} \alpha_{x,t,r}
        }{\Delta u} \\
        & \geq D_{u^{s}, \deltabar}(x,r')- D_{u^{s}, \deltabar}(x,r) \label{line:last_line_acc1}
    \end{align}

Line \ref{line:definition_so} follows from the definition of the shifted local dampening when using shifting by $s_0$. The last inequality (Line \ref{line:last_line_acc1}) is due to the dominance of $\delta^u$ over $\deltabar^u$ and as $u(x,r') > u(x,r)$.
\end{proof}

\subsection{Proof of Lemma \ref{lemma:local_dampening_exponential}}

\lemmalocaldampeningexponential*

\begin{proof}
    First, we show that $Pr[\mecexp(x, \epsilon, u, \mathcal{R})=r] =Pr[\mecld(x, \epsilon, u, \delta^{u}, \mathcal{R})= r]$ where $\delta^{u}(x,t,r) = \Delta u $. For that, it suffices to demonstrate that $D_{u,\delta^u}(x, r) = \frac{u(x,r)}{\Delta u}$. We prove this claim for $u(x,r) \geq 0$. The case where $u(x,r) < 0$ is symmetric and it is omitted here. Let $i$ be the smallest integer such that $u(x,r) \in \left[ b(x,i,r), b(x,i+1,r) \right)$. We have that:

    \begin{align}
        b(x,i,r)  & = \begin{cases}
			\sum^{i-1}_{j=0} \delta^u(x, j, r) & \text{if} \ i > 0 \\
			0 & \text{if} \ i = 0 \\
			- b(x, -i, r)                    & \text{otherwise} \\
		\end{cases} \\
        & = i.\Delta u
    \end{align}.

    The last equality follows from the fact that $\delta^{u}(x,t,r) = \Delta u $ and $i \geq 0$. Therefore:

    \begin{align}
        D_{u,\delta^u}(x, r) & = \frac{u(x, r)-b(x,i,r)}{b(x, i+1, r)- b(x, i, r)} + i\\
        & = \frac{u(x,r) - i.\Delta u}{\Delta u} + i \\    
        & = \frac{u(x,r)}{\Delta u}
    \end{align}

    Then we proceed to show that $Pr[\mecld(x, \epsilon, u, \delta^{u}, \mathcal{R})= r] = Pr[\mecsld(x, \epsilon, u, \delta^{u}, \mathcal{R})= r]$. 

    \begin{align}
        Pr&[\mecsld(x, \epsilon, u, \delta^{u}, \mathcal{R})= r] = \\
        & =  \lim_{s \to \infty} \left( 
			\frac{ 
				\exp \left( \frac{\epsilon \, D_{u^{s}, \deltau}(x,r)}{2} \right) 
			}{ 
				\sum_{r' \in \mathcal{R}} \exp \left( \frac{\epsilon \, D_{u^{s}, \deltau}(x,r')}{2} \right)
			} 
		\right) \\
        & = \lim_{s \to \infty} \left( 
			\frac{ 
				\exp \left( \frac{\epsilon \, (u(x,r) - s)}{2 \Delta u} \right) 
			}{ 
				\sum_{r' \in \mathcal{R}} \exp \left( \frac{\epsilon \, (u(x,r) - s)}{2 \Delta u} \right)
			}
		\right)\\
        & = \lim_{s \to \infty} \left( 
			\frac{ 
				\exp \left( 
                    \frac{
                        \epsilon \, u(x,r)
                    }{
                        2 \Delta u
                    } 
                \right) 
                \exp \left( 
                    \frac{
                        \Delta u
                    }{
                        s \epsilon
                    } 
                \right) 
			}{ 
				\sum_{
                    r' \in \mathcal{R}
                } 
                \exp \left( 
                    \frac{
                        \epsilon \, u(x,r) 
                    }{
                        2 \Delta u
                    } 
                \right)
                \exp \left( 
                    \frac{
                        \Delta u
                    }{
                        s \epsilon
                    } 
                \right)
			}
		\right) \\
        & = \frac{ 
				\exp \left( 
                    \frac{
                        \epsilon \, u(x,r)
                    }{
                        2 \Delta u
                    } 
                \right) 
			}{ 
				\sum_{
                    r' \in \mathcal{R}
                } 
                \exp \left( 
                    \frac{
                        \epsilon \, u(x,r) 
                    }{
                        2 \Delta u
                    } 
                \right)
			}
    \end{align}

\end{proof}

\subsection{Proof of Lemma \ref{lemma:minimumls}}

\lemmaminimumls*

\begin{proof}
	We show that $LS^u(x, t, r)$ is less than or equal any admissible sensitivity function $\delta^u(x, t, r)$ by induction on $t$.
			
	\textbf{Basis}: for $t=0$, $LS^u(x, 0, r) \leq  \delta^u(x, 0, r) $ holds since $\delta^u$ is admissible for all $x \in D^{n},r \in R$
			
	\textbf{Inductive step}: suppose that $LS^u(x, t, r) \leq \delta^u (x, t, r)$ is true for all $x \in D^{n},r \in R$. We must show that $LS^u(x, t+1, r) \leq \delta^u (x, t+1, r)$ is true for all $x \in D^{n},r \in R$. By the definition of element local sensitivity:
			
	\begin{align*}
		LS^u(x, t+1, r) & = \max_{y| d(x,y) \leq t+1} LS^u(y,0,r)                               \\
					  & = \max_{y| d(x,y) \leq 1} \max_{y| d(y,z) \leq t} LS^u(z,0,r)         \\
					  & \leq \max_{y| d(x,y) \leq 1} \delta^u(y, t, r) \leq \delta^u(x, t+1, r) \\
	\end{align*}
			
	First inequality holds by hypothesis and the second inequality follows by the admissibility of $\delta^u$. Thus $LS^u(x, t+1, r) \leq \delta^u(x, t+1, r)$ for all $x \in D^{n},r \in \range, t \geq 0$.
\end{proof}

\subsection{Proof of Lemma \ref{lemma:percentile_local_sensitivity_distance_0}}

\lemmalocalsensitivitydistancezero*

\begin{proof}
    For this proof, we restate the definitions. The utility function is given by 

    $$u_p(x,r) = |v_x(p_x) - v_x(r)|, $$
    where $v_x(r)$ is the value of the element $r$ in $x$, $p_x \in \range$ is the $p$-th percentile element, i.e., the element with $k-th$ largest value in $x$ with $k = \ceil{\frac{p (n+1)}{100}}$.

    Thus, the element local sensitivity $LS^{u_p}(x,0,r)$ for $u_p$ is rewritten as:
    \begin{align*}
        LS^{u_p}(x,0,r) & = 
        \max(\underbrace{|v_x(p_x)-v_x(r)|}_{(1)},\\
        & \underbrace{||v_x(p_x) - v_x(r)|-|v_x(p_x^+)-v_x(r)||}_{(2)}, \\
        & \underbrace{||v_x(p_x) - v_x(r)|-|v_x(p_x^-)-v_x(r)||}_{(3)}, \\
        & \underbrace{p(x,r)}_{(4)},\underbrace{q(x,r)}_{(5)})
    \end{align*}
    where
    \begin{equation*}
        p(x,r) = \max
		\begin{cases}
            \underbrace{\Lambda - v_x(p_x^+)}_{(4.1)} & \text{if} \ r = p  \\
			\underbrace{\Lambda - v_x(r)}_{(4.2)} & \text{else if} \ v_x(r) >= v_x(p)  \\
			\underbrace{|r(x) -v_x(r) - \Lambda  |}_{(4.3)} & \text{otherwise}
        \end{cases}, 
    \end{equation*}
    \begin{equation*}
        q(x,r) = \max
        \begin{cases}
            \underbrace{v^x(p_x^-)}_{(5.1)} & \text{if} \ r = p \\
            \underbrace{|v_x(r) - s(x)|}_{(5.2)} & \text{if} \ v_x(r) >= v_x(p) \\ 
            \underbrace{v^x(r)}_{(5.3)} & \text{otherwise}
        \end{cases},        
    \end{equation*}
    where $p_x^+ \in \range$ and $p_x^- \in \range$ are the elements with the $(k+1)$-th and $(k-1)$-th largest value, respectively, $r(x) = v_x(p_x)+ v_x(p_x^+)$ and $s(x) = v_x(p_x) - v_x(p_x^-)$.

    We first prove that for every case above there is a neighboring database $y$ of $x$ which $|u(x,r)-u(y,r)|$ is larger than it:

    (1) Let $y$ be a database obtained from $x$ by changing the value of $r$ to $v_x(p_x)$, $v_y(r)=v_x(p_x)$. Thus $|u(x,r)-u(y,r)|=||v_x(p_x)-v_x(r)|-|v_x(p_x)-v_y(r)||=|v_x(p_x)-v_x(r)|$  as $v_y(p_y)=v_x(p_x)$ and $v^y(r)=v_x(p_x)$.

    (2) Let $y$ be a database obtained from $x$ by changing the value of $p_x$ to $\Lambda$, $v_y(p_x)=\Lambda$. This way, $v_y(p_y)=v_x(p_x^+)$ and $v_y(r)=v_x(r)$. Then $|u(x,r)-u(y,r)|=||v_x(p_x)-v_x(r)|-|v_y(p_y)-v_y(r)|| = ||v_x(p_x) - v_x(r)|-|v_x(p_x^+)-v_x(r)||$.

    (3) Let $y$ be a database obtained from $x$ by changing the value of $p_x$ to $0$, $v_y(p_x)=0$. This way, $v_y(p_y)=v_x(p_x^-)$ and $v_y(r)=v_x(r)$. Then $|u(x,r)-u(y,r)|=||v_x(p_x)-v_x(r)|-|v_y(p_y)-v_y(r)|| = ||v_x(p_x) - v_x(r)|-|v_x(p_x^-)-v_x(r)||$.

    (4.1) Let $y$ be a database obtained from $x$ by changing the value of $r$ to $\Lambda$, $v_y(r)=\Lambda$. In this case, $r=p_x$, so $v_y(p_y)=v_x(p_x^+)$. Then, we have that $|u(x,r)-u(y,r)|=||v_x(p_x)-v_x(r)|-|v_y(p_y)-v_y(r)|| = ||v_x(p_x) - v_x(p_x)|-|v_x(p_x^+)-\Lambda|| = |v_x(p_x^+) - \Lambda| = \Lambda = v_x(p_x^+)$.

    (4.2) Let $y$ be a database obtained from $x$ by changing the value of $r$ to $\Lambda$, $v_y(r)=\Lambda$. In this case, $v_x(r) >= v_x(p)$, so $v_y(p_y)=v_x(p_x)$. Then, we have that $|u(x,r)-u(y,r)|=||v_x(p_x)-v_x(r)|-|v_y(p_y)-v_y(r)|| = |v_x(r)-v_x(p_x) - \Lambda + v_x(p_x)| = \Lambda - v_x(r)$.

    (4.3) Let $y$ be a database obtained from $x$ by changing the value of $r$ to $\Lambda$, $v_y(r)=\Lambda$. In this case, $v_x(r) < v_x(p)$, so $v_y(p_y)=v_x(p_x^+)$. Then, we have that $|u(x,r)-u(y,r)|=||v_x(p_x)-v_x(r)|-|v_y(p_y)-v_y(r)|| = |v_x(p_x)-v_x(r) - \Lambda + v_x(p_x^+) |$.

    (5.1) Let $y$ be a database obtained from $x$ by changing the value of $r$ to $0$, $v_y(r)=0$. In this case, $r=p_x$, so $v_y(p_y)=v_x(p_x^-)$. Then, we have that $|u(x,r)-u(y,r)|=||v_x(p_x)-v_x(r)|-|v_y(p_y)-v_y(r)|| = |v_x(p_x^-)-0| = v_x(p_x^-)$.

    (5.2) Let $y$ be a database obtained from $x$ by changing the value of $r$ to $0$, $v_y(r)=0$. In this case, $v_x(r) >= v_x(p)$, so $v_y(p_y)=v_x(p_x^-)$. Then, we have that $|u(x,r)-u(y,r)|=||v_x(p_x)-v_x(r)|-|v_y(p_y)-v_y(r)|| = |v_x(r) - v_x(p_x) - v_x(p_x^-)|$.

    (5.3) Let $y$ be a database obtained from $x$ by changing the value of $r$ to $0$, $v_y(r)=0$. In this case, $v_x(r) < v_x(p)$, so $v_y(p_y)=v_x(p_x)$. Then, we have that $|u(x,r)-u(y,r)|=||v_x(p_x)-v_x(r)|-|v_y(p_y)-v_y(r)|| = |v_x(p_x)-v_x(r) - v_x(p_x)| = v_x(r)$.

    Now we prove that for every neighboring database $y$ of $x$, the exists a case where $|u(x,r)-u(y,r)|$ is smaller or equal than it. The database $y$ can be obtained from $x$ by choosing a element $r' \in \range$ and then choosing a new value $v_y(r')$ for it. We divide the set of neighboring databases in some cases:

    (i) The element $r'$ is such that $r' \neq r$. This operation cannot change the value of $r$, $v_y(r)=v_x(r)$. However, this operation can set the value of $p_y$. The element $r'$ can move freely along the domain $[0, \Lambda]$. However, there are only three possibilities for $p_y$: it can keep the value of $p_x$, it can change to the value of $p_x^+$ or it can change to $p_x^-$. 
    
    (i.i) if $v_y(p_y)=v_x(p_x)$ then $|u(x,r)-u(y,r)|=0$; 
    
    (i.ii) if $v_y(p_y)=v_x(p_x^+)$ then $|u(x,r)-u(y,r)|=||v_x(p_x)-v_x(r)|-|v_y(p_y)-v_y(r)|| = ||v_x(p_x)-v_x(r)|-|v_x(p_x^+)-v_x(r)||$ (case 2); and 
    
    (i.iii) if $v_y(p_y)=v_x(p_x^-)$ then $|u(x,r)-u(y,r)|=||v_x(p_x)-v_x(r)|-|v_y(p_y)-v_y(r)|| = ||v_x(p_x)-v_x(r)|-|v_x(p_x^-)-v_x(r)||$ (case 3).

    (ii) The element $r'$ is such that $r' = r$. We split this argument in two subcases:

    (ii.i) $r'= r = p_x$. Here we divide again in 3 subcases to set the value of $r'$ on $y$: 
    
    (ii.i.i) $v_y(r') = p_x$. It means that $x=y$ and $|u(x,r)-u(y,r)|=||v_x(p_x)-v_x(r)|-|v_y(p_y)-v_y(r)|| = ||v_x(p_x)-v_x(r)|-|v_x(p_x)-v_x(r)|| = 0$. 

    (ii.i.ii) $v_y(r')=v_y(r') \geq p_x$. So we have that the pian element of $y$ is at maximum $v_x(p_x^+)$ and that $v_y(r') \geq v_y(p_y)$. Thus $|u(x,r)-u(y,r)|=||v_x(p_x)-v_x(r)|-|v_y(p_y)-v_y(r)|| = |0-|v_y(p_y)-v_y(r)|| = v_y(r) - v_y(p_y) \leq \Lambda - v_x(p_x^+)$ (case 4.1). 

    (ii.i.iii) $v_y(r')=v_y(r') \leq p_x$. So we have that the pian element of $y$ is at minimum $v_x(p_x^-)$ and that $v_y(r') \leq v_y(p_y)$. Thus $|u(x,r)-u(y,r)|=||v_x(p_x)-v_x(r)|-|v_y(p_y)-v_y(r)|| = |0-|v_y(p_y)-v_y(r)|| = v_y(p_y) - v_y(r) \leq v_x(p_x^-)$ (case 5.1). 

    (ii.ii) $r' = r \neq p_x$. We divide in two subcases:

    (ii.ii.i) $v_x(r) \geq v_x(p_x)$. We divide again in three cases:

    (ii.ii.i.i) $v_y(r) = v_y(r') \geq v_x(r) = v_x(r')$. Thus $|u(x,r)-u(y,r)|=||v_x(p_x)-v_x(r)|-|v_y(p_y)-v_y(r)|| = ||v_x(p_x)-v_x(r)|-|v_y(p_y)-v_y(r)|| = v_x(p_x) - v_x(r) + v_y(r) - v_y(p_y)  \leq \Lambda - v_x(r)$ since $v_x(p_x) = v_y(p_y)$ (case 4.2). 

    (ii.ii.i.ii) $v_x(p_x) \leq v_y(r) = v_y(r') < v_x(r) = v_x(r')$. Thus $|u(x,r)-u(y,r)|=||v_x(p_x)-v_x(r)|-|v_y(p_y)-v_y(r)|| = |v_x(r)-v_x(p_x)-|v_x(p_x)-v_y(r)|| = |v_x(p_x) - v_x(r) + v_y(r) - v_x(p_x)|  = v_x(r) - v_y(r) \leq v_x(r) - v_x(p_x) \leq |v_x(p_x) - v_x(r)|$ (case 1). 

    (ii.ii.i.iii) $0 \leq v_y(r) < v_x(p_x)$. So we have that $v_y(p_y) = v_x(p_x^-)$. Thus $|u(x,r)-u(y,r)|=||v_x(p_x)-v_x(r)|-|v_y(p_y)-v_y(r)|| = |v_x(r)-v_x(p_x)-v_x(p_x^-)+v_y(r)| \leq |v_x(r)-v_x(p_x)-v_x(p_x^-)|$ (case 5.2). The last inequality follows by setting $v_y(r)=0$.

    (ii.ii.ii) $v_x(r) \leq v_x(p_x)$. We divide in three cases:

    (ii.ii.ii.i) $v_y(r) = v_y(r') > v_x(p_x)$. So we have that $v_y(p_y) = v_x(p_x^+)$. Thus $|u(x,r)-u(y,r)|=||v_x(p_x)-v_x(r)|-|v_y(p_y)-v_y(r)|| = |v_x(p_x) -v_x(r) -v_y(r)+v_x(p_x^+)|  \leq |v_x(p_x) -v_x(r) -\Lambda+v_x(p_x^+)$ (case 4.3). 

    (ii.ii.ii.ii) $v_x(r) = v_x(r') < v_y(r) = v_y(r') \leq v_x(p_x)$. So we have that $v_y(p_y) = v_x(p_x)$ Thus $|u(x,r)-u(y,r)|=||v_x(p_x)-v_x(r)|-|v_y(p_y)-v_y(r)|| = |v_x(p_x) -v_x(r)-v_x(p_x)+v_y(r)| = |v_y(r)-v_x(r)| \leq |v_x(p_x) - v_x(r)|$ (case 1).

    (ii.ii.ii.i) $0 \leq v_y(r) = v_y(r') \leq v_x(r) = v_x(r')$. So we have that $v_y(p_y) = v_x(p_x)$. Thus $|u(x,r)-u(y,r)|=||v_x(p_x)-v_x(r)|-|v_y(p_y)-v_y(r)|| = |v_x(p_x) -v_x(r) - v_x(p_x) + v_y(r)| = |v_y(r) - v_x(r)| \leq v_x(r)$ (case 5.3). 

\end{proof}

\subsection{Proof of Lemma \ref{lemma:global_ebc}}
We first need to show the Lemma below.

\begin{restatable}{lemma}{lemmaebc}
    Let $G$ and $G'$ be two neighboring graphs and $v$ a node belonging to $V(G)$ and $V(G')$, we have that:
    \begin{align}
        \max_{G,G' | d(G,G') \leq 1} & |EBC^{G}(v) - EBC^{G'}(v)| \\ 
        & = \max{\left( d^{G}(v)(d^{G}(v)-1)/4, d^{G}(v) \right)},  
    \end{align}
    
    where $d^{G}(v)$ denotes the degree of $v$ in $G$, i.e., $d^{G}(v)=|N^{G}_v|$.
    \label{lemma:ebc}
\end{restatable}

\begin{proof}
    Let $\Delta (v)$ be defined as
    \begin{align*}
        \Delta (v) &= \max_{G,G' | d(G,G')} |EBC^{G}(v) - EBC^{G'}(v)| \\
        &= \max_{G,G' | d(G,G')} \left| \sum_{x,y \in N_v^{G} | x \neq y} b_{xy}^{G}(v) - \sum_{x,y \in N_v^{G'} | x \neq y} b_{xy}^{G'}(v) \right| 
    \end{align*}

    where $b_{uy}^{G'}(v) = g_{uy}^{G'}(c) / g_{uy}^{G'}$. Without loss of generality, let $e \in V(G)$ the edge that belongs to $G'$ and not to $G$, i.e., $E(G') = E(G) \cup \{ e \}$. We analyse two cases for $e$:        

    \textbf{Case (1).} One end of $e$ is $v$, $e=(v u)$, i.e., $N(v)^{G'} = N(v)^{G} \cup \{ u \}$. Since $e$ does not belong to $G$ then $u$ is not a neighbor of $v$ in $G$ which means that the terms $b_{uy}$ for all $y \in N_v^{G}$ are the only terms that do not exist on the expression for $EBC^{G}(v)$. So we rewrite $EBC^{G'}(v)$ as $\sum_{x,y \in N_v^{G'} | x \neq y} b_{xy}^{G}(v) + \sum_{y \in N_v^{G}} b_{uy}^{G'}(v)$ and $\Delta (v)$ as

    \begin{align*}            
        & \max_{G,G' | d(G,G')} \vast| \sum_{x,y \in N_v^{G} | x \neq y} b_{xy}^{G}(v) - \sum_{x,y \in N_v^{G} | x \neq y} b_{xy}^{G'}(v) \\
        & \quad \quad \quad \quad \quad - \sum_{y \in N_v^{G}} b_{uy}^{G'}(v) \vast| \\ 
        & = \max_{G,G' | d(G,G')} \left| \sum_{x,y \in N_v^{G} | x \neq y} (b_{xy}^{G}(v) - b_{xy}^{G'}(v)) - \sum_{y \in N_v^{G}} b_{uy}^{G'}(v) \right|
    \end{align*}
    
    We find bounds for $\sum_{y \in N_v^{G}} b_{uy}^{G'}(v)$. $b_{uy}^{G'}(v)$ is non-negative as $g_{uy}^{G'}$ is positive and $g_{uy}^{G'}(c)$ is non-negative. $b_{uy}^{G'}(v) \leq 1$ since $g_{uy}^{G'} \geq g_{uy}^{G'}(c)$. Moreover there are $|N_v^{G}| = d$ pairs $u,y$ since $y \in N(v)^G$. Thus 

    $$0 \leq \sum_{y \in N_v^{G}} b_{uy}^{G'}(v) \leq d$$

    Now we find bounds for $\sum_{x,y \in N_v^{G} | x \neq y} (b_{xy}^{G}(v) - b_{xy}^{G'}(v))$. If a geodesic path from $x$ to $y$ ($x,y \in N_v^{G}$) has size $1$ in $G'$, i.e., the edge $(xy)$ belongs to $E(G')$, then there is only one geodesic path and it does not contain the central node $v$ which implies that $b_{xy}^{G'}(v) = 0$. That also holds for $G$ since the edge $(xy)$ also exists in $G$ so $b_{xy}^{G}(v) = 0$. Therefore $b_{xy}^{G}(v) - b_{xy}^{G'}(v) = 0$ for a pair $x,y \in N_v^{G}$ at distance $1$. Also, there is no pairs of nodes $x,y$ at distance $3$ or more since it exist the path $<x \, v \, y>$ in both $G$ and $G'$.
    
    Consider a pair of nodes $x,y \in N_v^{G}$ where $x$ is at distance $2$ from $y$ in $G'$. If none of the geodesic paths from $x$ to $y$ contains $u$ in $G'$, then the number of geodesic paths from $x$ to $y$ (containing $v$ or not) does not change from $G$ to $G'$. So we have that $b_{xy}^{G}(v) - b_{xy}^{G'}(v) = 0$. Thus we are interested in the case that a given pair $x,y \in N_v^{G}$ is at distance $2$ where $u$ belongs to a geodesic path from $x$ to $y$ in $G'$ ( consequently, also in $G'$). All geodesic paths from $x$ to $y$ from $G$ are preserved in $G'$ as no edges were removed. But there is a new path $< x u y>$ in $G'$ so $g_{xy}^{G'} = g_{xy}^{G} + 1$ and as there is only one path $< x v z >$ that contains the central node $V$ in $G$ and $G'$, $g_{xy}^{G}(c) = g_{xy}^{G'}(c) = 1$. Then

    \begin{align}
        b_{xy}^{G}(v) - b_{xy}^{G'}(v) = \frac{1}{g_{xy}^{G}} -\frac{1}{g_{xy}^{G}+1}
    \end{align}
    
    Note that $(1 / g_{xy}^{G}) -(1 / (g_{xy}^{G}+1))$ is monotonically decreasing on $g_{xy}^{G}$ since 
    
        $$\frac{d}{d g_{xy}^{G}} \left[ \frac{1}{g_{xy}^{G}} - \frac{1}{g_{xy}^{G}+1} \right] = - \frac{1}{(g_{xy}^{G}+1)^2} - \frac{1}{(g_{xy}^{G})^2} < 0$$

    so since $b_{xy}^{G}(v) \leq 1$, we have $b_{xy}^{G}(v) - b_{xy}^{G'}(v) \leq  1/2$. Besides that, we count how many possible path of the form $<x u y>$ for $x,y \in N_c^{G}$ where $x \neq y$. Since there are $d = |N_c^{G}|$ there exists at most $\binom{d}{2} = d(d-1)/2$ of those paths. Thus we have:

    $$ 0 \leq \sum_{x,y \in N_v^{G} | x \neq y} (b_{xy}^{G}(v) - b_{xy}^{G'}(v)) \leq \frac{d(d-1)}{2} . \frac{1}{2} = \frac{d(d-1)}{4}$$

    Thus

    \begin{align*}
        & d \leq \sum_{x,y \in N_v^{G} | x \neq y} (b_{xy}^{G}(v) - b_{xy}^{G'}(v)) - \sum_{y \in N_v^{G}} b_{uy}^{G'}(v) \leq \frac{d(d-1)}{4} \\
        & \implies \left| \sum_{x,y \in N_v^{G} | x \neq y} (b_{xy}^{G}(v) - b_{xy}^{G'}(v)) - \sum_{y \in N_v^{G}} b_{uy}^{G'}(v) \right| \\
        & \quad \quad \leq \max{\left( \frac{d(d-1)}{4}, d \right)} \\
        & \implies \Delta (x) = \max{\left( \frac{d(d-1)}{4}, d \right)}
    \end{align*}

    \textbf{Case (2):} None of the ends of $e$ is $v$. We omit this part of the proof since it has similar reasoning to case 1 and it yields a lower sensitivity than case 1.
\end{proof}

Lemma \ref{lemma:global_ebc} is a a direct consequence of Lemma \ref{lemma:ebc} given that $\Delta G$ is the degree of the node with largest degree.

\subsection{Proof of Lemma \ref{lemma:admissible_delta_ebc}}

\lemmaebclocalsensitivity*

\begin{proof}
    We show that function $\delta^{EBC}(G,t,v)$ is admissible. 
    First, we show that $\delta^{EBC}(G,0,v) = \left( d(d-1)/4, d \right) \geq LS^{EBC}(G,0,v)$ where $d$ is the degree of $v$. Lemma \ref{lemma:ebc} proves that for every pair of neighboring graphs $G',G^*$, $|EBC^{G'}(v) - EBC^{G^*}(v)| = \max{\left( d(d-1)/4, d \right)}$. By fixing $G$, we obtain that 
    
    \begin{align*}
        LS^{EBC}(G,0,v) & = \max_{G', d(G, G') \leq 1} |EBC^{G}(v) - EBC^{G'}(v)| \\
        & \leq \max{\left( d(d-1)/4, d \right)}
        = \delta^{EBC}(G,0,v)
    \end{align*}

    It remains to demonstrate that $\delta^{EBC}(G,t,v) \leq \delta^{EBC}(G',t+1,v)$ for all neighboring graphs $G'$. We first show that $|d^{G} - d^{G'}| \leq 1$ where $d^{G}$ and $d^{G'}$ are the degree of $v$ in $G$ and $G'$ resp. $G$ and $G'$ differ in just on edge, say $e$. Suppose $e$ belongs to $G$ and not to $G'$, if $e$ is not incident on $v$ in $G$ then $|d^G - d^{G'}|=0$. Otherwise, if $e$ is incident on $v$ in $G$ then $d^{G} - d^{G'}=1$. So $d^{G} - d^{G'} \leq 1$. By symmetry,  $d^{G'} - d^{G} \leq 0$ holds. Then $|d^{G} - d^{G'}| \leq 1$            
    
    Applying this last fact to $\delta^{EBC}(G,t,v)$ we have:

    \begin{align*}
        \delta^{EBC}(G,t,v) & = \max{\left( \frac{(d^{G}+t)(d^{G}+t-1)}{4}, d^{G}+t \right)}\\
        & \leq \max{\left( \frac{(d^{G'}+t+1)(d^{G}+t)}{4}, d^{G}+t+1 \right)} \\
        & = \delta^{EBC}(G',t+1,v)
    \end{align*}
\end{proof}

\subsection{Proof of Lemma \ref{lemma:local_Infogain0}}

\lemmalocalinfogain*

\begin{proof}

    First, we provide some properties for function $f$:

    \begin{enumerate}
        \item The domain of $f$ is composed by integers greater or than $0$;
        \item $f$ is non-negative over its domain since $f(0)=0$ and, for $x \geq 1$, both terms $x.log((x+1)/x)$ and $log(x+1)$ are positive;
        \item $f$ is non-decreasing over its domain as, for $x \geq 1$, 
        \begin{equation*}
            f'(x) = \log \left( \frac{x+1}{x} \right) > 0
        \end{equation*}
        and $f(0)=0 \leq f(x)$, for $x \geq 1$.
    \end{enumerate}

    Also, we list those properties for function $g$.
    \begin{enumerate}
        \item The domain of $g$ is composed by integers greater or than $0$;
        \item $g$ is non-positive as $g(0)=0$ and $g(1)=0$ and, for $x \geq 2$, both terms $x.log((x-1)/x)$ and $log(x-1)$ are negative;
        \item $g$ is non-increasing over its domain as, for $x \geq 2$
        \begin{equation*}
            g'(x) = \log \left( \frac{x-1}{x}  \right)
        \end{equation*}
        and $g(0)=g(1)=0 \leq g(x)$, for $x \geq 2$.
    \end{enumerate}

    Let $\TT$ be some dataset. Neighboring datasets to $\TT$ differ in at most one tuple $r$. Note that the only terms in the sum of $IG(\TT, A)$ (Equation \ref{line:utility_ig}) that change are those related to the value that $r$ take on $A$, say $j$. So we can focus only on:

    \begin{align*}
        IG_j(\TT, A) & = \sum_{c \in C} \tjc . \log \frac{\tjc}{\tj} \\
        & = \sum_{c \in C} \tjc \log \tjc - \tj log \tj \\
    \end{align*}

    Let $\TT'$ be some dataset where $d(\TT, \TT') \leq 1$. Thus the datasets $\TT$ and $\TT$ differ in one tuple $r$ in one of the two cases: (1) $\TT'$ is obtained from $\TT$ by removing a tuple $r$, i.e., $\TT' = \TT \setminus \{r\}$  or (2) $\TT'$ is obtained from $\TT$ by adding a tuple $r$, i.e., $\TT' = \TT \cup \{r\}$.
    
    Consider case (1), assume $\TT' = \TT \setminus \{r\}$. The removal of the a tuple $r$ with $r_A = j$ and $r_C = c$, for some $j \in A$ and $c \in C$, decreases $\tjc$ by one which modify one term in the left hand sum and also $\tj$ decreases by one. So we have:

    \begin{align*}
        |u(\TT, r) - u(\TT', r)| & = |(\tjc -1).\log(\tjc - 1) - \tjc.\log(\tjc) \\
        & \quad + \tj.\log{\tj} - (\tj - 1).\log(\tj-  1)| \\
        & = |\tjc.\log \left( \frac{\tjc - 1}{\tjc} \right) -\log(\tjc - 1) \\
        & \quad - \tj.\log\left( \frac{\tj - 1}{\tj} \right) - \log(\tj-  1)| \\
        & = |g(\tjc) - g(\tj)| = g(\tjc) - g(\tj) \\
    \end{align*}

    The last equality comes from the fact that $\tj \geq \tjc \geq 0 $ and that $g(x)$ is non-increasing over its domain.

    For case (2), assume that $\TT' = \TT \cup \{r\}$. With analogous reasoning we get that:

    \begin{align*}
        |u(\TT, r) - u(\TT', r)| = |f(\tjc) - f(\tj)|  = f(\tj) - f(\tjc) \\
    \end{align*}

    The last equality holds since $\tj \geq \tjc \geq 0$ and $f(x)$ is non-decreasing over its domain.

    Therefore, our local sensitivity expression is given as:

    \begin{align*}
        LS^{IG}(\TT,0,A) & = \max_{c \in C, j \in A} \max(g(\tjc) - g(\tj), &
        \\ & \quad \quad \quad \quad \quad \quad \quad f(\tj) - f(\tjc)) \\
        & = h(\tj, \tjc)
    \end{align*}
\end{proof}


\subsection{Proof of Lemma \ref{lemma:local_Infogain}}    
    \lemmalocainfogaint*
    \begin{proof}

        By the definition of element local sensitivity (Definition \ref{def:item_local}) and lemma \ref{lemma:local_Infogain0}, we get that:

        \begin{align*}
            LS^{IG}(\TT,t,A) & = \max_{\TT'| d(\TT,\TT') \leq t} LS^{u}(\TT',0,A) \\
            & =  \max_{\TT'| d(\TT,\TT') \leq t} \max_{c \in C, j \in A} h(\tjx{\TT'}, \tjcx{\TT'}) \\             
            & =  \max_{c \in C, j \in A} \max_{\TT'| d(\TT,\TT') \leq t} h(\tjx{\TT'}, \tjcx{\TT'}) \\             
        \end{align*}

        where $\tjx{\TT} = \{ r \in \TT : r_A = j \}$ and $\tjcx{\TT} = \{ r \in \TT : r_A = j \wedge r_c = c \}$. So we need to choose $\TT'$ ($d(\TT,\TT') \leq t$) where the counts $\tjx{\TT'}$ and $\tjcx{\TT'}$ maximizes $h$ for a given $c \in C$ and $j \in A$.
        
        To obtain $\TT'$ from $\TT$, one can make up to $t$ operations of insertion or removal of tuples $r = <r_1, r_2, \dots, r_d>$. Note that $h(\tjx{\TT'}, \tjcx{\TT'})$ depends only on the values that $r$ takes on attribute $A$ and $C$. Thus we can group the operations by their influence on $\tjx{\TT'}$ and $\tjcx{\TT'}$ in $4$ types:

        \begin{enumerate}
            \item Insertion of a tuple $r$ where $r_A = j$ and $r_C = c$, e.g., inserting $r$ in $\TT$ increases the count $\tjx{\TT}$ and $\tjcx{\TT}$ by one;
            \item Insertion of a tuple $r$ where $r_A = j$ and $r_C \neq c$, e.g., inserting $r$ in $\TT$ increases the count $\tjx{\TT}$ by one and $\tjcx{\TT}$ remains the same;
            \item Removal of a tuple $r$ where $r_A = j$ and $r_C = c$, e.g., removing $r$ from $\TT$ decreases the count $\tjx{\TT}$ and $\tjcx{\TT}$ by one;
            \item Removal of a tuple $r$ where $r_A = j$ and $r_C \neq c$, e.g., removing $r$ from $\TT$ decreases the count $\tjx{\TT}$ by one and $\tjcx{\TT}$ remains the same.
        \end{enumerate}        

        Therefore, given $\TT$, $c \in C$ and $j \in A$ and $t$, we show that Candidates$(\TT, t, j, c)$  returns a set of pairs of counts $(\tjx{\TT'}, \tjcx{\TT'})$ ($d(\TT, \TT') = t$) such that:

        \begin{align*}
            & \max_{\TT'| d(\TT,\TT') = t} h(\tjx{\TT'}, \tjcx{\TT'})  = \max_{(a,b) \in Candidates(\TT, t, j, c)} h(a,b) \\
        \end{align*}

        It suffices to show the following properties: 1) $Candidates(\TT, t, j, c) \subseteq \{(\tjx{\TT'}, \tjcx{\TT'}) | d(\TT,\TT') = t \} $ and; 2) for all pairs $(a,b)$ that belongs to $\{(\tjx{\TT'}, \tjcx{\TT'}) | d(\TT,\TT') = t \}$ and does not belong to $Candidates(\TT, t, j, c)$, $h(a,b) \leq h(\tjx{\TT^*}, \tjcx{\TT^*})$ for some $\TT^*$ such that $d(\TT,\TT^*) = t$ holds. The latter means that we can safely discard this pair $(a,b)$ as it is not maximum because the value $h(a,b)$ is smaller than some other pair. 
        
        We demonstrate those property $1$ by induction on $t$.

        \textbf{Basis:} for $t=0$, 
        \begin{align*}
            Candidates(\TT, 0, j, c) & = \{ (\tjx{\TT}, \tjcx{\TT}) \} \\            
            & = \{(\tjx{\TT'}, \tjcx{\TT'}) | d(\TT,\TT') = 0 \} \\
        \end{align*}

        \textbf{Inductive step:} Suppose that $Candidates(\TT, t, j, c) \subseteq \{(\tjx{\TT'}, \tjcx{\TT'}) | d(\TT,\TT') = t \}$ (property 1). 
        
        Let $\TT^*$ be a dataset ($d(\TT,\TT^*) = t+1$) such that it is obtained by applying and operation of type $2$ on a dataset $\TT'$ ($d(TT,TT')=t$). It implies that $(\tjx{\TT'} +1, \tjcx{\TT'})  \in  \{(\tjx{\TT^*}, \tjcx{\TT^*}) | d(\TT,\TT^*) = t+1 \}$. Now, suppose that the dataset $\TT^*$ ($d(\TT,\TT^*) = t+1$) is obtained by applying and operation of type $3$ on a dataset $\TT'$ ($d(\TT,\TT')=t$). If it exists a tuple $r$ where $r_A=j$ and $r_C=c$, i.e. $\tjx{\TT'} >0$ and $\tjcx{\TT'} > 0$, we get that $(\tjx{\TT'} -1, \tjcx{\TT'} -1)  \in  \{(\tjx{\TT'}, \tjcx{\TT^*}) | d(\TT,\TT^*) = t+1 \}$

        
        By hypothesis, $(\tjx{\TT'}, \tjcx{\TT'}) \in Candidates(\TT, t, j, c)$. Thus, from lines $6$ to $11$ in the call for $Candidates(\TT, t+1, j, c)$, the algorithm produces and returns the pairs $(\tjx{\TT'}+1, \tjcx{\TT'})$ (if $\tjx{\TT'}<\ttt^{\TT}$) and $(\tjx{\TT'}-1, \tjcx{\TT'}-1)$ (if $\tjx{\TT'}>0$ and $\tjcx{\TT'}>0$). Thus $Candidates(\TT, t+1, j, c) \subseteq \{(\tjx{\TT'}, \tjcx{\TT'}) | d(\TT,\TT') = t+1 \}$. 
        
        Now we show property $2$. Let $\TT^*$ be a dataset ($d(T,T^*) = t$) such that $(\tjx{\TT^*},\tjcx{\TT^*}) \in \{(\tjx{\TT'}, \tjcx{\TT'}) | d(\TT,\TT') = t \}$ and $(\tjx{\TT^*},\tjcx{\TT^*}) \notin Candidates(\TT, t, j, c)$. Note that lines $8$ and $10$ of the Candidates algorithm mimic operations of type $2$ and $3$. All counts $(\tjx{\TT'}, \tjcx{\TT'})$ returned by $Candidates(\TT, t, j, c)$ comes from datasets $\TT'$ obtained from $\TT$ by the application of $t$ operations of type $2$ and $3$. Thus, if $(\tjx{\TT^*},\tjcx{\TT^*}) \notin Candidates(\TT, t, j, c)$, $\TT^*$ was obtained from $\TT$ by applying at least one operation of type $1$ or $4$. 
        
        Suppose that $\TT^*$ was obtained from $\TT$ by applying at least one operation of type $1$. Let $\TT^\dagger$ be a dataset obtained from $\TT$ from the same operations used to obtain $\TT^*$ except that we replace one operation of type $1$ to an operation of type $2$. So we have that $\tjx{\TT^*} = \tjx{\TT^\dagger}$ and $\tjcx{\TT^*} = \tjcx{\TT^\dagger} + 1$. Since $f$ is non-decreasing (see proof for Lemma \ref{lemma:local_Infogain0}):

        \begin{align*}
            f(\tjx{\TT^\dagger}) - f(\tjcx{\TT^\dagger}) & = f(\tjx{\TT^*}) - f(\tjcx{\TT^*}-1) \\
            & \geq f(\tjx{\TT^*}) - f(\tjcx{\TT^*})
        \end{align*}

        and as $g$ is non-increasing (see proof for Lemma \ref{lemma:local_Infogain0}):

        \begin{align*}
            g(\tjcx{\TT^\dagger}) - g(\tjx{\TT^\dagger}) & = g(\tjcx{\TT^*}-1) - g(\tjx{\TT^*}) \\
            & \geq g(\tjcx{\TT^*}) - g(\tjx{\TT^*})
        \end{align*}

        So $h(\tjx{\TT^*},\tjcx{\TT^*}) \leq h(\tjx{\TT^\dagger},\tjcx{\TT^\dagger})$ holds.

        Now suppose that $\TT^*$ was obtained from $\TT$ by applying at least one operation of type $4$. Similarly, Let $\TT^\dagger$ be a dataset obtained from $\TT$ from the same operations used to obtain $\TT^*$ except that we replace one operation of type $4$ to an operation of type $2$. we have that $\tjx{\TT^*} = \tjx{\TT^\dagger} - 2$ and $\tjcx{\TT^*} = \tjcx{\TT^\dagger}$. Since $f$ is non-decreasing (see proof for Lemma \ref{lemma:local_Infogain0}):

        \begin{align*}
            f(\tjx{\TT^\dagger}) - f(\tjcx{\TT^\dagger}) & = f(\tjx{\TT^*}+2) - f(\tjcx{\TT^*}) \\
            & \geq f(\tjx{\TT^*}) - f(\tjcx{\TT^*})
        \end{align*}

        and as $g$ is non-increasing (see proof for Lemma \ref{lemma:local_Infogain0}):

        \begin{align*}
            g(\tjcx{\TT^\dagger}) - g(\tjx{\TT^\dagger}) & = g(\tjcx{\TT^*}) - g(\tjx{\TT^*}+2) \\
            & \geq g(\tjcx{\TT^*}) - g(\tjx{\TT^*})
        \end{align*}

        Finally, $h(\tjx{\TT^*},\tjcx{\TT^*}) \leq h(\tjx{\TT^\dagger},\tjcx{\TT^\dagger})$ holds for some $\TT^\dagger$ such that $d(\TT,\TT^\dagger) = t$.
    \end{proof}

\section{PrivateSQL's SQL query}
\label{ap:privateSQL}
The SQL query $\mathcal{Q}(u,v,c)$ over the table of nodes node(id), the table of all pairs of nodes node\_pair(id1, id2) and private table edge(a,b) is given as:

\definecolor{dkgreen}{rgb}{0,0.6,0}
\definecolor{ltgray}{rgb}{0.5,0.5,0.5}

\lstset{%
  	backgroundcolor=\color{white},
  	basicstyle=\footnotesize,
  	breakatwhitespace=false,
  	breaklines=true,
  	captionpos=b,
  	commentstyle=\color{dkgreen},
  	deletekeywords={...},
  	escapeinside={\%*}{*)},
  	extendedchars=true,
  	frame=single,
  	keepspaces=true,
  	keywordstyle=\color{blue},
  	language=SQL,
  	morekeywords={*,modify,MODIFY,...},
  	numbersep=15pt,
  	numberstyle=\tiny,
  	rulecolor=\color{ltgray},
  	showspaces=false,
  	showstringspaces=false, 
  	showtabs=false,
  	stepnumber=1,
  	tabsize=2,
  	title=\lstname
}

\begin{lstlisting}[language=sql]
SELECT COUNT(*)
FROM edge e1, edge e2,
	(SELECT np1.id1 AS idd1, np1.id2 as idd2, COALESCE (CNT_2,0) AS CNT_1
     FROM node_pair AS np1 LEFT OUTER JOIN 
        (SELECT np2.id1 AS id1, np2.id2 as 
                id2, COUNT (*) AS CNT_2
        FROM node_pair AS np2, edge e3
        WHERE e3.a = np2.id1 AND
              e3.b = np2.id2
        GROUP BY np2.id1, np2.id2) AS pair2
        ON np1.id1 = pair2.id1
        AND np1.id2 = pair2.id2) AS magic1,
    (SELECT e4.a AS m2id, 
            count(*) AS CNT_4
    FROM edge e4
    WHERE e4.b = c
    GROUP BY e4.a) AS magic2,
    (SELECT e5.a AS m3id, 
            count(*) AS CNT_5
    FROM edge e5
    WHERE e5.b = c
    GROUP BY e5.a) AS magic3,
    (SELECT node.id as m4id, 
            COALESCE(CNT_7, 0) as CNT_6
    FROM node left outer join 
        (SELECT e6.a AS m5id, 
                count(*) AS CNT_7
        FROM edge e6
        WHERE e6.b = c
        GROUP BY e6.a)
        AS magic5 on node.id=magic5.m5id
    WHERE magic5.CNT_7>0 OR
          node.id=c) AS magic4
WHERE e1.b = e2.a AND
      e1.a = u AND
      e2.b = v AND
      e1.a = idd1 AND
      e2.b = idd2 AND
      CNT_1 = 0 AND
      m2id = e1.a AND
      m3id = e2.b AND
      m4id = e2.a;
\end{lstlisting}
}


\end{document}